\documentclass[a4paper,10pt]{article}
\usepackage[T1]{fontenc}
\usepackage[utf8]{inputenc}
\usepackage[english]{babel}
\usepackage{amssymb,amsfonts,amsmath,amsthm}
\usepackage{mathtools}
\usepackage[hidelinks]{hyperref}
\usepackage{booktabs}
\usepackage{subcaption}
\usepackage{enumitem}
\usepackage{bbm}
\usepackage{pgfplots}

\usepackage[margin=3cm]{geometry}
\usepackage{fancyhdr}
\mathtoolsset{showonlyrefs}
\pgfplotsset{compat=newest}
\newlength\figureheight
\newlength\figurewidth
\usepgfplotslibrary{external}
\newtheorem{theorem}{Theorem}[section]
\newtheorem{lemma}[theorem]{Lemma}
\newtheorem{proposition}[theorem]{Proposition}
\newtheorem{corollary}[theorem]{Corollary}
\theoremstyle{definition}
\newtheorem{definition}[theorem]{Definition}
\newtheorem{assumption}[theorem]{Assumption}

\newtheorem{remark}[theorem]{Remark}

\numberwithin{equation}{section}
\numberwithin{table}{section}
\numberwithin{figure}{section}

\newcommand{\R}{\mathbb{R}}
\newcommand{\N}{\mathbb{N}}
\newcommand{\calN}{\mathcal{N}}
\newcommand{\calF}{\mathcal{F}}
\newcommand{\rmd}{\mathrm{d}}
\newcommand{\ddt}{\frac{\rmd}{\rmd t}}
\newcommand{\rme}{\mathrm{e}}

\newcommand{\gam}[2]{Q^{#1}_{#2}}
\newcommand{\muhat}[2]{m^{#1}_{#2}}
\newcommand{\tstrut}{\rule{0pt}{2.6ex}}

\newcommand{\transp}{\top}
\newcommand{\mystrut}{\rule{0pt}{7.5pt}}
\newcommand{\Lp}{\mathrm{L}\mystrut^p}

\DeclareMathOperator{\E}{\mathbb{E}}
\DeclareMathOperator{\var}{var}
\DeclareMathOperator{\cov}{cov}
\DeclareMathOperator{\tr}{tr}

\DeclareMathOperator{\lmin}{\lambda_{\min}}

\begin{document}

\title{Diffusion Approximations for Expert Opinions in a Financial Market with Gaussian Drift}

\author{J\"{o}rn Sass\thanks{Department of Mathematics, Technische Universit\"{a}t Kaiserslautern, \newline P.O. Box 3049, 67653 Kaiserslautern, Germany; \texttt{sass@mathematik.uni-kl.de}}, Dorothee Westphal\footnote{Department of Mathematics, Technische Universit\"{a}t Kaiserslautern, \newline P.O. Box 3049, 67653 Kaiserslautern, Germany; \texttt{westphal@mathematik.uni-kl.de}} \,and Ralf Wunderlich\footnote{Institute of Mathematics, Brandenburg University of Technology Cottbus-Senftenberg, \newline P.O. Box 101344, 03013 Cottbus, Germany; \texttt{ralf.wunderlich@b-tu.de}}}

\date{March~5, 2020}

\maketitle

\begin{abstract}
	This paper investigates a financial market where returns depend on an unobservable Gaussian drift process. While the observation of returns yields information about the underlying drift, we also incorporate discrete-time expert opinions as an external source of information.
	
	For estimating the hidden drift it is crucial to consider the conditional distribution of the drift given the available observations, the so-called filter. For an investor observing both the return process and the discrete-time expert opinions, we investigate in detail the asymptotic behavior of the filter as the frequency of the arrival of expert opinions tends to infinity. In our setting, a higher frequency of expert opinions comes at the cost of accuracy, meaning that as the frequency of expert opinions increases, the variance of expert opinions becomes larger. We consider a model where information dates are deterministic and equidistant and another model where the information dates arrive randomly as the jump times of a Poisson process. In both cases we derive limit theorems stating that the information obtained from observing the discrete-time expert opinions is asymptotically the same as that from observing a certain diffusion process which can be interpreted as a continuous-time expert.
	
	We use our limit theorems to derive so-called diffusion approximations of the filter for high-frequency discrete-time expert opinions. These diffusion approximations are extremely helpful for deriving simplified approximate solutions of utility maximization problems.
	
	\bigskip
	
	
	\smallskip
	
\end{abstract}

\section{Introduction}

Optimal trading strategies in dynamic portfolio optimization problems depend crucially on the drift of the underlying asset price processes. However, drift parameters are notoriously difficult to estimate from historical asset price data. Drift processes tend to fluctuate randomly over time and even if they were constant, long time series would be needed to estimate this parameter with a satisfactory degree of precision. Typically, drift effects are overshadowed by volatility. For these reasons, practitioners also incorporate external sources of information such as news, company reports, ratings or their own intuitive views when determining optimal portfolio strategies. These outside sources of information are called \emph{expert opinions}. In the context of the classical one-period Markowitz model this leads to the well-known Black--Litterman approach, where return predictions are improved by means of views formulated by securities analysts, see Black and Litterman~\cite{black_litterman_1992}.

In this paper we consider a financial market where returns depend on an underlying drift process which is unobservable due to additional noise coming from a Brownian motion. The general setting has already been studied in Gabih et al.~\cite{gabih_kondakji_sass_wunderlich_2014} for a market with only one risky asset and in Sass et al.~\cite{sass_westphal_wunderlich_2017} for markets with an arbitrary number of stocks. The ability to choose good trading strategies depends on how well the unobserved drift can be estimated.
For estimating the hidden drift we consider the conditional distribution of the drift given the available observations, the so-called filter. The best estimate for the hidden drift process in a mean-square sense is the conditional mean of the drift given the available information. A measure for the goodness of this estimator is its conditional covariance matrix. In our setting, the filter is completely characterized by conditional mean and conditional covariance matrix since we deal with Gaussian distributions.

For investors who observe only the return process, the filter is the classical Kalman filter, see for example Liptser and Shiryaev~\cite{liptser_shiryaev_1974}. An additional source of information is provided by expert opinions which we model as unbiased drift estimates arriving at discrete points in time. Investors who, in addition to observing the return process, have access to these expert opinions update their current drift estimates at each arrival time. These updates decrease the conditional covariance, hence they yield better estimates. This can be seen as a continuous-time version of the above mentioned static Black--Litterman approach.

We investigate in detail an investor who observes both the return process and the discrete-time expert opinions and study the asymptotic behavior of the filter when the frequency of the arrival of  expert opinions tends to infinity. Sass et al.~\cite{sass_westphal_wunderlich_2017} and Gabih et al.~\cite{gabih_kondakji_wunderlich_2018} already addressed expert opinions which are independent of the arrival frequency and which have some minimal level of accuracy characterized by bounded covariances. In that setting, the conditional covariance of the drift estimate goes to zero as the arrival frequency goes to infinity. This implies that the conditional mean converges to the true drift process, i.e.\ in the limit investors have full information about the drift.
Here, we study a different situation in which a higher frequency of expert opinions is only available at the cost of accuracy of the single expert opinions. In other words, as the frequency of expert opinions increases, the variance of expert opinions becomes larger. On the one hand, this assumption ensures that it is not possible for investors to gain arbitrarily much information in a fixed time interval. On the other hand, it enables us to derive a certain asymptotic behavior that yields a reasonable approximation of the filter for the investor who observes a certain, fixed number of discrete-time expert opinions.
We consider two different situations, one with deterministic equidistant information dates and one with information dates that arrive randomly as the jump times of a Poisson process.
For properly scaled variance of expert opinions that grows linearly with the arrival frequency we prove $\Lp$-convergence of the conditional mean and conditional covariance matrices as the frequency of information dates goes to infinity. Our limit theorems imply that the information obtained from observing the discrete-time expert opinions is asymptotically the same as that from observing a certain diffusion process having the same drift as the return process. That process can be interpreted as a \emph{continuous-time expert} who permanently delivers noisy information about the drift.

Our limit theorems allow us to derive approximations of the filter for high-frequency discrete-time expert opinions which we call \emph{diffusion approximations}. These are useful since the limiting filter is easy to compute whereas the updates for the discrete-time expert opinions lead to a computationally involved filter. This is extremely helpful for deriving simplified approximate solutions of utility maximization problems. We apply our diffusion approximations to a portfolio optimization problem with logarithmic utility. Numerical simulations show that the approximation is very accurate even for a small number of expert opinions. Our rigorous $\Lp$-convergence results of the filters however also allow to derive convergence of the value function in the more complicated problem with power utility, see Remark~\ref{rem:power_utility}.

\medskip

The idea of a continuous-time expert is in line with Davis and Lleo~\cite{davis_lleo_2013} who study an approach called ``Black--Litterman in Continuous Time'' (BLCT). Our results show how the BLCT model can be obtained as a limit of BLCT models with discrete-time experts. First papers addressing BLCT are Frey et al.~\cite{frey_gabih_wunderlich_2012,frey_gabih_wunderlich_2014} who consider an HMM for the drift and expert opinions arriving at the jump times of a Poisson process.

Convergence of the discrete-time Kalman filter to the continuous-time equivalent has been addressed in the literature, e.g.\ by Salgado et al.~\cite{salgado_middleton_goodwin_1988} or Aalto~\cite{aalto_2016} for the case of deterministic information dates. Our results for that case however do not follow directly from these convergence results. The reason is that in our case a suitable continuous-time expert has to be constructed first. The discrete-time expert opinions are then not simply a discretization of the continuous-time expert. We assume that they are noisy observations of the true drift process where the noise term is correlated with the Brownian motion in the diffusion that forms the continuous-time expert. Contrary to~\cite{aalto_2016, salgado_middleton_goodwin_1988} we also obtain convergence results for the case where the discrete expert opinions arrive at random time points rather than on an equidistant time grid.

Coquet et al.~\cite{coquet_memin_slominski_2001} consider weak convergence of filtrations which allows to prove convergence of conditional expectations in a quite general setting. However, their results do not directly apply to our situation since the approximating sequence of filtrations in our case is not included in the limit filtration.

In the literature, diffusion approximations also appear in other contexts. They are well-known in operations research and actuarial mathematics. The basic idea is to replace a complicated stochastic process by an appropriate diffusion process which is analytically more tractable than the original process. The approach is comparable with the normal approximation of sums of random variables by the Central Limit Theorem. When looking at these sums as stochastic processes or random walks the well-known Donsker Theorem leads to an approximation by a Brownian motion.

For an introduction to diffusion approximations based on the theory of weak convergence and applications to queueing systems in heavy traffic we refer to the survey article by Glynn~\cite{glynn_1990}. In risk theory the application of diffusion approximations for computing ruin probabilities goes back to Iglehart~\cite{iglehart_1969}. We also refer to Grandell~\cite[Sec.~1.2]{grandell_1991}, Schmidli~\cite[Sec.~5.10 and 6.5]{schmidli_2017} and Asmussen and Albrecher~\cite[Sec.~V.5]{asmussen_albrecher_2010} as well as the references therein.
Starting point is the classical Cram\'{e}r--Lundberg model where the cumulated claim sizes and finally the surplus of an insurance company are modeled by a compound Poisson process. For a high intensity of the claim arrivals and small claim sizes the latter can be approximated by a Brownian motion with drift. This results from the corresponding weak convergence of the properly scaled compound Poisson processes to a Brownian motion as the intensity tends to infinity.
However, these classical results for compound Poisson processes cannot be applied directly to our problem. Here, the jumps of the filter processes do not constitute a sequence of i.i.d.\ random variables as in the compound Poisson case. Due to the Bayesian updating of the filter at the information dates the jump size distribution depends on the value of the filter at that time. This requires special techniques for proving limit theorems from which the diffusion approximations can be derived. To the best of our knowledge these techniques constitute a new contribution to the literature.

\medskip

The paper is organized as follows. In Section~\ref{sec:market_model_and_filtering} we introduce the model for our financial market including expert opinions and define different information regimes for investors with different sources of information. For each of those information regimes, we state the dynamics of the corresponding conditional mean and conditional covariance matrix.
Section~\ref{sec:diffusion_approximation_of_filters_for_deterministic_information_dates} investigates the situation where the discrete-time expert opinions arrive at deterministic equidistant time points. For an investor observing returns and discrete-time expert opinions we show convergence of the corresponding conditional mean and conditional covariance matrix to those of an investor observing the returns and the continuous-time expert.
In Section~\ref{sec:diffusion_approximation_of_filters_for_random_information_dates} we prove analogous results for the situation where the time points at which expert opinions arrive are not deterministic time points but jump times of a standard Poisson process, i.e.\ with exponentially distributed waiting times between information dates. For the conditional mean we can then use a representation involving a Poisson random measure. When letting the intensity of the Poisson process go to infinity, we prove convergence to the same limiting filter as in the case with deterministic information dates.
Section~\ref{sec:application_utility_maximization} provides an application of the convergence results to a utility maximization problem. For investors who maximize expected logarithmic utility of terminal wealth the optimal trading strategy depends on the conditional mean of the drift and the corresponding optimal terminal wealth is a functional of the conditional covariance matrices. That is why the convergence results from Sections~\ref{sec:diffusion_approximation_of_filters_for_deterministic_information_dates} and \ref{sec:diffusion_approximation_of_filters_for_random_information_dates} carry over to convergence of the corresponding value functions.
Section~\ref{sec:numerical_example} provides simulations and numerical calculations to illustrate our theoretical results.
In Appendix~\ref{app:technical_lemmas} we collect some auxiliary results needed for the proofs of our main theorems. Appendix~\ref{app:long_proofs_fixed_times} gives the proofs of Theorems~\ref{thm:q_C_n_goes_to_q_D} and~\ref{thm:muhat_C_n_goes_to_muhat_D} and Appendix~\ref{app:long_proofs_random_times} those of Theorems~\ref{thm:q_C_lambda_goes_to_q_D} and~\ref{thm:muhat_C_lambda_goes_to_muhat_D}.

\paragraph{Notation:} Throughout this paper, we use the notation $I_d$ for the identity matrix in $\R^{d\times d}$. For a symmetric and positive-semidefinite matrix $A\in\R^{d\times d}$ we call a symmetric and positive-semidefinite matrix $B\in\R^{d\times d}$ the \emph{square root} of $A$ if $B^2=A$. The square root is unique and will be denoted by $A^{\frac{1}{2}}$. Unless stated otherwise, whenever $A$ is a matrix, $\lVert A\rVert$ denotes the spectral norm of $A$.

\section{Market Model and Filtering}\label{sec:market_model_and_filtering}

\subsection{Financial Market Model}

We consider a financial market with one risk-free and multiple risky assets. The basic model is the same as in Sass et al.~\cite{sass_westphal_wunderlich_2017}. In the following, we denote by $T>0$ a finite investment horizon and fix a filtered probability space $(\Omega,\mathcal{G},\mathbb{G},\mathbb{P})$ where the filtration $\mathbb{G}=(\mathcal{G}_t)_{t\in[0,T]}$ satisfies the usual conditions. All processes are assumed to be $\mathbb{G}$-adapted.
The market consists of one risk-free bond with constant deterministic interest rate $r\in\R$, and $d$ risky assets such that the $d$-dimensional return process follows the stochastic differential equation
\[ \rmd R_t = \mu_t\,\rmd t+\sigma_R\,\rmd W^R_t. \]
Here $W^R=(W^R_t)_{t\in[0,T]}$ is an $m$-dimensional Brownian motion with $m\geq d$ and we assume that $\sigma_R\in\R^{d\times m}$ has full rank. The drift $\mu$ is an Ornstein--Uhlenbeck process and follows the dynamics
\[ \rmd \mu_t = \alpha (\delta - \mu_t)\,\rmd t + \beta\,\rmd B_t, \]
where $\alpha$ and $\beta\in\R^{d\times d}$, $\delta\in\R^d$ and $B=(B_t)_{t\in [0,T]}$ is a $d$-dimensional Brownian motion independent of $W^R$. We assume that $\alpha$ is a symmetric and positive-definite matrix to ensure that expectation and covariance of the drift process stay bounded and the drift process becomes asymptotically stationary. This is reasonable from an economic point of view. The initial drift $\mu_0$ is multivariate normally distributed, $\mu_0\sim\mathcal{N}(m_0,\Sigma_0)$, for some $m_0\in\R^d$ and some $\Sigma_0\in\R^{d\times d}$ which is symmetric and positive semidefinite. We assume that $\mu_0$ is independent of $B$ and $W^R$. We denote $m_t:=\E[\mu_t]$ and $\Sigma_t:=\cov(\mu_t)$.

Investors in this market know the model parameters and are able to observe the return process $R$. They neither observe the underlying drift process $\mu$ nor the Brownian motion $W^R$. However, information about $\mu$ can be drawn from observing $R$. Additionally, we include expert opinions in our model. These expert opinions arrive at discrete time points and give an unbiased estimate of the state of the drift at that time point. Let $(T_k)_{k\in I}$ be an increasing sequence with values in $(0,T]$, where we allow for index sets $I=\N$ or $I=\{1,\dots,N\}$ for some $N\in\N$. The $T_k$, $k\in I$, are the time points at which expert opinions arrive. For the sake of convenience we also write $T_0=0$ although there is not necessarily an expert opinion arriving at time zero.

The expert view at time $T_k$ is modelled as an $\R^d$-valued random vector
\[ Z_k = \mu_{T_k}+(\Gamma_k)^{\frac12}\varepsilon_k, \]
where the matrix $\Gamma_k\in\R^{d\times d}$ is symmetric and positive definite and $\varepsilon_k$ is multivariate $\mathcal{N}(0,I_d)$-distributed. We assume that the sequence of $\varepsilon_k$ is independent and also that it is independent of both $\mu_0$ and the Brownian motions $B$ and $W^R$.
Note that, given $\mu_{T_k}$, the expert opinion $Z_k$ is multivariate $\mathcal{N}(\mu_{T_k},\Gamma_k)$-distributed. That means that the expert view at time $T_k$ gives an unbiased estimate of the state of the drift at that time. The matrix $\Gamma_k$ reflects the reliability of the expert.

Note that the time points $T_k$ do not need to be deterministic. However, we impose the additional assumption that the sequence $(T_k)_{k\in I}$ is independent of the $(\varepsilon_k)_{k\in I}$ and also of the Brownian motions in the market and of $\mu_0$. This essentially says that the timing of information dates carries no additional information about the drift $\mu$. Nevertheless, information on the sequence $(T_k)_{k\in I}$ may be important for optimal portfolio decisions. In the next sections we consider on the one hand the situation with deterministic information dates and on the other hand a case where information dates are the jump times of a Poisson process.

It is possible to allow relative expert views in the sense that an expert may give an estimate for the difference in drift of two stocks instead of absolute views. See Sch\"ottle et al.~\cite{schoettle_werner_zagst_2010} for how to switch between these two models for expert opinions by means of a pick matrix.

Our main results in Sections~\ref{sec:diffusion_approximation_of_filters_for_deterministic_information_dates} and \ref{sec:diffusion_approximation_of_filters_for_random_information_dates} address the question how to obtain rigorous convergence results when the number of information dates increases. We will show that, for certain sequences of expert opinions, the information drawn from these expert opinions, expressed by the filter, is for a large number of expert opinions essentially the same as the information one gets from observing yet another diffusion process. This diffusion process can then be interpreted as an expert who gives a continuous-time estimation about the state of the drift. Let this estimate be given by the diffusion process
\begin{equation}\label{eq:continuous_expert_J}
	\rmd J_t =\mu_t\,\rmd t +\sigma_J\,\rmd W^J_t,
\end{equation}
where $W^J$ is an $l$-dimensional Brownian motion with $l\geq d$ that is independent of all other Brownian motions in the model and of the information dates $T_k$. The matrix $\sigma_J\in\R^{d\times l}$ has full rank equal to $d$.

\subsection{Filtering for Different Information Regimes}

For an investor in the financial market defined above, the ability to choose good trading strategies is based heavily on which information is available about the unknown drift process $\mu$. To be able to assess the value of information coming from observing expert opinions, we consider various types of investors with different sources of information. This follows the approach in Gabih et al.~\cite{gabih_kondakji_sass_wunderlich_2014} and in Sass et al.~\cite{sass_westphal_wunderlich_2017}. The information available to an investor can be described by the investor filtration $\mathbb{F}^H=(\mathcal{F}^H_t)_{t\in[0,T]}$ where $H$ serves as a placeholder for the various information regimes. We work with filtrations that are augmented by $\calN_{\mathbb{P}}$, the set of null sets under measure $\mathbb{P}$. We consider the cases
\begingroup
\allowdisplaybreaks
\begin{alignat*}{3}
	&\mathbb{F}^R && =(\mathcal{F}^R_t)_{t\in[0,T]} && \text{ where } \mathcal{F}^R_t=\sigma((R_s)_{s\in[0,t]})\vee\sigma(\calN_{\mathbb{P}}), \\
	&\mathbb{F}^Z && =(\mathcal{F}^Z_t)_{t\in[0,T]} && \text{ where } \mathcal{F}^Z_t=\sigma((R_s)_{s\in[0,t]})\vee\sigma((T_k,Z_k)_{T_k\leq t})\vee\sigma(\calN_{\mathbb{P}}), \\
	&\mathbb{F}^J && =(\mathcal{F}^J_t)_{t\in[0,T]} && \text{ where } \mathcal{F}^J_t=\sigma((R_s)_{s\in[0,t]})\vee\sigma((J_s)_{s\in[0,t]})\vee\sigma(\calN_{\mathbb{P}}),\\
	&\mathbb{F}^F && =(\mathcal{F}^F_t)_{t\in[0,T]} && \text{ where } \mathcal{F}^F_t=\sigma((R_s)_{s\in[0,t]})\vee\sigma((\mu_s)_{s\in[0,t]})\vee\sigma(\calN_{\mathbb{P}}).
\end{alignat*}
\endgroup
When speaking of the $H$-investor we mean the investor with investor filtration $\mathbb{F}^H=(\mathcal{F}^H_t)_{t\in[0,T]}$, $H\in\{R,Z,J,F\}$. Note that the $R$-investor observes only the return process $R$, the $Z$-investor combines the information from observing the return process and the discrete-time expert opinions $Z_k$, and the $J$-investor observes the return process and the continuous-time expert $J$. The $F$-investor has full information about the drift in the sense that she can observe the drift process directly. This case is included as a benchmark.

As already mentioned, the investors in our financial market make trading decisions based on available information about the drift process $\mu$. Only the $F$-investor can observe the drift, the other investors have to estimate it. The conditional distribution of the drift under partial information is called the \emph{filter}. In the mean-square sense, an optimal estimator for the drift at time $t$ given the available information is then the \emph{conditional mean} $\muhat{H}{t}:=\E[\mu_t\,|\,\mathcal{F}^H_t]$. How close this estimator is to the true state of the drift can be assessed by looking at the corresponding \emph{conditional covariance matrix}
\[ \gam{H}{t} := \E\bigl[(\mu_t-\muhat{H}{t})(\mu_t-\muhat{H}{t})^\transp\,\big|\,\mathcal{F}^H_t\bigr]. \]
Note that since we deal with Gaussian distributions here, the filter is also Gaussian and completely characterized by conditional mean and conditional covariance matrix.
In the next sections we investigate the behavior of the filter for a $Z$-investor with access to an increasing number of expert opinions. For this purpose, we state in the following the dynamics of the filters for the various investors defined above.
For the $R$-investor, we are in the setting of the well-known Kalman filter.

\begin{lemma}\label{lem:filter_R_dynamics}
	The filter of the $R$-investor is Gaussian. The conditional mean $\muhat{R}{}$ follows the dynamics
	\[ \rmd\muhat{R}{t} = \alpha(\delta-\muhat{R}{t})\,\rmd t + \gam{R}{t}(\sigma_R\sigma_R^\transp)^{-1}(\rmd R_t-\muhat{R}{t}\,\rmd t), \]
	where $\gam{R}{}$ is the solution of the ordinary Riccati differential equation
	\[ \frac{\rmd}{\rmd t} \gam{R}{t} = -\alpha \gam{R}{t} -\gam{R}{t}\alpha + \beta\beta^\transp - \gam{R}{t}(\sigma_R\sigma_R^\transp)^{-1}\gam{R}{t}. \]
	The initial values are $\muhat{R}{0}=m_0$ and $\gam{R}{0}=\Sigma_0$.
\end{lemma}

This lemma follows directly from the Kalman filter theory, see for example Theorem~10.3 of Liptser and Shiryaev~\cite{liptser_shiryaev_1974}. Note that $\gam{R}{t}$ follows an ordinary differential equation, called Riccati equation, and is hence deterministic.

Next, we consider the $J$-investor who observes the diffusion processes $R$ and $J$.

\begin{lemma}\label{lem:filter_D_dynamics}
	The filter of the $J$-investor is Gaussian. The conditional mean $\muhat{J}{}$ follows the dynamics
	\[ \rmd \muhat{J}{t} = \alpha(\delta-\muhat{J}{t})\,\rmd t + \gam{J}{t} \begin{pmatrix} (\sigma_R\sigma_R^\transp)^{-1} \\[1mm] (\sigma_J\sigma_J^\transp)^{-1} \end{pmatrix}^\transp \begin{pmatrix}\rmd R_t- \muhat{J}{t}\,\rmd t \\[1mm] \rmd J_t-\muhat{J}{t}\,\rmd t \end{pmatrix}, \]
	where $\gam{J}{}$ is the solution of the ordinary Riccati differential equation
	\begin{equation}\label{eq:Riccati_ode_D}
		\ddt \gam{J}{t} = -\alpha\gam{J}{t}-\gam{J}{t}\alpha+\beta\beta^\transp-\gam{J}{t}\bigl((\sigma_R\sigma_R^\transp)^{-1}+(\sigma_J\sigma_J^\transp)^{-1}\bigr) \gam{J}{t}
	\end{equation}
	with $\muhat{J}{0}=m_0$ and $\gam{J}{0}=\Sigma_0$.
\end{lemma}

\begin{proof}
	First, note that the matrix $(\sigma_R\sigma_R^\transp)^{-1}+(\sigma_J\sigma_J^\transp)^{-1}\in\R^{d\times d}$ is symmetric and positive definite, and hence nonsingular. The distribution of the filter as well as the dynamics of $\muhat{J}{}$ and $\gam{J}{}$ then follow immediately from the Kalman filter theory, see again Theorem~10.3 in Liptser and Shiryaev~\cite{liptser_shiryaev_1974}.
\end{proof}

Note that, just like in the case for the $R$-investor, the conditional covariance matrix is deterministic.

Let us now come to the $Z$-investor. Recall that this investor observes the return process $R$ continuously in time and at (possibly random) information dates $T_k$ the expert opinions $Z_k$. We state the dynamics of $\muhat{Z}{}$ and $\gam{Z}{}$ in the following lemma.

\begin{lemma}\label{lem:filter_C_dynamics}
	Given a sequence of information dates $T_k$, the filter of the $Z$-investor is Gaussian. The dynamics of the conditional mean and conditional covariance matrix are given as follows:
	\begin{enumerate}[label=(\roman*)]
		\item Between the information dates $T_k$ and $T_{k+1}$, $k\in\N_0$, it holds
		\[ \rmd \muhat{Z}{t} = \alpha (\delta-\muhat{Z}{t})\,\rmd t +\gam{Z}{t} (\sigma_R\sigma_R^\transp)^{-1}(\rmd R_t-\muhat{Z}{t} \,\rmd t) \]
		for $t\in[T_k,T_{k+1})$, where $\gam{Z}{}$ follows the ordinary Riccati differential equation
		\[ \frac{\rmd}{\rmd t} \gam{Z}{t} = -\alpha\gam{Z}{t} -\gam{Z}{t}\alpha +\beta\beta^\transp - \gam{Z}{t}(\sigma_R\sigma_R^\transp)^{-1}\gam{Z}{t} \]
		for $t\in[T_k,T_{k+1})$.
		The initial values are $\muhat{Z}{T_k}$ and $\gam{Z}{T_k}$, respectively, with $\muhat{Z}{0}=m_0$ and $\gam{Z}{0}=\Sigma_0$.
		
		\item The update formulas at information dates $T_k$, $k\in\N$, are
		\begin{equation*}
			\begin{aligned}
				\muhat{Z}{T_k} &= \rho_k(\gam{Z}{T_k-}) \muhat{Z}{T_k-}+\bigl(I_d-\rho_k(\gam{Z}{T_k-})\bigr)Z_k \\
				&= \muhat{Z}{T_k-}+\bigl(I_d-\rho_k(\gam{Z}{T_k-})\bigr)\bigl(Z_k-\muhat{Z}{T_k-}\bigr)
			\end{aligned}
		\end{equation*}
		and
		\begin{equation*}
			\begin{aligned}
				\gam{Z}{T_k} &= \rho_k(\gam{Z}{T_k-})\gam{Z}{T_k-} \\
				&=\gam{Z}{T_k-}+\bigl(\rho_k(\gam{Z}{T_k-})-I_d\bigr)\gam{Z}{T_k-},
			\end{aligned}
		\end{equation*}
		where $\rho_k(\gam{}{})=\Gamma_k(\gam{}{}+\Gamma_k)^{-1}$.
	\end{enumerate}
\end{lemma}

\begin{proof}
	For deterministic time points $T_k$, the above lemma is Lemma~2.3 of Sass et al.~\cite{sass_westphal_wunderlich_2017} where a detailed proof is given. For the more general case where the $T_k$ need not be deterministic, recall that we have made the assumption that the sequence $(T_k)_{k\in I}$ is independent of the other random variables in the market. In particular, $(T_k)_{k\in I}$ and the drift process $\mu$ are independent. Because of that, the dynamics of the conditional mean and conditional covariance matrix are the same as for deterministic information dates and we get the same update formulas, the only difference being that the update times might now be non-deterministic.
	
	The Gaussian distribution of the filter between information dates follows as in the previous lemmas from the Kalman filter theory. The updates at information dates can be seen as a degenerate discrete-time Kalman filter. Hence, the distribution of the filter at information dates remains Gaussian after the Bayesian update.
\end{proof}

Note that the dynamics of $\muhat{Z}{}$ and $\gam{Z}{}$ between information dates are the same as for the $R$-investor, see Lemma~\ref{lem:filter_R_dynamics}. The values at an information date $T_k$ are obtained from a Bayesian update.
If we have non-deterministic information dates $T_k$ then in contrast to both the $R$-investor and the $J$-investor, the conditional covariance matrices $\gam{Z}{}$ of the $Z$-investor are non-deterministic since updates take place at random times.

In the proofs of our main results we repeatedly need to find upper bounds for various expressions that involve the conditional covariance matrices $\gam{J}{}$ or $\gam{Z}{}$. A key tool is boundedness of these matrices. Here, it is useful to consider a partial ordering of symmetric matrices. For symmetric matrices $A,B\in\R^{d\times d}$ we write $A\preceq B$ if $B-A$ is positive semidefinite. Note that $A\preceq B$ in particular implies that $\lVert A\rVert\leq\lVert B\rVert$.

\begin{lemma}\label{lem:boundedness_of_covariances}
	For any sequence $(T_k,Z_k)_{k\in I}$ we have $\gam{Z}{t}\preceq\gam{R}{t}$ and $\gam{J}{t}\preceq\gam{R}{t}$ for all $t\geq 0$. In particular, there exists a constant $C_{\gam{}{}}>0$ such that
	\[ \lVert\gam{Z}{t}\rVert\leq C_{\gam{}{}} \quad \text{and} \quad \lVert\gam{J}{t}\rVert\leq C_{\gam{}{}} \]
	for all $t\in[0,T]$.
\end{lemma}

\begin{proof}
	Let $(T_k,Z_k)_{k\in I}$ be any sequence of expert opinions and $(\gam{Z}{t})_{t\in[0,T]}$ the conditional covariance matrices of the corresponding filter.
	Every update decreases the covariance in the sense that $\gam{Z}{T_k}\preceq\gam{Z}{T_k-}$, see Proposition~2.2 in Sass et al.~\cite{sass_westphal_wunderlich_2017}. Also, if $(P_t)_{t\geq 0}$ and $(\tilde{P}_t)_{t\geq 0}$ are solutions of the same Riccati differential equation, where the initial values fulfill $P_0\preceq\tilde{P}_0$, then $P_t\preceq\tilde{P}_t$ for all $t\geq 0$, see for example Theorem~10 in Ku\u{c}era~\cite{kucera_1973}. Inductively, we can deduce that in our setting $\gam{Z}{t}\preceq\gam{R}{t}$ for all $t\geq 0$.
	Also, one can show that $\gam{J}{t}\preceq\gam{R}{t}$ for all $t\geq 0$ in analogy to the proof of Proposition~3.1 in Sass et al.~\cite{sass_westphal_wunderlich_2017}. The key idea for the proof is to use the fact that $\mathcal{F}^R_t\subseteq\mathcal{F}^J_t$ for all $t\geq 0$.
	
	By Theorem~4.1 in Sass et al.~\cite{sass_westphal_wunderlich_2017} there exists a positive-semidefinite matrix $\gam{R}{\infty}$ such that
	\[ \lim_{t\to\infty} \gam{R}{t} = \gam{R}{\infty}. \]
	Hence, $\lVert\gam{R}{t}\rVert$ is bounded by some constant $C_{\gam{}{}}>0$, and the claim follows.
\end{proof}

\section{Diffusion Approximation of Filters for Deterministic Information Dates}\label{sec:diffusion_approximation_of_filters_for_deterministic_information_dates}

In this section we investigate the asymptotic behavior of the filters for a $Z$-investor when the frequency of expert opinion arrivals goes to infinity.
We consider first the case for deterministic and equidistant information dates. Therefore, let $n\in\N$ and $\Delta_n=\frac{T}{n}$. Now assume that $T_k=t_k$ for every $k=1,\dots,n$, where $(t_k)_{k=1,\dots,n}$ is the sequence of deterministic time points $t_k=k\Delta_n$. So there are $n$ expert opinions that arrive equidistantly in the time interval $[0,T]$, the distance between two information dates being $\Delta_n$.

In the following, we deduce convergence results for both the conditional means and the conditional covariance matrices of the $Z$-investor when sending $n$ to infinity.
Note that convergence of discrete-time filters is addressed in earlier papers, e.g.\ by Salgado et al.~\cite{salgado_middleton_goodwin_1988} or Aalto~\cite{aalto_2016}. There, the authors show convergence of the discrete-time Kalman filter to the continuous-time equivalent. In Aalto~\cite{aalto_2016} the discrete-time filter is based on discrete-time observations of the continuous-time observation process whereas in Salgado et al.~\cite{salgado_middleton_goodwin_1988} the authors approximate both the continuous-time signal and observation by discrete-time processes. Neither of these assumptions match our model for the discrete-time expert opinions which is why we need to prove convergence in the following.

We use an additional superscript $n$ to emphasize dependence on the number of expert opinions, writing for example $(\gam{Z,n}{t})_{t\in[0,T]}$ for the conditional covariance matrix of the filter corresponding to these $n$ expert opinions.
In Sass et al.~\cite{sass_westphal_wunderlich_2017} a convergence result is proven for the case where the expert opinions are of the form
\begin{equation}\label{eq:form_of_expert_opinions_fixed}
	Z_k^{(n)} = \mu_{t^{(n)}_k}+(\Gamma_k^{(n)})^{\frac{1}{2}}\varepsilon_k^{(n)}
\end{equation}
with expert's covariances $\Gamma_k^{(n)}$ that are bounded for all $n\in\N$ and $k=1,\dots,n$, see Theorem~3.1 in Sass et al.~\cite{sass_westphal_wunderlich_2017}. There it is shown that under the assumption of bounded expert's covariances it holds
\[ \lim_{n\to\infty} \lVert\gam{Z,n}{t}\rVert=0 \]
for any $t\in(0,T]$. Since $\gam{Z,n}{t}$ is a measure for the goodness of the estimator $\muhat{Z,n}{t}$, this means that the conditional mean of the $Z$-investor becomes an arbitrarily good estimator for the true state of the drift $\mu_t$. One can easily deduce that
\[ \lim_{n\to\infty} \E\Bigl[\bigl\lVert \muhat{Z,n}{t}-\mu_t\bigr\rVert^2\Bigr]=0 \]
for any $t\in(0,T]$. Hence, the $Z$-investor essentially approximates the fully informed $F$-investor.

This result heavily relies on the assumption that the expert covariances $\Gamma_k^{(n)}$ are all bounded, meaning that there is some minimal level of reliability of the experts. Here, we study a different situation where more frequent expert opinions are only available at the cost of accuracy. In other words, we assume that, as $\Delta_n$ goes to zero, the variance of expert opinions $Z_k^{(n)}$ increases. This is done for the purpose of approximating $\muhat{Z,n}{}$ and $\gam{Z,n}{}$ for large $n\in\N$ and large $\Gamma_k^{(n)}$.
In the following we assume for the sake of simplicity that $\Gamma_k^{(n)}=\Gamma^{(n)}$ is not time-dependent. We then show that for properly scaled $\Gamma^{(n)}$ which grows linearly in $n$, the information obtained from observing the discrete-time expert opinions is asymptotically the same as that from observing another diffusion process. This will be the diffusion $J$ already defined in \eqref{eq:continuous_expert_J}.

\begin{assumption}\label{ass:deterministic_time_points}
	Let $(T_k^{(n)})_{k=1,\dots,n}=(t_k^{(n)})_{k=1,\dots,n}$ where $t_k^{(n)}=k\Delta_n$ for $k=1,\dots,n$. Furthermore, let the experts' covariance matrices be given by
	\[ \Gamma_k^{(n)}=\Gamma^{(n)}=\frac{1}{\Delta_n}\sigma_J\sigma_J^\transp \]
	for $k=1,\dots,n$. Further, we assume that in~\eqref{eq:form_of_expert_opinions_fixed} the $\calN(0,I_d)$-distributed random variables $\varepsilon_k^{(n)}$ are linked with the Brownian motion $W^J$ from~\eqref{eq:continuous_expert_J} via $\varepsilon_k^{(n)}=\frac{1}{\sqrt{\Delta_n}}\int_{t_k^{(n)}}^{t_{k+1}^{(n)}}\rmd W^J_s$, so that the expert opinions are given as
	\begin{equation}\label{eq:expert_opinions_for_fixed_information_dates}
		Z^{(n)}_k = \mu_{t^{(n)}_k}+\frac{1}{\Delta_n}\sigma_J\int_{t^{(n)}_k}^{t^{(n)}_{k+1}} \rmd W^J_s
	\end{equation}
	for $k=1,\dots,n$.
\end{assumption}

Recall that the matrix $\sigma_J\in\R^{d\times l}$ is exactly the volatility of the diffusion process $J$ with the dynamics
\[ \rmd J_t =\mu_t\,\rmd t +\sigma_J\,\rmd W^J_t, \]
and that $\sigma_J$ has full rank. With $Z^{(n)}_k$ as defined above the discrete-time expert opinions and the continuous-time expert $J$ are obviously correlated. In fact, it holds
\[ Z^{(n)}_k \approx \frac{1}{\Delta_n}\int_{t^{(n)}_k}^{t^{(n)}_{k+1}} \rmd J_s = \frac{1}{\Delta_n}\Bigl(J_{t^{(n)}_{k+1}}-J_{t^{(n)}_k}\Bigr). \]
Further, one can easily show  by using Donsker's Theorem that the piecewise constant process $(\widetilde{J}_t)_{t\in[0,T]}$, defined by
\[ \widetilde{J}_t := \Delta_n\sum_{k=1}^{\lfloor t/\Delta_n\rfloor} Z_k^{(n)} \]
for all $t\in[0,T]$, converges in distribution to $J_t$ as $n$ goes to infinity. For our main convergence results that are given in the following, we however require stronger notions of convergence.

The following theorem now states uniform convergence of $\gam{Z,n}{t}$ to $\gam{J}{t}$ on $[0,T]$ for $n$ going to infinity.

\begin{theorem}\label{thm:q_C_n_goes_to_q_D}
	Under Assumption~\ref{ass:deterministic_time_points} there exists a constant $K_Q>0$ such that
	\[ \bigl\lVert\gam{Z,n}{t}-\gam{J}{t}\bigr\rVert \leq K_Q\Delta_n \]
	for all $t\in[0,T]$. In particular,
	\[ \lim_{n\to\infty} \sup_{t\in[0,T]} \bigl\lVert \gam{Z,n}{t}-\gam{J}{t} \bigr\rVert = 0. \]
\end{theorem}

The proof of Theorem~\ref{thm:q_C_n_goes_to_q_D} is given in Appendix~\ref{app:long_proofs_fixed_times}. It makes use of a discrete version of Gronwall's Lemma for error accumulation, see Lemma~\ref{lem:discrete_gronwall} in Appendix~\ref{app:technical_lemmas}.

Using the uniform convergence of the conditional covariance matrices $\gam{Z,n}{}$ to $\gam{J}{}$ we can also deduce convergence of the corresponding conditional mean $\muhat{Z,n}{}$ to $\muhat{J}{}$ in an $\Lp$-sense.

\begin{theorem}\label{thm:muhat_C_n_goes_to_muhat_D}
	Let $p\in[1,\infty)$. Under Assumption~\ref{ass:deterministic_time_points} there exists a constant $K_{m,p}>0$ such that
	\[ \E\Bigl[\bigl\lVert \muhat{Z,n}{t}-\muhat{J}{t}\bigr\rVert^p\Bigr]\leq K_{m,p}\Delta_n^{p/2} \]
	for all $t\in[0,T]$. In particular,
	\[ \lim_{n\to\infty} \sup_{t\in[0,T]}\E\Bigl[\bigl\lVert \muhat{Z,n}{t}-\muhat{J}{t}\bigr\rVert^p\Bigr] = 0. \]
\end{theorem}

The proof of Theorem~\ref{thm:muhat_C_n_goes_to_muhat_D} can also be found in Appendix~\ref{app:long_proofs_fixed_times}.
Theorems~\ref{thm:q_C_n_goes_to_q_D} and \ref{thm:muhat_C_n_goes_to_muhat_D} state that in the setting of Assumption~\ref{ass:deterministic_time_points} the filter of a $Z$-investor observing $n$ equidistant expert opinions on $[0,T]$ converges to the filter of the $J$-investor. Recalling that the $J$-investor observes the diffusion processes $R$ and $J$, this implies that the information obtained from observing the discrete-time expert opinions is for large $n$ arbitrarily close to the information that comes with observing the continuous-time diffusion-type expert $J$.
This diffusion approximation of the discrete expert opinions is useful since the associated filter equations for $\muhat{J}{}$ and $\gam{J}{}$ are much simpler than those for $\muhat{Z,n}{}$ and $\gam{Z,n}{}$ which contain updates at information dates. Computing $\gam{Z,n}{}$ on $[0,T]$ in the multivariate case requires the numerical solution of a Riccati differential equation on each subinterval $[t_k, t_{k+1})$. For high numbers $n$ of expert opinions this leads to very small time steps and high computing times. For computing the $J$-investor's filter one has to find the solution to only one Riccati differential equation on $[0,T]$ for which we can use more efficient numerical solvers.
We will see in Section~\ref{sec:application_utility_maximization} that the convergence results carry over to convergence of the value function in a portfolio optimization problem.

\begin{remark}
	Note that for the convergence of the conditional covariance matrices $\gam{Z,n}{}$ to $\gam{J}{}$ in Theorem~\ref{thm:q_C_n_goes_to_q_D} we do not need the assumption that $Z^{(n)}_k$ is given as in \eqref{eq:expert_opinions_for_fixed_information_dates}. This is because the conditional covariance matrices $\gam{Z,n}{t}$ do not depend on the actual form of the expert opinions, see Lemma~\ref{lem:filter_C_dynamics}. Hence, it would be sufficient to assume that the experts' covariance matrices are given by $\Gamma_k^{(n)}=\Gamma^{(n)}=\frac{1}{\Delta_n}\sigma_J\sigma_J^\transp$. The assumption on the form of $Z^{(n)}_k$ is only needed in Theorem~\ref{thm:muhat_C_n_goes_to_muhat_D} where the conditional mean $\muhat{Z,n}{t}$ is considered.
\end{remark}

\section{Diffusion Approximation of Filters for Random Information Dates}\label{sec:diffusion_approximation_of_filters_for_random_information_dates}

In this section we consider the situation where the experts' opinions do not arrive at deterministic time points but at random information dates $T_k$, where the waiting times $T_{k+1}-T_k$ between information dates are independent and exponentially distributed with rate $\lambda>0$. Recall that we have set $T_0=0$ for ease of notation. The information dates can therefore be seen as the jump times of a standard Poisson process with intensity $\lambda$.
In this situation, the total number of expert opinions arriving in $[0,T]$ is no longer deterministic. However, as the intensity $\lambda$ increases, expert opinions will arrive more and more frequently. So the question we address in this section is, in analogy to sending $n$ to infinity in the last section, what happens when $\lambda$ goes to infinity. We use a superscript $\lambda$ to emphasize the dependence on the intensity.
The expert opinions are of the form
\begin{equation}\label{eq:form_of_expert_opinions_random}
	Z_k^{(\lambda)} = \mu_{T_k^{(\lambda)}}+(\Gamma_k^{(\lambda)})^{\frac12}\varepsilon_k^{(\lambda)}.
\end{equation}

For constant variances $\Gamma_k^{(\lambda)}=\Gamma$, i.e.\ when there is some constant level of the expert's reliability which does not depend on the arrival intensity $\lambda$, one can derive a similar result for the convergence to full information as in the case of deterministic information dates. This result implies that for large $\lambda$ the $Z$-investor approximates the fully informed investor. More precisely, it holds
\[ \lim_{\lambda\to\infty} \E\bigl[\bigl\lVert \gam{Z,\lambda}{t} \bigr\rVert\bigr] = 0 \quad \text{and} \quad \lim_{\lambda\to\infty}  \E\bigl[\bigl\lVert \muhat{Z,\lambda}{t} - \mu_t \bigr\rVert^2\bigr] = 0 \]
for all $t\in (0,T]$, see Gabih et al.~\cite{gabih_kondakji_wunderlich_2018}.
In contrast to the above case we now again assume that, as the frequency of expert opinions increases, the variance of the expert opinions $Z_k^{(\lambda)}$ also increases. As in Section~\ref{sec:diffusion_approximation_of_filters_for_deterministic_information_dates} it will turn out that letting $\Gamma_k^{(\lambda)}$ grow linearly in $\lambda$ is the proper scaling for deriving diffusion limits.

\begin{assumption}\label{ass:random_time_points}
	Let $(N^{(\lambda)}_t)_{t\in[0,T]}$ be a standard Poisson process with intensity $\lambda>0$ that is independent of the Brownian motions in the model. Define the information dates $(T_k^{(\lambda)})_{k=1,\dots,N^{(\lambda)}_T}$ as the jump times of that process and set $T^{(\lambda)}_0=0$. Furthermore, let the experts' covariance matrices be given as $\Gamma_k^{(\lambda)}=\Gamma^{(\lambda)}=\lambda\sigma_J\sigma_J^\transp$ for $k=1,\dots,N^{(\lambda)}_T$. Further, we assume that in~\eqref{eq:form_of_expert_opinions_random} the $\calN(0,I_d)$-distributed random variables $\varepsilon_k^{(\lambda)}$ are linked with the Brownian motion $W^J$ from~\eqref{eq:continuous_expert_J} via
	\[ \varepsilon_k^{(\lambda)}=\sqrt{\lambda}\int_{\frac{k-1}{\lambda}}^{\frac{k}{\lambda}}\rmd W^J_s, \]
	so that
	\begin{equation}\label{eq:expert_opinions_for_random_information_dates}
		Z_k^{(\lambda)} = \mu_{T_k^{(\lambda)}}+\lambda\sigma_J\int_{\frac{k-1}{\lambda}}^{\frac{k}{\lambda}}\rmd W^J_s
	\end{equation}
	is the expert opinion at information date $T_k^{(\lambda)}$. Note that for defining the $Z_k^{(\lambda)}$, the Brownian motion $W^J$ has to be extended to a Brownian motion on $[0,\infty)$.
\end{assumption}

Given a realization of the drift process at the random information date $T_k^{(\lambda)}$, the only randomness in the expert opinion comes from the Brownian motion $W^J$ between the deterministic times $\frac{k-1}{\lambda}$ and $\frac{k}{\lambda}$. Recall that $W^J$ is the Brownian motion that drives the diffusion $J$ which we interpret as our continuous expert. Hence there is a direct connection between the discrete expert opinions $Z_k^{(\lambda)}$ and the continuous expert.

In the following, we will omit the superscript $\lambda$ at the time points $T_k^{(\lambda)}$ for better readability, keeping the dependence on the intensity in mind.

\begin{remark}
	At first glance, it seems more intuitive to construct the expert opinions as
	\[ \widetilde{Z}_k^{(\lambda)} = \mu_{T_k}+\sqrt{\lambda}\sigma_J\frac{1}{\sqrt{T_k-T_{k-1}}}\int_{T_{k-1}}^{T_k}\rmd W^J_s \]
	rather than in \eqref{eq:expert_opinions_for_random_information_dates}.
	However, we later want to prove convergence of $\muhat{Z,\lambda}{t}$ to $\muhat{J}{t}$, which requires to look at the difference of a weighted sum of $\frac{1}{\lambda}(Z_k^{(\lambda)}-\mu_{T_k})$ and $\int_0^t\gam{J}{s}\,\rmd W^J_s$. It turns out that when replacing $Z_k^{(\lambda)}$ with $\widetilde{Z}_k^{(\lambda)}$, this leads to an integral where the integrand is defined piecewisely as
	\[ \biggl(\frac{1}{\sqrt{\lambda(T_{k}-T_{k-1})}}-1\biggr)\gam{J}{s}. \]
	However, the term in brackets does not have a finite variance. This carries over to the weighted sum mentioned above. This is mainly due to the fact that for $X\sim\mathrm{Exp}(\lambda)$, the expectation of $\frac{1}{X}$ does not exist. When considering $Z_k^{(\lambda)}$ instead, the difference that appears has finite variance since the additional randomness from the information dates is missing.
	Intuitively, the problem with the $\widetilde{Z}_k^{(\lambda)}$ is that the expert opinions of this form put different weight on the paths of the Brownian motion $W^J$ in different intervals. This is in contrast to the continuous expert whose information comes from observing the diffusion $J$, driven by the Brownian motion $W^J$, continuously in time.
	Therefore, in terms of information about the Brownian motion $W^J$, the $Z_k^{(\lambda)}$ modelled as in \eqref{eq:expert_opinions_for_random_information_dates} are closer to the continuous expert than the $\widetilde{Z}_k^{(\lambda)}$.
\end{remark}

The aim of this section is to determine the behavior of the conditional covariance matrix $\gam{Z,\lambda}{}$ and of the conditional mean $\muhat{Z,\lambda}{}$ under Assumption~\ref{ass:random_time_points} when $\lambda$ goes to infinity, i.e.\ when expert opinions arrive more and more frequently, becoming at the same time less and less reliable. Here, it is useful to express the dynamics of $\gam{Z,\lambda}{}$ and $\muhat{Z,\lambda}{}$ in a way that comprises both the behavior between information dates and the jumps at times $T_k$. For this purpose, we work with a representation using a Poisson random measure as introduced in Cont and Tankov~\cite[Sec.~2.6]{cont_tankov_2004}.

\begin{definition}
	Let $(\Omega_0,\mathcal{A},\mathbb{Q})$ be a probability space and $\nu$ a measure on a measurable space $(E,\mathcal{E})$. A \emph{Poisson random measure} with intensity measure $\nu$ is a function $N\colon\Omega_0\times\mathcal{E}\to\N_0$ such that
	\begin{enumerate}
		\item For each $\omega\in\Omega_0$, $N(\omega,\cdot)$ is a measure on $(E,\mathcal{E})$.
		
		\item For every $B\in\mathcal{E}$, $N(\cdot,B)$ is a Poisson random variable with parameter $\nu(B)$.
		
		\item For disjoint $E_1,\dots,E_p\in\mathcal{E}$, the random variables $N(\cdot,E_1),\dots,N(\cdot,E_p)$ are independent.
	\end{enumerate}
	For a Poisson random measure $N$, the \emph{compensated measure} $\tilde{N}$ is defined by $\tilde{N}\colon\Omega_0\times\mathcal{E}\to\R$ with $\tilde{N}(\omega,B)=N(\omega,B)-\nu(B)$.
\end{definition}

The following proposition states the results we will need in the following. For a proof, see Cont and Tankov~\cite[Sec.~2.6.3]{cont_tankov_2004}.

\begin{proposition}\label{prop:properties_Poisson_random_measure}
	Let $E=[0,T]\times\R^d$. Let $(T_k)_{k\geq 1}$ be the jump times of a Poisson process with intensity $\lambda>0$ and let $U_k$, $k=1,2,\dots$, be a sequence of independent multivariate standard Gaussian random variables on $\R^d$. For any $I\in\mathcal{B}([0,T])$ and $B\in\mathcal{B}(\R^d)$ let
	\[ N(I\times B)=\sum_{k\colon T_k\in I} \mathbbm{1}_{\{U_k\in B\}} \]
	denote the number of jump times in $I$ where $U_k$ takes a value in $B$. Then $N$ defines a Poisson random measure and it holds:
	\begin{enumerate}[label=(\roman*)]
		\item The corresponding intensity measure $\nu$ satisfies
		\[ \nu([t_1,t_2]\times B)=\int_{[t_1,t_2]}\lambda\,\rmd t\int_B \varphi(u)\,\rmd u \]
		for $0\leq t_1\leq t_2\leq T$, where $\varphi$ is the multivariate standard normal density on $\R^d$.
		
		\item For Borel-measurable functions $g$ defined on $\R^d$ it holds
		\[ \sum_{k\colon T_k\in [0,t]} g(U_k)=\int_{[0,t]}\int_{\R^d}g(u)\,N(\rmd s,\rmd u). \]
	\end{enumerate}
\end{proposition}

Now we can use the Poisson random measure for reformulating the dynamics of $\gam{Z,\lambda}{}$.

\begin{proposition}\label{prop:integral_equations_for_q_D_and_q_C_lambda}
	Let $L\colon\R^{d\times d}\to\R^{d\times d}$ denote the function with
	\[ L(\gam{}{})=-\alpha \gam{}{}-\gam{}{}\alpha+\beta\beta^\transp-\gam{}{}(\sigma_R\sigma_R^\transp)^{-1}\gam{}{}. \]
	Then under Assumption~\ref{ass:random_time_points} we can write
	\[ \gam{J}{t} = \Sigma_0 + \int_0^t \bigl(L(\gam{J}{s}) - \gam{J}{s}(\sigma_J\sigma_J^\transp)^{-1}\gam{J}{s}\bigr)\,\rmd s \]
	and
	\begin{equation*}
		\begin{aligned}
			\gam{Z,\lambda}{t} &= \Sigma_0+\int_0^t \bigl(L(\gam{Z,\lambda}{s})-\lambda\gam{Z,\lambda}{s-}(\gam{Z,\lambda}{s-}+\lambda\sigma_J\sigma_J^\transp)^{-1}\gam{Z,\lambda}{s-}\bigr)\,\rmd s \\
			&\quad-\int_0^t\int_{\R^d}\gam{Z,\lambda}{s-}(\gam{Z,\lambda}{s-}+\lambda\sigma_J\sigma_J^\transp)^{-1}\gam{Z,\lambda}{s-}\,\tilde{N}(\rmd s,\rmd u)
		\end{aligned}
	\end{equation*}
	for any $t\in[0,T]$.
\end{proposition}

The proof of Proposition~\ref{prop:integral_equations_for_q_D_and_q_C_lambda} is given in Appendix~\ref{app:technical_lemmas}.

In the following, we give convergence results in analogy to those in Theorems~\ref{thm:q_C_n_goes_to_q_D} and~\ref{thm:muhat_C_n_goes_to_muhat_D} stating that the conditional covariance matrix and the conditional mean of the $Z$-investor converge to the conditional covariance matrix and conditional mean of the $J$-investor as $\lambda$ goes to infinity.
In the setting with deterministic equidistant information dates in Section~\ref{sec:diffusion_approximation_of_filters_for_deterministic_information_dates} the conditional covariance matrices were deterministic. For the conditional means we proved $\Lp$-convergence. Due to the joint Gaussian distribution of the conditional means it was enough to prove $\mathrm{L}^2$-convergence and use a result from Rosi\'{n}ski and Suchanecki~\cite{rosinski_suchanecki_1980} to generalize to $\Lp$-convergence. In the setting of this section with random information dates, the conditional covariance matrices of the $Z$-investor are random and the joint distribution of the conditional means is no longer Gaussian. Therefore, the generalization mentioned above does not apply here. Hence, we directly prove $\Lp$-convergence in the following.
The next theorem states $\Lp$-convergence of $\gam{Z,\lambda}{}$ to $\gam{J}{}$ on $[0,T]$ as $\lambda$ goes to infinity.

\begin{theorem}\label{thm:q_C_lambda_goes_to_q_D}
	Let $p\in[1,\infty)$. Under Assumption~\ref{ass:random_time_points} there exists a constant $\widetilde{K}_{Q, p}>0$ such that
	\[ \E\Bigl[\bigl\lVert\gam{Z,\lambda}{t}-\gam{J}{t}\bigr\rVert^p\Bigr] \leq \frac{\widetilde{K}_{Q,p}}{\lambda^{\overline{p}}} \]
	for all $t\in[0,T]$ and $\lambda\geq 1$, where $\overline{p}=\min\{\frac{p}{2},1\}$. In particular,
	\[ \lim_{\lambda\to\infty} \sup_{t\in[0,T]} \E\Bigl[\bigl\lVert\gam{Z,\lambda}{t}-\gam{J}{t}\bigr\rVert^p\Bigr] = 0. \]
\end{theorem}

The proof of Theorem~\ref{thm:q_C_lambda_goes_to_q_D} is given in Appendix~\ref{app:long_proofs_random_times}. It is based on applying Gronwall's Lemma in integral form which we recall in Lemma~\ref{lem:gronwall}.
We also prove $\Lp$-convergence of the conditional means.

\begin{theorem}\label{thm:muhat_C_lambda_goes_to_muhat_D}
	Let $p\in[1,\infty)$. Under Assumption~\ref{ass:random_time_points} there exists a constant $\widetilde{K}_{m,p}>0$ such that
	\[ \E\Bigl[\bigl\lVert \muhat{Z,\lambda}{t}-\muhat{J}{t}\bigr\rVert^p\Bigr]\leq \frac{\widetilde{K}_{m,p}}{\lambda^{\frac{\overline{p}}{2}}} \]
	for all $t\in[0,T]$ and $\lambda\geq 1$, where $\overline{p}=\min\{\frac{p}{2},1\}$. In particular,
	\[ \lim_{\lambda\to\infty} \sup_{t\in[0,T]}\E\Bigl[\bigl\lVert \muhat{Z,\lambda}{t}-\muhat{J}{t}\bigr\rVert^p\Bigr] = 0. \]
\end{theorem}

The proof of Theorem~\ref{thm:muhat_C_lambda_goes_to_muhat_D} can be found in Appendix~\ref{app:long_proofs_random_times}.

Theorems~\ref{thm:q_C_lambda_goes_to_q_D} and \ref{thm:muhat_C_lambda_goes_to_muhat_D} show that under Assumption~\ref{ass:random_time_points}, the filter of the $Z$-investor converges to the filter of the $J$-investor. These are the analogous results to those in Section~\ref{sec:diffusion_approximation_of_filters_for_deterministic_information_dates} where we have assumed deterministic and equidistant information dates. Here, we see that the convergence result also holds for non-deterministic information dates $T_k$ being defined as the jump times of a standard Poisson process, i.e.\ where the time between information dates is exponentially distributed with a parameter $\lambda>0$. When sending $\lambda$ to infinity, the frequency of expert opinions goes to infinity.

Again, as for the case with deterministic information dates, the assumption that $Z_k^{(\lambda)}$ is given as in \eqref{eq:expert_opinions_for_random_information_dates} is only needed for the proof of Theorem~\ref{thm:muhat_C_lambda_goes_to_muhat_D}. For the proof of Theorem~\ref{thm:q_C_lambda_goes_to_q_D} it is sufficient to assume that the experts' covariance matrices are of the form $\Gamma_k^{(\lambda)}=\Gamma^{(\lambda)}=\lambda\sigma_J\sigma_J^\transp$.

\begin{remark}
	Note that when comparing the convergence results from Theorems~\ref{thm:q_C_n_goes_to_q_D} and \ref{thm:q_C_lambda_goes_to_q_D} for the conditional covariance matrices in the case $p=2$, there is a difference in the speed of convergence that we have shown. For deterministic equidistant information dates, the speed of convergence of $\lVert \gam{Z,n}{t}-\gam{J}{t}\rVert^2$ to zero is of the order $\frac{1}{n^2}$. For random information dates, however, we only get a speed of $\frac{1}{\lambda}$ for the convergence of
	\[ \E\Bigl[\bigl\lVert\gam{Z,\lambda}{t}-\gam{J}{t}\bigr\rVert^2\Bigr] \]
	to zero.
	This can be explained by the additional randomness coming from the Poisson process that determines the information dates $T_k$ in this situation.
\end{remark}

The above theorems provide a useful diffusion approximation since the filter of the $J$-investor is easier to compute than the filter of the $Z$-investor for which there are updates at each information date. Further, the conditional covariance $\gam{J}{}$ is deterministic and can be computed offline in advance while $\gam{Z,\lambda}{}$ is a stochastic process that has to be updated when a new expert opinion arrives. For high-frequency expert opinions one may simplify the computation of $\muhat{Z,\lambda}{}$ by replacing the exact conditional covariance $\gam{Z,\lambda}{}$ by its diffusion approximation $\gam{J}{}$.
Given the discrete-time expert's covariance matrix $\Gamma$ and the arrival intensity $\lambda$ the volatility $\sigma_J$ is chosen such that $\sigma_J\sigma_J^\transp=\lambda^{-1}\Gamma$.

Even more important are the benefits from the simpler filter equations if we consider utility maximization  problems for financial markets with partial information and discrete-time expert opinions. See the next section for an application to logarithmic utility and Remark~\ref{rem:power_utility} as well as Kondakji~\cite[Ch.~7,8]{kondakji_2019} for the more involved power utility case where closed-form expressions for the optimal strategies are available for the $J$-investor but not for the $Z$-investor.

\section{Application to Utility Maximization}\label{sec:application_utility_maximization}

As an application of the convergence results from the last two sections we now consider a portfolio optimization problem in our financial market. For the sake of convenience, we assume here that the interest rate $r$ of the risk-free asset is equal to zero. However, the results below can easily be extended to a market model with $r\neq 0$.

An investor's trading in the market can be described by a self-financing trading strategy $(\pi_t)_{t\in[0,T]}$ with values in $\R^d$. Here, $\pi_t^i$, $i=1,\dots,d$, is the proportion of wealth that is invested in asset $i$ at time $t$. The corresponding wealth process $(X^{\pi}_t)_{t\in[0,T]}$ is then governed by the stochastic differential equation
\[ \rmd X^\pi_t = X^\pi_t \pi_t^\transp\bigl(\mu_t\,\rmd t + \sigma_R \,\rmd W^R_t\bigr) \]
with initial capital $X^\pi_0=x_0>0$. An investor's trading strategy has to be adapted to her investor filtration. To ensure strictly positive wealth, we also impose some integrability constraint on the trading strategies. Then we denote by
\[ \mathcal{A}^H(x_0) = \biggl\{\pi=(\pi_t)_{t\in[0,T]} \;\bigg|\; \pi \text{ is } \mathbb{F}^H\text{-adapted}, \; X^\pi_0=x_0, \; \E\biggl[\int_0^T \lVert\sigma^\transp\pi_t\rVert^2\,\rmd t\biggr]<\infty\biggr\} \]
the class of admissible trading strategies for the $H$-investor. The optimization problem we address is a utility maximization problem where investors want to maximize expected logarithmic utility of terminal wealth. Hence,
\begin{equation}\label{eq:optimization_problem}
	V^H(x_0) = \sup\Bigl\{\E\bigl[\log(X^\pi_T)\bigr] \;\Big|\; \pi\in\mathcal{A}^H(x_0)\Bigr\}
\end{equation}
is the value function of our optimization problem. This utility maximization problem under partial information has been solved in Brendle~\cite{brendle_2006} for the case of power utility. Karatzas and Zhao~\cite{karatzas_zhao_2001} address also the case with logarithmic utility. In Sass et al.~\cite{sass_westphal_wunderlich_2017}, the optimization problem has been solved for an $H$-investor with logarithmic utility in the context of the different information regimes addressed in this paper. We recall the result in the proposition below.

\begin{proposition}\label{prop:representation_of_value_function}
	The optimal strategy for the optimization problem \eqref{eq:optimization_problem} is $(\pi^{H,*}_t)_{t\in[0,T]}$ with $\pi^{H,*}_t=(\sigma_R\sigma_R^\transp)^{-1}\muhat{H}{t}$, and the optimal value is
	\begin{equation*}
		\begin{aligned}
			V^H(x_0) &= \log(x_0)+\frac{1}{2}\int_0^T \tr\bigl((\sigma_R\sigma_R^\transp)^{-1}\E[\muhat{H}{t}(\muhat{H}{t})^\transp]\bigr)\,\rmd t \\
			&= \log(x_0)+\frac{1}{2}\int_0^T \tr\bigl((\sigma_R\sigma_R^\transp)^{-1}\bigl(\Sigma_t+m_tm_t^\transp-\E[\gam{H}{t}]\bigr)\bigr)\,\rmd t.
		\end{aligned}
	\end{equation*}
\end{proposition}

\begin{proof}
	The form of the optimal strategy and the first representation of the value function are already given in Proposition~5.1, respectively Theorem~5.1 of Sass et al.~\cite{sass_westphal_wunderlich_2017}. For the second representation of the value function, note that
	\begin{equation*}
		\begin{aligned}
			\gam{H}{t}&=\E[(\mu_t-\muhat{H}{t})(\mu_t-\muhat{H}{t})^\transp\,|\, \mathcal{F}^H_t] \\
			&=\E[\mu_t\mu_t^\transp-\muhat{H}{t}\mu_t^\transp-\mu_t(\muhat{H}{t})^\transp+\muhat{H}{t}(\muhat{H}{t})^\transp\,|\, \mathcal{F}^H_t] \\
			&=\E[\mu_t\mu_t^\transp\,|\, \mathcal{F}^H_t]-\muhat{H}{t}(\muhat{H}{t})^\transp.
		\end{aligned}
	\end{equation*}
	Therefore, by taking expectation on both sides,
	\[ \E[m^H_t(m^H_t)^\transp]=\E[\mu_t\mu_t^\transp]-\E[\gam{H}{t}]=\Sigma_t+m_tm_t^\transp-\E[\gam{H}{t}], \]
	which we can plug into the first representation.
\end{proof}

Due to the representation of the optimal strategy via the conditional means it follows directly from Theorems~\ref{thm:muhat_C_n_goes_to_muhat_D} and~\ref{thm:muhat_C_lambda_goes_to_muhat_D} that the optimal strategy of the $Z$-investor converges in the $\Lp$-sense to the optimal strategy of the $J$-investor as $n$, respectively $\lambda$, goes to infinity.

Further, note that the value function of the $H$-investor is an integral functional of the expectation of $(\gam{H}{t})_{t\in[0,T]}$. The convergence results of Theorems~\ref{thm:q_C_n_goes_to_q_D} and \ref{thm:q_C_lambda_goes_to_q_D} therefore carry over to convergence results for the respective value functions.
First, we address the situation with deterministic information dates $t_k$ from Section~\ref{sec:diffusion_approximation_of_filters_for_deterministic_information_dates} where we have shown uniform convergence of $\gam{Z,n}{}$ to $\gam{J}{}$.

\begin{corollary}\label{cor:convergence_value_functions_fixed_time_points}
	Under Assumption~\ref{ass:deterministic_time_points} it holds
	\[ \bigl\lvert V^{Z,n}(x_0)-V^J(x_0)\bigr\rvert \leq K_V\Delta_n \]
	for any initial wealth $x_0>0$, where $K_V=\frac{1}{2}K_QT\tr((\sigma_R\sigma_R^\transp)^{-1})$ with $K_Q$ from Theorem~\ref{thm:q_C_n_goes_to_q_D}. In particular, $\lim_{n\to\infty} V^{Z,n}(x_0) = V^J(x_0)$.
\end{corollary}

\begin{proof}
	From Proposition~\ref{prop:representation_of_value_function} we deduce
	\begin{equation}\label{eq:difference_of_value_functions}
		\begin{aligned}
			\bigl\lvert V^{Z,n}(x_0)-V^J(x_0)\bigr\rvert &= \biggl\lvert\frac{1}{2}\int_0^T \tr\bigl((\sigma_R\sigma_R^\transp)^{-1}(\gam{J}{t}-\gam{Z,n}{t})\bigr)\,\rmd t\biggr\rvert \\
			&\leq \frac{1}{2}\int_0^T \bigl\lvert\tr\bigl((\sigma_R\sigma_R^\transp)^{-1}(\gam{J}{t}-\gam{Z,n}{t})\bigr)\bigr\rvert\,\rmd t,
		\end{aligned}
	\end{equation}
	noting that $\gam{Z,n}{t}$ and $\gam{J}{t}$ are deterministic for every $t\in[0,T]$.
	Since $(\sigma_R\sigma_R^\transp)^{-1}$ is symmetric and positive definite, and $\gam{J}{t}-\gam{Z,n}{t}$ is symmetric, it follows from Lemma~1 in Wang et al.~\cite{wang_kuo_hsu_1986} that
	\[ \bigl\lvert\tr\bigl((\sigma_R\sigma_R^\transp)^{-1}(\gam{J}{t}-\gam{Z,n}{t})\bigr)\bigr\rvert \leq \tr\bigl((\sigma_R\sigma_R^\transp)^{-1}\bigr) \bigl\lVert \gam{J}{t}-\gam{Z,n}{t} \bigr\rVert. \]
	Inserting this into \eqref{eq:difference_of_value_functions} we then get from Theorem~\ref{thm:q_C_n_goes_to_q_D} that
	\[ \bigl\lvert V^{Z,n}(x_0)-V^J(x_0)\bigr\rvert \leq \frac{1}{2}T\tr\bigl((\sigma_R\sigma_R^\transp)^{-1}\bigr)K_Q\Delta_n \]
	which proves the claim when setting $K_V=\frac{1}{2}K_QT\tr((\sigma_R\sigma_R^\transp)^{-1})$.
\end{proof}

The analogous result also holds in the setting of Section~\ref{sec:diffusion_approximation_of_filters_for_random_information_dates} where information dates $T_k$ are the jump times of a Poisson process. Recall that in Theorem~\ref{thm:q_C_lambda_goes_to_q_D} we have shown convergence of $\gam{Z,\lambda}{}$ to $\gam{J}{}$.

\begin{corollary}\label{cor:convergence_value_functions_random_time_points}
	Under Assumption~\ref{ass:random_time_points} it holds
	\[ \bigl\lvert V^{Z,\lambda}(x_0)-V^J(x_0)\bigr\rvert \leq \frac{\widetilde{K}_V}{\sqrt{\lambda}} \]
	for any initial wealth $x_0>0$ and all $\lambda\geq 1$, where $\widetilde{K}_V=\frac{1}{2}\widetilde{K}_{Q,1}T\tr((\sigma_R\sigma_R^\transp)^{-1})$ with $\widetilde{K}_{Q,1}$ from Theorem~\ref{thm:q_C_lambda_goes_to_q_D}. In particular, $\lim_{\lambda\to\infty} V^{Z,\lambda}(x_0) = V^J(x_0)$.
\end{corollary}

\begin{proof}
	As in the proof of Corollary~\ref{cor:convergence_value_functions_fixed_time_points} we first use Proposition~\ref{prop:representation_of_value_function} to obtain
	\begin{equation*}
		\begin{aligned}
			\bigl\lvert V^{Z,\lambda}(x_0)-V^J(x_0)\bigr\rvert \leq \frac{1}{2}\int_0^T \E\Bigl[\bigl\lvert\tr\bigl((\sigma_R\sigma_R^\transp)^{-1}(\gam{J}{t}-\gam{Z,\lambda}{t})\bigr)\bigl\rvert\Bigr]\,\rmd t.
		\end{aligned}
	\end{equation*}
	By the same reasoning as in the proof of Corollary~\ref{cor:convergence_value_functions_fixed_time_points} and by applying Theorem~\ref{thm:q_C_lambda_goes_to_q_D} we get
	\begin{equation*}
		\begin{aligned}
			\bigl\lvert V^{Z,\lambda}(x_0)-V^J(x_0)\bigr\rvert &\leq \frac{1}{2}\int_0^T \E\Bigl[\tr\bigl((\sigma_R\sigma_R^\transp)^{-1}\bigr) \bigl\lVert \gam{J}{t}-\gam{Z,\lambda}{t} \bigr\rVert\Bigr]\,\rmd t \\
			&\leq \frac{1}{2} T\tr\bigl((\sigma_R\sigma_R^\transp)^{-1}\bigr)\frac{\widetilde{K}_{Q,1}}{\sqrt{\lambda}},
		\end{aligned}
	\end{equation*}
	for all $\lambda\geq 1$.
\end{proof}

Corollary~\ref{cor:convergence_value_functions_fixed_time_points} and Corollary~\ref{cor:convergence_value_functions_random_time_points} show that both under Assumption~\ref{ass:deterministic_time_points} and Assumption~\ref{ass:random_time_points}, the value function of the $Z$-investor converges to the value function of the $J$-investor when the frequency of information dates goes to infinity.

The following proposition shows that not only does the value function of the $Z$-investor converge to the value function of the $J$-investor, also the absolute difference of the utility attained by $\pi^{Z,*}$, respectively $\pi^{J,*}$, goes to zero when increasing the number or the frequency of discrete-time expert opinions. This implies that the utility of the $Z$-investor observing the discrete-time expert opinions also pathwise becomes arbitrarily close to the utility of the $J$-investor when the number of discrete-time expert opinions becomes large. For this result, we need the strong $\mathrm{L}^2$-convergence of the conditional expectations, convergence in distribution would not be enough here.

\begin{proposition}
	Under Assumption~\ref{ass:deterministic_time_points} it holds
	\[ \lim_{n\to\infty} \E\Bigl[\bigl|\log(X^{\pi^{Z,n,*}}_T)-\log(X^{\pi^{J,*}}_T)\bigr|\Bigr]=0, \]
	under Assumption~\ref{ass:random_time_points} it holds
	\[ \lim_{\lambda\to\infty} \E\Bigl[\bigl|\log(X^{\pi^{Z,\lambda,*}}_T)-\log(X^{\pi^{J,*}}_T)\bigr|\Bigr]=0. \]
\end{proposition}

\begin{proof}
	Consider the setting of Assumption~\ref{ass:deterministic_time_points}. Note that
	\begin{equation*}
		\begin{aligned}
			&\log(X^{\pi^{Z,n,*}}_T)-\log(X^{\pi^{J,*}}_T) \\
			&= \int_0^T \Bigl((\pi^{Z,n,*}_t-\pi^{J,*}_t)^\transp\mu_t-\frac{1}{2}\bigl(\lVert\sigma_R^\transp\pi^{Z,n,*}_t\rVert^2-\lVert\sigma_R^\transp\pi^{J,*}_t\rVert^2\bigr)\Bigr)\rmd t + \int_0^T (\pi^{Z,n,*}_t-\pi^{J,*}_t)^\transp\sigma_R\,\rmd W^R_t \\
			&= \int_0^T \Bigl((\muhat{Z,n}{t}-\muhat{J}{t})^\transp(\sigma_R\sigma_R^\transp)^{-1}\mu_t -\frac{1}{2}\bigl((\muhat{Z,n}{t})^\transp(\sigma_R\sigma_R^\transp)^{-1}\muhat{Z,n}{t}-(\muhat{J}{t})^\transp(\sigma_R\sigma_R^\transp)^{-1}\muhat{J}{t}\bigr)\Bigr)\rmd t \\
			&\quad+ \int_0^T (\muhat{Z,n}{t}-\muhat{J}{t})^\transp(\sigma_R\sigma_R^\transp)^{-1}\sigma_R\,\rmd W^R_t \\
			&= \frac{1}{2}\int_0^T (\muhat{Z,n}{t}-\muhat{J}{t})^\transp(\sigma_R\sigma_R^\transp)^{-1}(2\mu_t-\muhat{Z,n}{t}-\muhat{J}{t})\,\rmd t + \int_0^T (\muhat{Z,n}{t}-\muhat{J}{t})^\transp(\sigma_R\sigma_R^\transp)^{-1}\sigma_R\,\rmd W^R_t,
		\end{aligned}
	\end{equation*}
	where we have used the representation of the optimal strategies from Proposition~\ref{prop:representation_of_value_function}.
	When applying the absolute value and the expectation we obtain
	\begin{equation}\label{eq:difference_of_objectives}
		\begin{aligned}
			\E\Bigl[\bigl|\log(X^{\pi^{Z,n,*}}_T)-\log(X^{\pi^{J,*}}_T)\bigr|\Bigr] &\leq \frac{1}{2}\E\biggl[\biggl|\int_0^T (\muhat{Z,n}{t}-\muhat{J}{t})^\transp(\sigma_R\sigma_R^\transp)^{-1}(\mu_t-\muhat{Z,n}{t})\,\rmd t\biggr|\biggr]\\
			&\quad+\frac{1}{2}\E\biggl[\biggl|\int_0^T (\muhat{Z,n}{t}-\muhat{J}{t})^\transp(\sigma_R\sigma_R^\transp)^{-1}(\mu_t-\muhat{J}{t})\,\rmd t\biggr|\biggr] \\
			&\quad+ \E\biggl[\biggl|\int_0^T (\muhat{Z,n}{t}-\muhat{J}{t})^\transp(\sigma_R\sigma_R^\transp)^{-1}\sigma_R\,\rmd W^R_t\biggr|\biggr].
		\end{aligned}
	\end{equation}
	For the first summand in~\eqref{eq:difference_of_objectives} we have, due to the Cauchy--Schwarz inequality,
	\begin{equation*}
		\begin{aligned}
			&\E\biggl[\biggl|\int_0^T (\muhat{Z,n}{t}-\muhat{J}{t})^\transp(\sigma_R\sigma_R^\transp)^{-1}(\mu_t-\muhat{Z,n}{t})\,\rmd t\biggr|\biggr] \\
			&\leq \E\biggl[\int_0^T \bigl|(\muhat{Z,n}{t}-\muhat{J}{t})^\transp(\sigma_R\sigma_R^\transp)^{-1}(\mu_t-\muhat{Z,n}{t})\bigr|\,\rmd t\biggr] \\
			&\leq \lVert(\sigma_R\sigma_R^\transp)^{-1}\rVert \E\biggl[\int_0^T \Vert\muhat{Z,n}{t}-\muhat{J}{t}\rVert\,\lVert\mu_t-\muhat{Z,n}{t}\rVert\,\rmd t\biggr] \\
			&\leq \lVert(\sigma_R\sigma_R^\transp)^{-1}\rVert \E\biggl[\int_0^T \Vert\muhat{Z,n}{t}-\muhat{J}{t}\rVert^2\,\rmd t\biggr]^{1/2}\E\biggl[\int_0^T \Vert\mu-\muhat{Z,n}{t}\rVert^2\,\rmd t\biggr]^{1/2}.
		\end{aligned}
	\end{equation*}
	The right-hand side of this expression goes to zero when $n$ goes to infinity by Theorem~\ref{thm:muhat_C_n_goes_to_muhat_D} and by boundedness of $\gam{Z,n}{}$, see Lemma~\ref{lem:boundedness_of_covariances}.
	The second summand in~\eqref{eq:difference_of_objectives} goes to zero by an analogous argumentation.
	For the third summand in~\eqref{eq:difference_of_objectives}, note that
	\begin{equation*}
		\begin{aligned}
			\E\biggl[\biggl|\int_0^T (\muhat{Z,n}{t}-\muhat{J}{t})^\transp(\sigma_R\sigma_R^\transp)^{-1}\sigma_R\,\rmd W^R_t\biggr|\biggr]
			&\leq \E\biggl[\biggl(\int_0^T (\muhat{Z,n}{t}-\muhat{J}{t})^\transp(\sigma_R\sigma_R^\transp)^{-1}\sigma_R\,\rmd W^R_t\biggr)^2\biggr]^{1/2} \\
			&= \E\biggl[\int_0^T \lVert\sigma_R^\transp(\sigma_R\sigma_R^\transp)^{-1}(\muhat{Z,n}{t}-\muhat{J}{t})\rVert^2\,\rmd t\biggr]^{1/2} \\
			&\leq \lVert\sigma_R^\transp(\sigma_R\sigma_R^\transp)^{-1}\rVert\E\biggl[\int_0^T \lVert\muhat{Z,n}{t}-\muhat{J}{t}\rVert^2\,\rmd t\biggr]^{1/2}.
		\end{aligned}
	\end{equation*}
	In the second step we have used the It\^{o} isometry. Again, the right-hand side of the above inequality goes to zero as $n$ goes to infinity by Theorem~\ref{thm:muhat_C_n_goes_to_muhat_D}.
	The proof for the convergence under Assumption~\ref{ass:random_time_points} is completely analogous.
\end{proof}

Note that the convergence of the value functions can also be deduced directly from the previous proposition. However, the proofs that we give in Corollaries~\ref{cor:convergence_value_functions_fixed_time_points}, respectively~\ref{cor:convergence_value_functions_random_time_points} using the convergence of the conditional covariance matrices are more direct and thus yield a sharper bound for the order of convergence than what we would get from the previous proposition.

\begin{remark}\label{rem:power_utility}
	For simplicity, we have restricted ourselves in this section to the case with logarithmic utility, where $\mathrm{L}^2$-convergence of the conditional covariance matrices and the conditional means is sufficient for proving convergence of the value functions and optimal strategies.
	Portfolio problems that consider maximization of expected power utility instead of logarithmic utility are typically much more demanding.
	We have seen that for logarithmic utility the value function is given in terms of an integral functional of the expected conditional variance of the filter. The resulting optimal portfolio strategy is myopic and depends on the current drift estimate only.
	
	For power utility, the value functions can be expressed as the expectation of the exponential of a quite involved integral functional of the conditional mean. Hence it depends on the complete filter distribution and not only on its second-order moments.
	Further, the optimal strategies are no longer myopic and do not depend only on the current drift estimate but contain correction terms depending on the distribution of the future drift estimates.
	Therefore, for power utility, one needs the $\Lp$-convergence for $p>2$ for proving convergence of the value functions, $\mathrm{L}^2$-convergence would not be enough.
	
	For the portfolio problem of the $Z$-investor in the power utility case closed-form expressions as above for the optimal strategies in terms of the filter are no longer available. One can apply the dynamic programming approach to the associated stochastic optimal control problem. For the $Z$-investor this leads to dynamic programming equations (DPEs) for the value function in form of a partial integro-differential equation (PIDE), see Kondakji~\cite[Ch.~7]{kondakji_2019}. Solutions of those DPEs can usually only be determined numerically.
	The optimal strategy can be given in terms of that value function and the filter processes $\muhat{Z,\lambda}{}$ and $\gam{Z,\lambda}{}$. Meanwhile, for the $J$-investor the above approach leads to DPEs which can be solved explicitly such that the value function can be given in terms of solutions to some Riccati equations. Again, the optimal strategies can be computed in terms of the value function and the filter processes $\muhat{J}{}$ and $\gam{J}{}$.
	
	Diffusion approximations for the filter and the value function thus allow us to find approximate solutions for the $Z$-investor which can be given in closed form and with less numerical effort. This is extremely helpful for financial markets with multiple assets since the numerical solution of the resulting problem suffers from the curse of dimensionality and becomes intractable. While for a model with a single asset the PIDE has two spatial variables, for two assets there are already five and for three assets nine variables. For details we refer to our forthcoming papers on that topic.
\end{remark}

\section{Numerical Example}\label{sec:numerical_example}

In this section we illustrate our convergence results from the previous sections by a numerical example. We consider a financial market with investment horizon one year. For simplicity, we assume that there is only one risky asset in the market, i.e.\ $d=1$. Let the parameters of our model be defined as in Table~\ref{tab:model_parameters_for_numerical_example}.

\begin{table}[ht]   
	\centering
	\begin{tabular}{llll}
		\hline
		investment horizon			& $T$		&$=$	& 1 \tstrut\\
		interest rate				& $r$		&$=$	& 0 \\
		mean reversion speed of drift process	& $\alpha$	&$=$	& 3 \\
		volatility of drift process		& $\beta$	&$=$	& 1 \\
		mean reversion level of drift process	& $\delta$	&$=$	& 0.05 \\
		initial mean of drift process		& $m_0$		&$=$	& 0.05 \\
		initial variance of drift process	& $\Sigma_0$	&$=$	& 0.2 \\
		volatility of returns			& $\sigma_R$	&$=$	& 0.25 \\
		volatility of continuous expert		& $\sigma_J$	&$=$	& 0.2 \\\hline
	\end{tabular}
	\caption{Model parameters for numerical example}\label{tab:model_parameters_for_numerical_example}
\end{table}

First, we illustrate our results from Section~\ref{sec:diffusion_approximation_of_filters_for_deterministic_information_dates} in the setting with deterministic equidistant information dates $t_k=k\Delta_n$, $k=1,\dots,n$, where $\Delta_n=\frac{T}{n}$. Recall that the variance of the discrete-time expert in that case is
\[ \Gamma^{(n)}=\frac{1}{\Delta_n}\sigma_J^2 \]
and that expert opinions are defined as in \eqref{eq:expert_opinions_for_fixed_information_dates} by
\[ Z^{(n)}_k = \mu_{t_k}+\frac{1}{\Delta_n}\sigma_J\int_{t_k}^{t_{k+1}} \rmd W^J_s \]
for $k=1,\dots,n$.
In Figure~\ref{fig:fixed_times_simulation} we plot the filters of the $R$-, $J$- and $Z$-investor against time. For the $Z$-investor we consider the cases $n=10,20,100$. In the upper plot one sees the conditional variances $\gam{R}{}$ and $\gam{J}{}$ as well as $\gam{Z,n}{}$ plotted against time. The lower plot shows a realization of the conditional means $\muhat{R}{}$, $\muhat{J}{}$ and $\muhat{Z,n}{}$ for the same parameters.

Recall that $\gam{R}{}$ and $\gam{J}{}$ as well as $\gam{Z,n}{}$ for any $n\in\N$ are deterministic. In the upper plot of Figure~\ref{fig:fixed_times_simulation} one sees that for any fixed $t\in[0,T]$, the value of $\gam{J}{t}$ as well as the value of $\gam{Z,n}{t}$ for any $n$ is less or equal than the value of $\gam{R}{t}$. This is due to Lemma~\ref{lem:boundedness_of_covariances}. For the $Z$-investors one sees that the updates at information dates lead to a decrease in the conditional variance. As the number $n$ increases, the conditional variances $\gam{Z,n}{t}$ approach $\gam{J}{t}$ for any $t\in[0,T]$. This is due to what has been shown in Theorem~\ref{thm:q_C_n_goes_to_q_D}.

Note that for $t$ going to infinity, $\gam{R}{t}$ and $\gam{J}{t}$ approach a finite value. Convergence has been proven in Proposition~4.6 of Gabih et al.~\cite{gabih_kondakji_sass_wunderlich_2014} for markets with $d=1$ stock and generalized in Theorem~4.1 of Sass et al.~\cite{sass_westphal_wunderlich_2017} for markets with an arbitrary number of stocks. For $(\gam{Z,n}{t})_{t\geq 0}$ we observe a periodic behavior with asymptotic upper and lower bounds in the limit. This has been studied in detail in Sass et al.~\cite[Sec.~4.2]{sass_westphal_wunderlich_2017}.

In the lower subplot we show a realization of the various conditional means. For $\muhat{Z,n}{}$ the updating steps at information dates are visible. In general, we observe that when increasing the value of $n$, the distance between the paths of $\muhat{J}{}$ and $\muhat{Z,n}{}$ becomes smaller, as shown in Theorem~\ref{thm:muhat_C_n_goes_to_muhat_D}.

\begin{figure}[ht]
	\centering
	\setlength\figureheight{4cm}
	\setlength\figurewidth{0.9\textwidth}
	\includegraphics[]{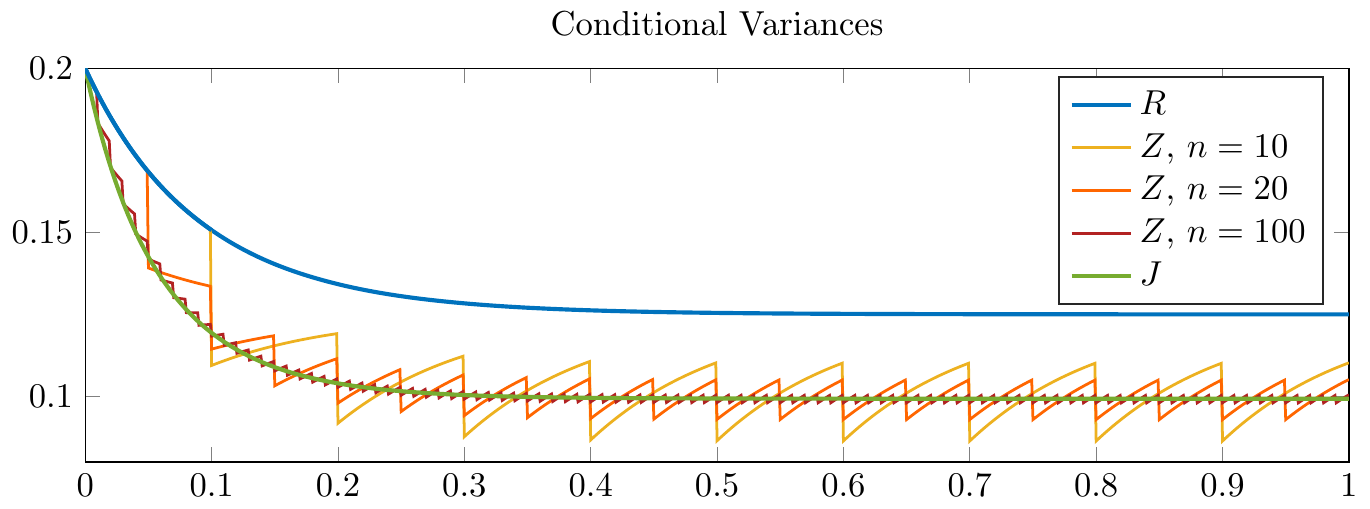}
	\includegraphics[]{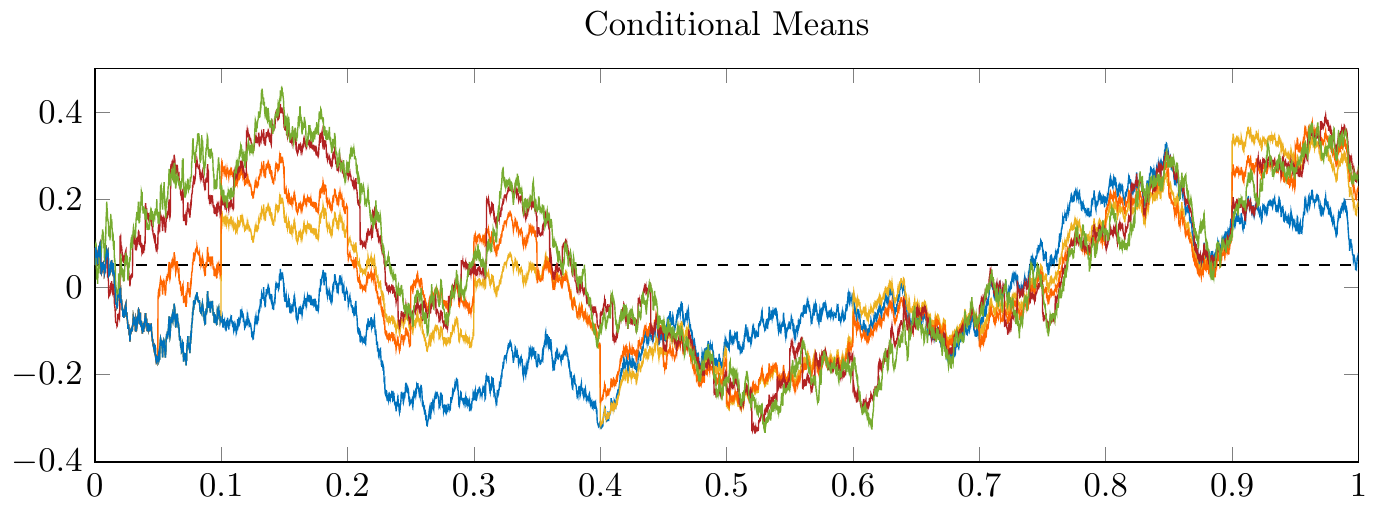}
	\caption{A simulation of the filters for deterministic equidistant information dates. The upper subplot shows the conditional variances of the $R$-, $J$- and $Z$-investor for various values of $n$, the lower subplot shows a realization of the corresponding conditional means. The dashed black line is the mean reversion level $\delta$ of the drift.}\label{fig:fixed_times_simulation}
\end{figure}

The analogous simulation can be done for the setting of Section~\ref{sec:diffusion_approximation_of_filters_for_random_information_dates} with random information dates $T_k$ defined as the jump times of a Poisson process. We again suppose that the model parameters are as given in Table~\ref{tab:model_parameters_for_numerical_example}. Recall that under Assumption~\ref{ass:random_time_points} the expert's variance is of the form $\Gamma^{(\lambda)}=\lambda\sigma_J^2$ with expert opinions given as in \eqref{eq:expert_opinions_for_random_information_dates} via
\[ Z_k^{(\lambda)} = \mu_{T_k}+\lambda\sigma_J\int_{\frac{k-1}{\lambda}}^{\frac{k}{\lambda}}\rmd W^J_s. \]
Figure~\ref{fig:random_times_simulation} shows, in addition to the filters of the $R$- and $J$-investor, the filters of the $Z$-investor for different intensities $\lambda$. Note that the conditional variances of the filter in the case of the $Z$-investor behave qualitatively like in the situation with deterministic information dates. The time at which the expert opinions arrive is now random, however. The waiting times between two information dates are exponentially distributed with parameter $\lambda$. As a consequence, the updates for the $Z$-investor do not take place as regularly as in Figure~\ref{fig:fixed_times_simulation}.

The upper plot of Figure~\ref{fig:random_times_simulation} shows realizations for $\lambda=10,100,1000$. In general, by increasing the value of the parameter $\lambda$, one can increase the frequency of information dates, causing convergence of $\gam{Z,\lambda}{t}$ to $\gam{J}{t}$ for any $t\in[0,T]$, as shown in Theorem~\ref{thm:q_C_lambda_goes_to_q_D}.
In the lower subplot, we see the corresponding realizations of $\muhat{Z,\lambda}{}$, in addition to $\muhat{R}{}$ and $\muhat{J}{}$ as before. Again, the updates in the conditional mean of the $Z$-investor are visible. What is also striking is that, when we consider the $Z$-investor with intensity $\lambda=10$, there are times where the distance between two sequent information dates is rather big. During those times, the conditional mean of the $Z$-investor comes closer to the conditional mean of the $R$-investor who does not observe any expert opinion. When the intensity $\lambda$ is increased, however, the conditional mean of the $Z$-investor approaches the conditional mean of the $J$-investor. For $\lambda=1000$, the conditional means $\muhat{Z,\lambda}{}$ and $\muhat{J}{}$ already behave quite similarly. Note, however, that for fixed information dates $\muhat{Z,n}{}$ is rather close to $\muhat{J}{}$ for $n=100$ already.
Convergence of $\muhat{Z,\lambda}{}$ to $\muhat{J}{}$ has been shown in Theorem~\ref{thm:muhat_C_lambda_goes_to_muhat_D}. The difference in the speed of convergence when comparing the situation with equidistant information dates to the situation with random information dates is also discussed there.

\begin{figure}[ht]
	\centering
	\setlength\figureheight{4cm}
	\setlength\figurewidth{0.9\textwidth}
	\includegraphics[]{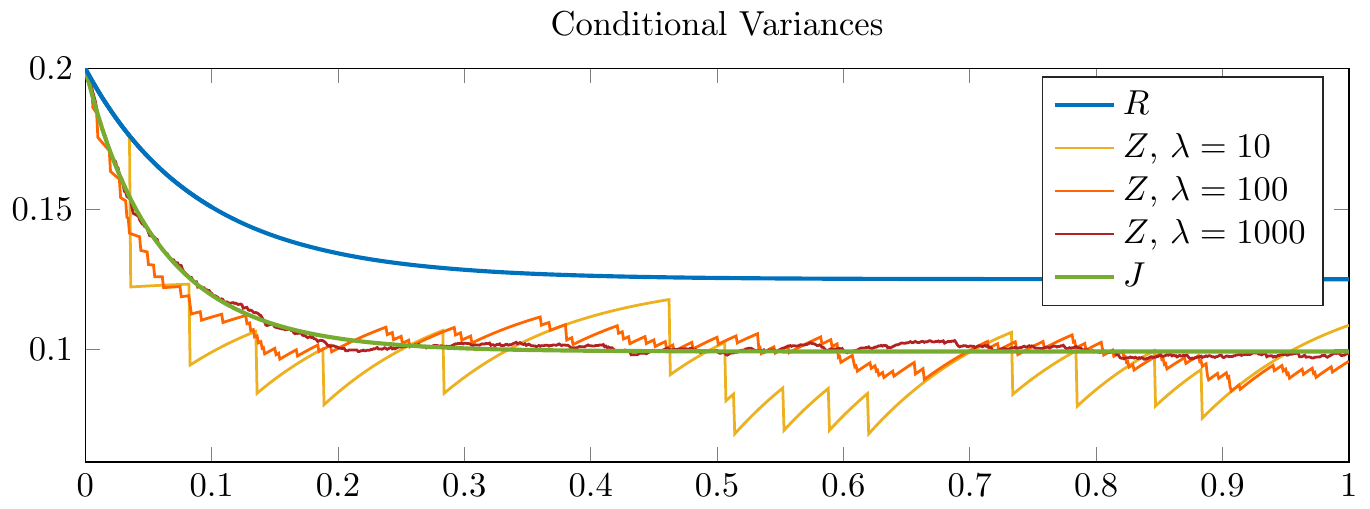}
	\includegraphics[]{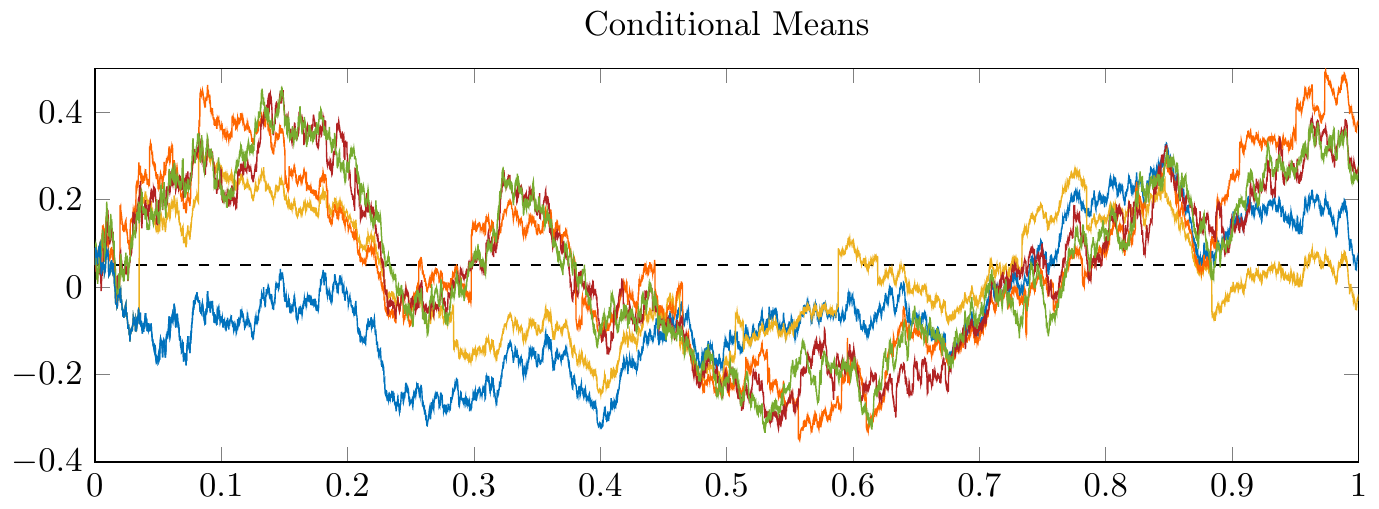}
	\caption{A simulation of the filters for random information dates coming as jump times of a Poisson process. The upper subplot shows the conditional variances of the $R$- and $J$-investor as well as realizations of $\gam{Z,\lambda}{}$ for various intensities $\lambda$, the lower subplot shows a realization of the corresponding conditional means. The dashed black line is the mean reversion level $\delta$ of the drift.}\label{fig:random_times_simulation}
\end{figure}

In the remaining part of this section we want to illustrate the convergence results in the portfolio optimization problem that was introduced in Section~\ref{sec:application_utility_maximization}. Recall that the value function of the $H$-investor has the form
\begin{equation}\label{eq:reminder_value_function}
	V^H(x_0) = \log(x_0)+\frac{1}{2}\int_0^T \tr\bigl((\sigma_R\sigma_R^\transp)^{-1}\bigl(\Sigma_t+m_tm_t^\transp-\E[\gam{H}{t}]\bigr)\bigr)\,\rmd t,
\end{equation}
i.e.\ it is an integral functional of the conditional covariance matrices $(\gam{H}{t})_{t\in[0,T]}$. This leads to convergence of $V^{Z,n}$ and $V^{Z,\lambda}$ to $V^J$ when $n$, respectively $\lambda$, goes to infinity, see Corollaries~\ref{cor:convergence_value_functions_fixed_time_points} and~\ref{cor:convergence_value_functions_random_time_points}. In Table~\ref{tab:value_functions_fixed_times} we list the value functions of the $R$-investor and of the $J$-investor as well as the value function of the $Z$-investor in the setting with $n$ equidistant information dates for different values of $n$. We assume that investors have initial capital $x_0=1$ and that the model parameters are again those from Table~\ref{tab:model_parameters_for_numerical_example}.
We see that the value functions $V^{Z,n}(1)$ are increasing in $n$ and approach the value $V^J(1)$ for large values of $n$.

Calculating the value function of the $Z$-investor in the situation with non-deterministic information dates is a little more involved. This is because the conditional covariance matrices $(\gam{Z,\lambda}{t})_{t\in[0,T]}$ are then also non-deterministic. The value function, see again \eqref{eq:reminder_value_function}, depends on the expectation of $\gam{Z,\lambda}{t}$ for $t\in[0,T]$. This value cannot be calculated easily. To determine the value function numerically for the parameters in Table~\ref{tab:model_parameters_for_numerical_example}, we therefore perform for each value of $\lambda$ a Monte Carlo simulation with $10\,000$ iterations. In each iteration, we generate a sequence of information dates as jump times of a Poisson process with intensity $\lambda$ and calculate the corresponding conditional variances $(\gam{Z,\lambda}{t})_{t\in[0,T]}$. By taking an average of all simulations this leads to a good approximation of $V^{Z,\lambda}(1)$. The diffusion approximation $V^J(1)$ is available in closed form, its computation does not require numerical methods. Table~\ref{tab:value_functions_random_times} shows the resulting estimations for $V^{Z,\lambda}(1)$ and in brackets the corresponding 95\% confidence intervals. 

The values $V^{Z,\lambda}(1)$ lie between $V^R(1)$ and $V^J(1)$, they are increasing in the intensity $\lambda$ and for large values of $\lambda$ they approach the value $V^J(1)$. This is in line with Corollary~\ref{cor:convergence_value_functions_random_time_points}. We also observe that $V^{Z,\lambda}(1)\leq V^{Z,n}(1)$ when setting the intensity $\lambda$ equal to the deterministic number $n$. Recall that an intensity $\lambda=n$ means that there are on average $n$ information dates in the time interval $[0,1]$. The randomness coming from the Poisson process however leads to a lower value function, compared to $V^{Z,n}(1)$. This difference is negligible for large intensities.

\begin{table}[ht]
	\centering
	\begin{subtable}{0.4\textwidth}
		\centering
		\begin{tabular}{lll}
			$H$	&	$n$	&	$V^{H,n}(1)$	\\\hline
			$R$	&		&	0.3410		\tstrut\vspace{1mm}\\
			$Z$	&	$10$	&	0.5245		\\
			$Z$	&	$100$	&	0.5511		\\
			$Z$	&	$1000$	&	0.5531		\\
			$Z$	&	$10\,000$	&	0.5533		\vspace{1mm}\\
			$J$	&		&	0.5533
		\end{tabular}
		\caption{Equidistant information dates}\label{tab:value_functions_fixed_times}
	\end{subtable}%
	\begin{subtable}{0.5\textwidth}
		\centering
		\begin{tabular}{lll}
			$H$	&	$\lambda$	&	$V^{H,\lambda}(1)$	\\\hline
			$R$	&			&	0.3410			\tstrut\vspace{1mm}\\
			$Z$	&	$10$		&	0.5221 (0.5211, 0.5230)	\\
			$Z$	&	$100$		&	0.5499 (0.5496, 0.5502)	\\
			$Z$	&	$1000$		&	0.5530 (0.5529, 0.5531)	\\
			$Z$	&	$10\,000$		&	0.5533 (0.5533, 0.5533)	\vspace{1mm}\\
			$J$	&			&	0.5533
		\end{tabular}
	\caption{Random information dates}\label{tab:value_functions_random_times}
	\end{subtable}
	\caption{Value function for different investors and in the situation with random information dates in brackets the 95\% confidence intervals for the $Z$-investor}
\end{table}

\appendix

\section{Auxiliary Results}\label{app:technical_lemmas}

In this appendix we give the proof of Proposition~\ref{prop:integral_equations_for_q_D_and_q_C_lambda} and collect some auxiliary results that are used in the proofs of our main results.  

\begin{proof}[Proof of Proposition~\ref{prop:integral_equations_for_q_D_and_q_C_lambda}]
	From Lemma~\ref{lem:filter_D_dynamics} one directly obtains
	\[ \ddt \gam{J}{t} = L(\gam{J}{t})-\gam{J}{t}(\sigma_J\sigma_J^\transp)^{-1}\gam{J}{t}, \]
	and the representation of $\gam{J}{t}$ follows immediately.
	From Lemma~\ref{lem:filter_C_dynamics} recall that between information dates the matrix differential equation for $\gam{Z,\lambda}{}$ reads
	\[ \frac{\rmd}{\rmd t} \gam{Z,\lambda}{t} = L(\gam{Z,\lambda}{t}). \]
	Now we can use Proposition~\ref{prop:properties_Poisson_random_measure} to include the updates of $\gam{Z,\lambda}{}$ at information dates and write
	\begin{equation}\label{eq:jump_sde_for_q_C_lambda}
		\rmd \gam{Z,\lambda}{t} = L(\gam{Z,\lambda}{t})\,\rmd t+\int_{\R^d} \bigl(\rho^{(\lambda)}(\gam{Z,\lambda}{t-})-I_d\bigr)\gam{Z,\lambda}{t-}\,N(\rmd t,\rmd u)
	\end{equation}
	for $\rho^{(\lambda)}(\gam{}{})=\Gamma^{(\lambda)}(\gam{}{}+\Gamma^{(\lambda)})^{-1}$. Note that the integrand is matrix-valued and the integral is defined componentwise.
	By \eqref{eq:jump_sde_for_q_C_lambda} we can write
	\begin{equation}\label{eq:integral_equation_1_for_q_C_lambda}
		\begin{aligned}
			\gam{Z,\lambda}{t}&=\Sigma_0+\int_0^t L(\gam{Z,\lambda}{s})\,\rmd s+\int_0^t\int_{\R^d}\bigl(\rho^{(\lambda)}(\gam{Z,\lambda}{s-})-I_d\bigr)\gam{Z,\lambda}{s-}\,N(\rmd s,\rmd u) \\
			&=\Sigma_0+\int_0^t L(\gam{Z,\lambda}{s})\,\rmd s+\int_0^t\int_{\R^d}\bigl(\rho^{(\lambda)}(\gam{Z,\lambda}{s-})-I_d\bigr)\gam{Z,\lambda}{s-}\,\tilde{N}(\rmd s,\rmd u) \\
			&\quad+\int_0^t\int_{\R^d}\bigl(\rho^{(\lambda)}(\gam{Z,\lambda}{s-})-I_d\bigr)\gam{Z,\lambda}{s-}\,\nu(\rmd s,\rmd u).
		\end{aligned}
	\end{equation}
	We see that
	\begin{equation*}
		\bigl(\rho^{(\lambda)}(\gam{}{})-I_d\bigr)\gam{}{}=\bigl(\Gamma^{(\lambda)}(\gam{}{}+\Gamma^{(\lambda)})^{-1}-I_d\bigr)\gam{}{}=-\gam{}{}(\gam{}{}+\Gamma^{(\lambda)})^{-1}\gam{}{}=-\gam{}{}(\gam{}{}+\lambda\sigma_J\sigma_J^\transp)^{-1}\gam{}{}.
	\end{equation*}
	Therefore, the last integral in \eqref{eq:integral_equation_1_for_q_C_lambda} can be written as
	\begin{equation*}
		\begin{aligned}
			\int_0^t\int_{\R^d}\bigl(\rho^{(\lambda)}(\gam{Z,\lambda}{s-})-I_d\bigr)\gam{Z,\lambda}{s-}\,\nu(\rmd s,\rmd u) &= -\int_0^t\int_{\R^d}\gam{Z,\lambda}{s-}(\gam{Z,\lambda}{s-}+\lambda\sigma_J\sigma_J^\transp)^{-1}\gam{Z,\lambda}{s-}\,\nu(\rmd s,\rmd u) \\
			&= -\int_0^t\int_{\R^d}\gam{Z,\lambda}{s-}(\gam{Z,\lambda}{s-}+\lambda\sigma_J\sigma_J^\transp)^{-1}\gam{Z,\lambda}{s-}\varphi(u)\lambda\,\rmd u\,\rmd s \\
			&= -\int_0^t \lambda\gam{Z,\lambda}{s-}(\gam{Z,\lambda}{s-}+\lambda\sigma_J\sigma_J^\transp)^{-1}\gam{Z,\lambda}{s-}\,\rmd s,
		\end{aligned}
	\end{equation*}
	where the second equality follows from Proposition~\ref{prop:properties_Poisson_random_measure} and the last equality is due to $\varphi$ being a probability density. Plugging back in into \eqref{eq:integral_equation_1_for_q_C_lambda} yields
	\begin{equation*}
		\begin{aligned}
			\gam{Z,\lambda}{t} &= \Sigma_0+\int_0^t \bigl(L(\gam{Z,\lambda}{s})-\lambda\gam{Z,\lambda}{s-}(\gam{Z,\lambda}{s-}+\lambda\sigma_J\sigma_J^\transp)^{-1}\gam{Z,\lambda}{s-}\bigr)\,\rmd s \\
			&\quad-\int_0^t\int_{\R^d}\gam{Z,\lambda}{s-}(\gam{Z,\lambda}{s-}+\lambda\sigma_J\sigma_J^\transp)^{-1}\gam{Z,\lambda}{s-}\,\tilde{N}(\rmd s,\rmd u),
		\end{aligned}
	\end{equation*}
	and the representation of $\gam{Z,\lambda}{t}$ is also proven.
\end{proof}

The following lemma can be interpreted as a discrete version of Gronwall's Lemma for error accumulation. A statement similar to Lemma~\ref{lem:discrete_gronwall} can be found in Demailly~\cite[Sec.~8.2.4]{demailly_1994}.

\begin{lemma}\label{lem:discrete_gronwall}
	Let $(a_j)_{j=0,\dots,n}$, $(h_j)_{j=0,\dots,n}$ be real-valued sequences with $a_j\geq 0$, $h_j>0$, and $L>0$, $b\geq 0$ real numbers such that
	\[ a_{j+1}\leq (1+h_jL)a_j+h_jb, \qquad j=0,1,\dots, n-1. \]
	Then for all $j=0,1,\dots, n$ it holds
	\[ a_j\leq \frac{\rme^{Lt_j}-1}{L}b+\rme^{Lt_j}a_0, \]
	where $t_j=\sum_{i=0}^{j-1}h_i.$
\end{lemma}

\begin{proof}
	The proof can be done by induction. For $j=0$ the claim is obvious. For the induction step we observe that $1+x\leq \rme^x$ for all $x\in\R$ and hence
	\[ a_{j+1} \leq (1+h_jL)a_j+h_jb \leq \rme^{h_jL}a_j+h_jb. \]
	Due to the induction hypothesis we therefore have
	\begin{equation*}
		\begin{aligned}
			a_{j+1} &\leq \rme^{h_jL}\Bigl(\frac{\rme^{Lt_j}-1}{L}b+\rme^{Lt_j}a_0\Bigr)+h_jb \\
			&=\Bigl(\frac{\rme^{L(t_j+h_j)}-\rme^{Lh_j}+h_jL}{L}\Bigr)b+\rme^{L(t_j+h_j)}a_0 \\
			&\leq \frac{\rme^{Lt_{j+1}}-1}{L}b+\rme^{Lt_{j+1}}a_0,
		\end{aligned}
	\end{equation*}
	which completes the proof.
\end{proof}

The next lemmas are used in the proof of Theorem~\ref{thm:muhat_C_n_goes_to_muhat_D}. Firstly, the following lemma is a Cauchy--Schwarz inequality for multidimensional integrals.

\begin{lemma}\label{lem:cauchy_schwarz_multidimensional}
	Let $(X_s)_{s\in[0,t]}$ be an $\R^d$-valued stochastic process. Then
	\[ \E\biggl[\biggl\lVert\int_0^t X_s\,\rmd s\biggr\rVert^2\biggr]\leq t\int_0^t \E\bigl[\lVert X_s\rVert^2\bigr]\,\rmd s. \]
\end{lemma}

\begin{proof}
	Firstly, pulling the norm into the integral increases the expression on the left-hand side, so
	\[ \E\biggl[\biggl\lVert\int_0^t X_s\,\rmd s\biggr\rVert^2\biggr]\leq\E\biggl[\biggl(\int_0^t \lVert X_s\rVert\,\rmd s\biggr)^2\biggr]. \]
	Now we can apply the usual Cauchy--Schwarz inequality to the one-dimensional integral and get
	\[ \E\biggl[\biggl(\int_0^t \lVert X_s\rVert\,\rmd s\biggr)^2\biggr]\leq\E\biggl[t\int_0^t \lVert X_s\rVert^2\,\rmd s\biggr]=t\int_0^t \E\bigl[\lVert X_s\rVert^2\bigr]\,\rmd s. \]
	The last step is due to Fubini.
\end{proof}

A key tool for estimations involving stochastic integrals is the It\^{o} isometry. The following lemma uses the isometry to obtain an estimation for multivariate integrals.

\begin{lemma}\label{lem:ito_isometry_multidimensional}
	Let $W=(W_s)_{s\in[0,t]}$ be an $m$-dimensional Brownian motion. Let $(H_s)_{s\in[0,t]}$ be an $\R^{d\times m}$-valued stochastic process that is independent of $W$, and $\tau$ a stopping time that is bounded by $t$ and also independent of $W$. Then
	\[ \E\biggl[\biggl\lVert\int_0^\tau H_s\,\rmd W_s\biggr\rVert^2\biggr] = \E\biggl[\int_0^\tau \lVert H_s\rVert_F^2\,\rmd s\biggr]\leq C_{\textrm{norm}}\E\biggl[\int_0^\tau \lVert H_s\rVert^2\,\rmd s\biggr], \]
	where $\lVert\cdot\rVert_F$ denotes the Frobenius norm and $C_{\textrm{norm}}>0$ only depends on the dimension $d\times m$ of the integrand $H$.
\end{lemma}

\begin{proof}
	Note that for fixed, deterministic $t$, the integral $\int_0^t H_s\,\rmd W_s$ is a random variable with values in $\R^d$. The $i$-th entry is
	\[ \sum_{j=1}^m \int_0^t H^{ij}_s\,\rmd W^j_s. \]
	Hence,
	\[ \biggl\lVert\int_0^t H_s\,\rmd W_s\biggr\rVert^2 = \sum_{i=1}^d \biggl(\sum_{j=1}^m \int_0^t H^{ij}_s\,\rmd W^j_s\biggr)^2. \]
	When applying the expectation, we get due to independence
	\begin{equation}\label{eq:ito_first_step}
		\begin{aligned}
			\E\biggl[\biggl\lVert\int_0^t H_s\,\rmd W_s\biggr\rVert^2\biggr] &= \sum_{i=1}^d \sum_{j,k=1}^m \E\biggl[\int_0^t H^{ij}_s\,\rmd W^j_s \int_0^t H^{ik}_s\,\rmd W^k_s\biggr] \\
			&= \sum_{i=1}^d \sum_{j=1}^m \E\biggl[\biggl(\int_0^t H^{ij}_s\,\rmd W^j_s\biggr)^2 \biggr].
		\end{aligned}
	\end{equation}
	Note that we can consider the filtration $(\mathcal{G}_s)_{s\in[0,t]}$ where $\mathcal{G}_s=\sigma(W_u, u\leq s)\vee \sigma(H_u, u\in[0,t])$. Since $H$ and $W$ are independent, $W$ is a Brownian motion with respect to $(\mathcal{G}_s)_{s\in[0,t]}$. Also, $H$ is obviously adapted with respect to $(\mathcal{G}_s)_{s\in[0,t]}$. Hence, we can apply the usual It\^{o} isometry and obtain that the right-hand side of \eqref{eq:ito_first_step} equals
	\[ \sum_{i=1}^d \sum_{j=1}^m \E\biggl[\int_0^t (H^{ij}_s)^2\,\rmd s \biggr] = \E\biggl[\int_0^t \lVert H_s\rVert_F^2\,\rmd s\biggr]. \]
	Now when taking into account the stopping time $\tau$, we can write
	\[ \E\biggl[\biggl\lVert\int_0^\tau H_s\,\rmd W_s\biggr\rVert^2\biggr]=\E\biggl[\biggl\lVert\int_0^t \mathbbm{1}_{\{s\leq\tau\}} H_s\,\rmd W_s\biggr\rVert^2\biggr]. \]
	Since $\tau$ is independent of $W$ we can deduce from the previous part of the proof that
	\[ \E\biggl[\biggl\lVert\int_0^t \mathbbm{1}_{\{s\leq\tau\}} H_s\,\rmd W_s\biggr\rVert^2\biggr] = \E\biggl[\int_0^t \lVert \mathbbm{1}_{\{s\leq\tau\}} H_s\rVert_F^2\,\rmd s\biggr] = \E\biggl[\int_0^\tau \lVert H_s\rVert_F^2\,\rmd s\biggr]. \]
	Equivalence of norms implies the existence of the constant $C_{\textrm{norm}}>0$ with the property that
	\[ \E\biggl[\int_0^\tau \lVert H_s\rVert_F^2\,\rmd s\biggr]\leq C_{\textrm{norm}}\E\biggl[\int_0^\tau \lVert H_s\rVert^2\,\rmd s\biggr], \]
	which concludes the proof.
\end{proof}

Another estimate that is useful in the convergence proofs is given in the following lemma.

\begin{lemma}\label{lem:estimation_lemma}
	Let $\kappa>0$ and let $\gam{\kappa}{}$ be a symmetric and positive-definite matrix in $\R^{d\times d}$ with $\lVert \gam{\kappa}{}\rVert\leq C_{\gam{}{}}$ for all $\kappa$. Then there exists a constant $\bar{C}>0$ such that
	\[ \Bigl\lVert \gam{\kappa}{}-\gam{\kappa}{}\bigl(\gam{\kappa}{}+\kappa\sigma_J\sigma_J^\transp\bigr)^{-1}\kappa\sigma_J\sigma_J^\transp \Bigr\rVert\leq\frac{\bar{C}}{\kappa}. \]
\end{lemma}

\begin{proof}
	For abbreviation let $A:=\gam{\kappa}{}$, $B:=\sigma_J\sigma_J^\transp$. Then we can write
	\begin{equation*}
		\begin{aligned}
			A-A(A+\kappa B)^{-1}\kappa B &= A(A+\kappa B)^{-1}(A+\kappa B-\kappa B)=A(A+\kappa B)^{-1}A\\
			&=\bigl(A^{-1}(A+\kappa B)A^{-1}\bigr)^{-1}=\bigl(A^{-1}+\kappa A^{-1}BA^{-1}\bigr)^{-1},
		\end{aligned}
	\end{equation*}
	and therefore
	\begin{equation*}
		\begin{aligned}
			\bigl\lVert A-A(A+\kappa B)^{-1}\kappa B\bigr\rVert &= \bigl\lVert \bigl(A^{-1}+\kappa A^{-1}BA^{-1}\bigr)^{-1}\bigr\rVert = \Bigl(\lmin(A^{-1}+\kappa A^{-1}BA^{-1})\Bigr)^{-1}\\
			&\leq \Bigl(\lmin(A^{-1})+\lmin(\kappa A^{-1}BA^{-1})\Bigr)^{-1} \leq \Bigl(\lmin(\kappa A^{-1}BA^{-1})\Bigr)^{-1} \\
			&= \frac{1}{\kappa} \lVert AB^{-1}A\rVert.
		\end{aligned}
	\end{equation*}
	Hence, we obtain
	\[ \Bigl\lVert \gam{\kappa}{}-\gam{\kappa}{}\bigl(\gam{\kappa}{}+\kappa\sigma_J\sigma_J^\transp\bigr)^{-1}\kappa\sigma_J\sigma_J^\transp \Bigr\rVert\leq\frac{C_{\gam{}{}}^2\lVert(\sigma_J\sigma_J^\transp)^{-1}\rVert}{\kappa} = \frac{\bar{C}}{\kappa}, \]
	where $\bar{C}=C_{\gam{}{}}^2\lVert(\sigma_J\sigma_J^\transp)^{-1}\rVert$.
\end{proof}

The next lemma states Gronwall's Lemma in integral form which we use in the proofs of Theorems~\ref{thm:q_C_lambda_goes_to_q_D} and \ref{thm:muhat_C_lambda_goes_to_muhat_D}. A proof can be found for example in Pachpatte~\cite[Sec.~1.3]{pachpatte_1997}.

\begin{lemma}[Gronwall]\label{lem:gronwall}
	Let $I=[a,b]$ be an interval and let $u,\alpha$ and $\beta\colon I\to[0,\infty)$ be continuous functions with
	\[ u(t)\leq\alpha(t)+\int_a^t \beta(s)u(s)\,\rmd s \]
	for all $t\in I$. Then
	\[ u(t)\leq\alpha(t)+\int_a^t \alpha(s)\beta(s)\rme^{\int_s^t \beta(r)\,\rmd r}\,\rmd s \]
	for all $t\in I$.
\end{lemma}

In Section~\ref{sec:diffusion_approximation_of_filters_for_random_information_dates} we work with a Poisson random measure. An important property of the compensated Poisson measure that we use for the proof of Theorem~\ref{thm:q_C_lambda_goes_to_q_D} is given in the following lemma, see Proposition~2.16 in Cont and Tankov~\cite{cont_tankov_2004}.

\begin{lemma}\label{lem:variance_of_integral_with_respect_to_compensated_measure}
	For an integrable real-valued function $f\colon[0,T]\times\R^d\to\R$, the process $(X_t)_{t\geq 0}$ with
	\[ X_t=\int_0^t\int_{\R^d} f(s,u)\,\tilde{N}(\rmd s,\rmd u) \]
	is a martingale with $\E[X_t]=0$ and
	\begin{equation}
		\var(X_t)=\E[X_t^2]=\E\biggl[\int_0^t\int_{\R^d}f^2(s,u)\lambda\varphi(u)\,\rmd u\,\rmd s\biggr].
	\end{equation}
\end{lemma}

\section{Proofs for Deterministic Information Dates}\label{app:long_proofs_fixed_times}

\subsection{Proof of Theorem \ref{thm:q_C_n_goes_to_q_D}: Convergence of Covariance Matrices}

Throughout the proof, we omit the superscript $n$ at information dates $t_k^{(n)}$ for the sake of better readability, keeping the dependence on $n$ in mind.
The proof is based on finding a recursive formula for the distance between $\gam{Z,n}{t_k-}$ and $\gam{J}{t_k}$ where we make use of an Euler approximation of $\gam{J}{}$.

\paragraph{Euler scheme approximation of $\boldsymbol{\gam{J}{}}$.}
Recall the dynamics of $\gam{J}{}$ from Lemma~\ref{lem:filter_D_dynamics}.
To shorten notation, let $G\colon\R^{d\times d}\to\R^{d\times d}$ with
\[ G(\gam{}{})=-\alpha \gam{}{}-\gam{}{}\alpha+\beta\beta^\transp-\gam{}{}\bigl((\sigma_R\sigma_R^\transp)^{-1}+(\sigma_J\sigma_J^\transp)^{-1}\bigr)\gam{}{} \]
denote the right-hand side of the differential equation~\eqref{eq:Riccati_ode_D}. Then~\eqref{eq:Riccati_ode_D} reads as
\[ \ddt \gam{J}{t} = G(\gam{J}{t}). \]
The first step is to approximate $\gam{J}{}$ by an Euler scheme. Therefore, define $\gam{J,n}{}$ by setting
\begin{equation}\label{eq:definition_q_D_n}
	\gam{J,n}{t} := \gam{J}{t_k}+G(\gam{J}{t_k})(t-t_k)
\end{equation}
for all $t\in[t_k,t_{k+1})$. From a Taylor expansion we get that
\[ \gam{J}{t}=\gam{J}{t_k}+G(\gam{J}{t_k})(t-t_k)+\xi_t(t-t_k)^2 \]
where $\xi$ is a matrix-valued function involving the second derivative of $\gam{J}{t}$. Since $\gam{J}{}$ and its derivatives are bounded on $[0,T]$, see Lemma~\ref{lem:boundedness_of_covariances}, the matrices $\xi_t$ are bounded, hence the local truncation error is proportional to $\Delta_n^2$. In other words, there exists some $C_{\text{Euler}}>0$ such that
\begin{equation}\label{eq:Euler_approximation}
	\bigl\lVert \gam{J}{t}-\gam{J,n}{t} \bigr\rVert \leq C_{\text{Euler}}\Delta_n^2
\end{equation}
for all $t\in[0,T]$.

\paragraph{Estimation of the error in $\boldsymbol{G}$.}
Let $C_e$, $C_{\gam{}{}}>0$ and let $\varepsilon\in\R^{d\times d}$ with $\lVert\varepsilon\rVert\leq C_e$, $\gam{}{}\in\R^{d\times d}$ with $\lVert \gam{}{}\rVert\leq C_{\gam{}{}}$. Then
\begin{equation*}
	\begin{aligned}
		G(\gam{}{}+\varepsilon) &= -\alpha(\gam{}{}+\varepsilon)-(\gam{}{}+\varepsilon)\alpha+\beta\beta^\transp-(\gam{}{}+\varepsilon)\bigl((\sigma_R\sigma_R^\transp)^{-1}+(\sigma_J\sigma_J^\transp)^{-1}\bigr)(\gam{}{}+\varepsilon) \\
		&= G(\gam{}{})-\varepsilon\bigl(\alpha+\bigl((\sigma_R\sigma_R^\transp)^{-1}+(\sigma_J\sigma_J^\transp)^{-1}\bigr)(\gam{}{}+\varepsilon)\bigr)\\
		&\quad- \bigl(\alpha+\gam{}{}\bigl((\sigma_R\sigma_R^\transp)^{-1}+(\sigma_J\sigma_J^\transp)^{-1}\bigr)\bigr)\varepsilon.
	\end{aligned}
\end{equation*}
Hence,
\[ \lVert G(\gam{}{}+\varepsilon)-G(\gam{}{})\rVert \leq \lVert\varepsilon\rVert \bigl(2\lVert\alpha\rVert+\lVert(\sigma_R\sigma_R^\transp)^{-1}+(\sigma_J\sigma_J^\transp)^{-1}\rVert(2\lVert \gam{}{}\rVert + \lVert\varepsilon\rVert)\bigr). \]
This implies that there exists a constant $C_G>0$ such that for all $\varepsilon$, $\gam{}{}\in\R^{d\times d}$ with $\lVert\varepsilon\rVert\leq C_e$ and $\lVert \gam{}{}\rVert\leq C_{\gam{}{}}$ it holds
\begin{equation}\label{eq:estimation_error_in_G}
	\lVert G(\gam{}{}+\varepsilon)-G(\gam{}{})\rVert \leq C_G\lVert\varepsilon\rVert.
\end{equation}

\paragraph{Dynamics of $\boldsymbol{\gam{Z,n}{}}$.}
Next, we take a look at the dynamics of $\gam{Z,n}{}$, i.e.\ of the covariance matrix corresponding to the investor who observes the stock returns and the opinions of the discrete expert. We know that at information dates $t_k$, $k=1, \dots, n$, we have the update formula
\[ \gam{Z,n}{t_k} = \Gamma^{(n)}\bigl(\gam{Z,n}{t_k-}+\Gamma^{(n)}\bigr)^{-1}\gam{Z,n}{t_k-}. \]
Observe that
\[ \Gamma^{(n)}\bigl(\gam{Z,n}{t_k-}+\Gamma^{(n)}\bigr)^{-1} = \bigl(I_d+\gam{Z,n}{t_k-}(\Gamma^{(n)})^{-1}\bigr)^{-1} = \bigl(I_d+\Delta_n \gam{Z,n}{t_k-}(\sigma_J\sigma_J^\transp)^{-1}\bigr)^{-1} \]
which can be written as the Neumann series
\[ \sum_{i=0}^{\infty} \bigl(-\Delta_n \gam{Z,n}{t_k-} (\sigma_J\sigma_J^\transp)^{-1}\bigr)^i = I_d-\Delta_n \gam{Z,n}{t_k-}(\sigma_J\sigma_J^\transp)^{-1}+\sum_{i=2}^{\infty} \bigl(-\Delta_n \gam{Z,n}{t_k-} (\sigma_J\sigma_J^\transp)^{-1}\bigr)^i. \]
It follows that
\begin{equation}\label{eq:update_of_q_C_n}
	\gam{Z,n}{t_k} = \gam{Z,n}{t_k-}-\Delta_n\gam{Z,n}{t_k-}(\sigma_J\sigma_J^\transp)^{-1}\gam{Z,n}{t_k-}+\bar{R}^n
\end{equation}
where $\lVert \bar{R}^n\rVert\leq r\Delta_n^2$, since $\gam{Z,n}{t_k-}$ is bounded.
Between information dates, the matrix $\gam{Z,n}{}$ follows the dynamics
\[ \ddt \gam{Z,n}{t} = -\alpha\gam{Z,n}{t}-\gam{Z,n}{t}\alpha+\beta\beta^\transp-\gam{Z,n}{t}(\sigma_R\sigma_R^\transp)^{-1}\gam{Z,n}{t} \]
for $t\in[t_k,t_{k+1})$.

\paragraph{One time step for $\boldsymbol{\gam{Z,n}{}}$.}
In the following, we construct a formula that connects $\gam{Z,n}{t_{k+1}-}$ with $\gam{Z,n}{t_k-}$. Firstly, by making a Taylor expansion we see that
\[ \gam{Z,n}{t_{k+1}-} = \gam{Z,n}{t_k} + \bigl(-\alpha\gam{Z,n}{t_k}-\gam{Z,n}{t_k}\alpha+\beta\beta^\transp-\gam{Z,n}{t_k}(\sigma_R\sigma_R^\transp)^{-1}\gam{Z,n}{t_k}\bigr)\Delta_n+L^n, \]
where $\lVert L^n\rVert\leq C_L\Delta_n^2$.
Now, when inserting the representation of $\gam{Z,n}{t_k}$ from~\eqref{eq:update_of_q_C_n} and rearranging terms we can conclude that
\begin{equation}\label{eq:one_time_step_q_C_n}
	\gam{Z,n}{t_{k+1}-} = \gam{Z,n}{t_k-}+\Delta_n G(\gam{Z,n}{t_k-})+R^n,
\end{equation}
where $R^n$ is a matrix with $\lVert R^n\rVert\leq C_{\text{Taylor}}\Delta_n^2$ for $C_{\text{Taylor}}>0$.

\paragraph{Recursive formula for estimation error.}
For $k=0, \dots, n$, define $A_k=\gam{Z,n}{t_k-}-\gam{J}{t_k}$ and $a_k=\lVert A_k\rVert$. Our aim is to find a recursive formula that yields an upper bound for these estimation errors. Let $k\geq 0$. Then we have by~\eqref{eq:one_time_step_q_C_n} that
\begin{equation*}
	a_{k+1} = \lVert A_{k+1}\rVert = \lVert\gam{Z,n}{t_{k+1}-}-\gam{J}{t_{k+1}}\rVert = \lVert\gam{Z,n}{t_k-}+\Delta_n G(\gam{Z,n}{t_k-})+R^n-\gam{J}{t_{k+1}}\rVert.
\end{equation*}
Thus, by definition of $A_k$ and $\gam{J,n}{}$ as given in~\eqref{eq:definition_q_D_n},
\begin{equation*}
	\begin{aligned}
		a_{k+1} &= \lVert (\gam{J}{t_k}+A_k) + \Delta_n G(\gam{J}{t_k}+A_k) +R^n -\gam{J}{t_{k+1}}\rVert \\
		&= \lVert \gam{J}{t_k} + \Delta_n\bigl(G(\gam{J}{t_k})+G(\gam{J}{t_k}+A_k)-G(\gam{J}{t_k})\bigr) +A_k+R^n-\gam{J}{t_{k+1}}\rVert \\
		&= \lVert \gam{J,n}{t_{k+1}-} + \Delta_n\bigl(G(\gam{J}{t_k}+A_k)-G(\gam{J}{t_k})\bigr)+A_k+R^n-\gam{J}{t_{k+1}}\rVert.
	\end{aligned}
\end{equation*}
Now, the estimations from~\eqref{eq:Euler_approximation}, \eqref{eq:estimation_error_in_G} and~\eqref{eq:one_time_step_q_C_n} yield
\begin{equation*}
	a_{k+1} \leq C_{\text{Euler}}\Delta_n^2+\Delta_n C_G\lVert A_k\rVert + \lVert A_k\rVert + C_{\text{Taylor}}\Delta_n^2 =(1+\Delta_n C_G)a_k +(C_{\text{Euler}}+C_{\text{Taylor}})\Delta_n^2.
\end{equation*}
By a discrete version of Gronwall's Lemma, see Lemma~\ref{lem:discrete_gronwall} in the appendix, this implies
\begin{equation*}
	a_k \leq \frac{\rme^{C_G k\Delta_n}-1}{C_G}(C_{\text{Euler}}+C_{\text{Taylor}})\Delta_n
	\leq\frac{\rme^{C_G T}-1}{C_G}(C_{\text{Euler}}+C_{\text{Taylor}})\Delta_n
	=:\tilde{C}\Delta_n.
\end{equation*}
Therefore, for all $k=0, \dots, n$ we have
\begin{equation}\label{eq:inequality_at_time_points}
	\lVert \gam{Z,n}{t_k-}-\gam{J}{t_k} \rVert \leq \tilde{C}\Delta_n.
\end{equation}

\paragraph{Difference of $\boldsymbol{\gam{Z,n}{t}}$ and $\boldsymbol{\gam{J}{t}}$ for arbitrary $\boldsymbol{t}$.}
We now show that there exists some $K_Q>0$ such that $\lVert \gam{Z,n}{t}-\gam{J}{t} \rVert \leq K_Q\Delta_n$ for all $t\in[0,T]$. Let $t\in[0,T]$ with $t\in[t_k,t_{k+1})$. We can write
\[ \gam{Z,n}{t}-\gam{J}{t} = (\gam{Z,n}{t}-\gam{Z,n}{t_k-})+(\gam{Z,n}{t_k-}-\gam{J}{t_k})+(\gam{J}{t_k}-\gam{J}{t}), \]
and hence
\[ \lVert\gam{Z,n}{t}-\gam{J}{t}\rVert \leq \lVert\gam{Z,n}{t}-\gam{Z,n}{t_k-}\rVert+\lVert\gam{Z,n}{t_k-}-\gam{J}{t_k}\rVert+\lVert\gam{J}{t_k}-\gam{J}{t}\rVert. \]
By~\eqref{eq:inequality_at_time_points}, the second summand is bounded by $\tilde{C}\Delta_n$.
We now take a look at the other two summands. By definition of $\gam{J,n}{}$ we can write the third summand as
\begin{equation*}
	\begin{aligned}
		\lVert\gam{J}{t_k}-\gam{J}{t}\rVert &= \lVert \gam{J,n}{t}-G(\gam{J}{t_k})(t-t_k)-\gam{J}{t} \rVert \\
		&\leq \lVert\gam{J,n}{t}-\gam{J}{t}\rVert+\Delta_n\lVert G(\gam{J}{t_k})\rVert \\
		&\leq C_{\text{Euler}}\Delta_n^2+\Delta_n\lVert G(\gam{J}{t_k})\rVert
	\end{aligned}
\end{equation*}
where the second inequality is due to~\eqref{eq:Euler_approximation}. Since $G$ and $\gam{J}{}$ are continuous, the function $t\mapsto \lVert G(\gam{J}{t})\rVert$ is bounded by some $\tilde{C}_G$ on $[0,T]$. Hence,
\[\lVert\gam{J}{t_k}-\gam{J}{t}\rVert \leq C_{\text{Euler}}\Delta_n^2+\tilde{C}_G\Delta_n. \]
For the first summand we observe that, like in~\eqref{eq:one_time_step_q_C_n}, we get the representation
\[ \lVert\gam{Z,n}{t}-\gam{Z,n}{t_k-}\rVert = \lVert (t-t_k)G(\gam{Z,n}{t_k-})+R^n\rVert \]
for some matrix $R^n$ with $\lVert R^n\rVert\leq C_{\text{Taylor}}(t-t_k)^2$. Then the right-hand side is bounded by $\Delta_n\lVert G(\gam{Z,n}{t_k-})\rVert+C_{\text{Taylor}}\Delta_n^2$.
Also, we have
\begin{equation*}
	\begin{aligned}
		\lVert G(\gam{Z,n}{t_k-})\rVert &= \lVert G(\gam{J}{t_k}+\gam{Z,n}{t_k-}-\gam{J}{t_k})\rVert \leq \lVert G(\gam{J}{t_k})\rVert+C_G\lVert\gam{Z,n}{t_k-}-\gam{J}{t_k}\rVert
	\end{aligned}
\end{equation*}
by~\eqref{eq:estimation_error_in_G}. Again by continuity, $\lVert G(\gam{J}{t_k})\rVert\leq \tilde{C}_G$, and $\lVert\gam{Z,n}{t_k-}-\gam{J}{t_k}\rVert\leq\tilde{C}\Delta_n$ by~\eqref{eq:inequality_at_time_points}.

Putting these results together we obtain that there exists a constant $K_Q>0$ such that
\[ \lVert\gam{Z,n}{t}-\gam{J}{t}\rVert \leq K_Q\Delta_n \]
for all $t\in[0,T]$. \qed

\subsection{Proof of Theorem~\ref{thm:muhat_C_n_goes_to_muhat_D}: Convergence of Conditional Means}

We first prove the claim for $p=2$. We omit the superscript $n$ at information dates $t_k^{(n)}$ for the sake of better readability. The idea of the proof is to find a recursion for
\[ \E\Bigl[\bigl\lVert \muhat{Z,n}{t_k-}-\muhat{J}{t_k}\bigr\rVert^2\Bigr] \]
and to apply the discrete version of Gronwall's Lemma from Lemma~\ref{lem:discrete_gronwall} to derive an appropriate upper bound.

For the proof we introduce the notation
\[ L^{(n)}_k:=\gam{Z,n}{t_k-}\bigl(\gam{Z,n}{t_k-}+\Gamma^{(n)}\bigr)^{-1}\Gamma^{(n)} \]
for $k=1,\dots,n$. Then Lemma~\ref{lem:estimation_lemma} in particular implies that
\[ \lVert\gam{Z,n}{t_k-}-L^{(n)}_k\rVert\leq\bar{C}\Delta_n \]
for some constant $\bar{C}>0$.

\paragraph{Recursive formulas for $\boldsymbol{\muhat{J}{}}$ and $\boldsymbol{\muhat{Z,n}{}}$.}
The representation of $\muhat{J}{}$ via the stochastic differential equation in Lemma~\ref{lem:filter_D_dynamics} yields the recursion
\begin{equation}\label{eq:first_recursion_muhat_D}
	\begin{aligned}
		\muhat{J}{t_{k+1}} &= \rme^{-\alpha\Delta_n}\muhat{J}{t_k} + (I_d-\rme^{-\alpha\Delta_n})\delta + \int_{t_k}^{t_{k+1}} \rme^{-\alpha(t_{k+1}-s)}\gam{J}{s}(\sigma_R\sigma_R^\transp)^{-1}\sigma_R\,\rmd V^{J,1}_s \\
		&\quad+ \int_{t_k}^{t_{k+1}} \rme^{-\alpha(t_{k+1}-s)}\gam{J}{s}(\sigma_J\sigma_J^\transp)^{-1}\sigma_J\,\rmd V^{J,2}_s,
	\end{aligned}
\end{equation}
where
\begin{equation*}
	\begin{aligned}
		\sigma_R\,\rmd V^{J,1}_t &= \rmd R_t-\muhat{J}{t}\,\rmd t,\\
		\sigma_J\,\rmd V^{J,2}_t &= \rmd J_t-\muhat{J}{t}\,\rmd t,
	\end{aligned}
\end{equation*}
and $V^J=(V^{J,1},V^{J,2})^\transp$, the innovation process corresponding to the investor filtration $\calF^J$, is an $(m+l)$-dimensional $\calF^J$-Brownian motion.
Similarly, we get for the conditional mean $\muhat{Z,n}{}$ the recursion
\begin{equation}\label{eq:first_recursion_muhat_C_N}
	\muhat{Z,n}{t_{k+1}-} = \rme^{-\alpha\Delta_n}\muhat{Z,n}{t_k} + (I_d-\rme^{-\alpha\Delta_n})\delta + \int_{t_k}^{t_{k+1}} \rme^{-\alpha(t_{k+1}-s)}\gam{Z,n}{s}(\sigma_R\sigma_R^\transp)^{-1}\sigma_R\,\rmd V^Z_s,
\end{equation}
where
\[ \sigma_R\,\rmd V^Z_t=\rmd R_t-\muhat{Z,n}{t}\,\rmd t, \]
and $V^Z$, the innovation process corresponding to investor filtration $\calF^{Z,n}$, is an $m$-dimensional $\calF^{Z,n}$-Brownian motion. Furthermore, the update formula for $\muhat{Z,n}{}$ yields
\begin{equation}\label{eq:update_step_muhat_C_N}
	\begin{aligned}
		\muhat{Z,n}{t_k} &= \muhat{Z,n}{t_k-}+\bigl(I_d-\Gamma^{(n)}\bigl(\gam{Z,n}{t_k-}+\Gamma^{(n)}\bigr)^{-1}\bigr)\bigl(Z^{(n)}_k-\muhat{Z,n}{t_k-}\bigr) \\
		&= \muhat{Z,n}{t_k-}+\gam{Z,n}{t_k-}\bigl(\gam{Z,n}{t_k-}+\Gamma^{(n)}\bigr)^{-1}\biggl(\mu_{t_k}+\frac{1}{\Delta_n}\sigma_J\int_{t_k}^{t_{k+1}} \rmd W^J_s-\muhat{Z,n}{t_k-}\biggr) \\
		&= \muhat{Z,n}{t_k-}+\Delta_nL^{(n)}_k(\sigma_J\sigma_J^\transp)^{-1}\biggl(\mu_{t_k}+\frac{1}{\Delta_n}\sigma_J\int_{t_k}^{t_{k+1}} \rmd W^J_s-\muhat{Z,n}{t_k-}\biggr).
	\end{aligned}
\end{equation}
When looking at the difference between $\muhat{J}{}$ and $\muhat{Z,n}{}$ it is convenient to work with representations that use the same Brownian motions.

\paragraph{Relation between the innovation processes.}
Note that
\[ \sigma_R\,\rmd V^{J,1}_t = \rmd R_t-\muhat{J}{t}\,\rmd t = \sigma_R\,\rmd V^Z_t+(\muhat{Z,n}{t}-\muhat{J}{t})\,\rmd t \]
and
\[ \sigma_J\,\rmd V^{J,2}_t = \rmd J_t-\muhat{J}{t}\,\rmd t = \sigma_J\,\rmd W^J_t+(\mu_t-\muhat{J}{t})\,\rmd t. \]
Using this connection between the innovation processes, we obtain from~\eqref{eq:first_recursion_muhat_D} that
\begin{equation}\label{eq:second_recursion_muhat_D}
	\begin{aligned}
		\muhat{J}{t_{k+1}} &= \rme^{-\alpha\Delta_n}\muhat{J}{t_k} + (I_d-\rme^{-\alpha\Delta_n})\delta + \int_{t_k}^{t_{k+1}} \rme^{-\alpha(t_{k+1}-s)}\gam{J}{s}(\sigma_R\sigma_R^\transp)^{-1}\sigma_R\,\rmd V^Z_s \\
		&\quad+\int_{t_k}^{t_{k+1}} \rme^{-\alpha(t_{k+1}-s)}\gam{J}{s}(\sigma_R\sigma_R^\transp)^{-1}(\muhat{Z,n}{s}-\muhat{J}{s})\,\rmd s \\
		&\quad+ \int_{t_k}^{t_{k+1}} \rme^{-\alpha(t_{k+1}-s)}\gam{J}{s}(\sigma_J\sigma_J^\transp)^{-1}\sigma_J\,\rmd W^J_s \\
		&\quad+ \int_{t_k}^{t_{k+1}} \rme^{-\alpha(t_{k+1}-s)}\gam{J}{s}(\sigma_J\sigma_J^\transp)^{-1}(\mu_s-\muhat{J}{s})\,\rmd s.
	\end{aligned}
\end{equation}
Also, plugging~\eqref{eq:update_step_muhat_C_N} into~\eqref{eq:first_recursion_muhat_C_N} yields
\begin{equation}\label{eq:second_recursion_muhat_C_N}
	\begin{aligned}
		\muhat{Z,n}{t_{k+1}-} &= \rme^{-\alpha\Delta_n}\muhat{Z,n}{t_k-} + (I_d-\rme^{-\alpha\Delta_n})\delta + \int_{t_k}^{t_{k+1}} \rme^{-\alpha(t_{k+1}-s)}\gam{Z,n}{s}(\sigma_R\sigma_R^\transp)^{-1}\sigma_R\,\rmd V^Z_s \\
		&\quad +\rme^{-\alpha\Delta_n}L^{(n)}_k(\sigma_J\sigma_J^\transp)^{-1}\sigma_J\int_{t_k}^{t_{k+1}}\rmd W^J_s + \rme^{-\alpha\Delta_n}\Delta_nL^{(n)}_k(\sigma_J\sigma_J^\transp)^{-1}(\mu_{t_k}-\muhat{Z,n}{t_k-}).
	\end{aligned}
\end{equation}

\paragraph{Splitting the difference of $\boldsymbol{\muhat{J}{}}$ and $\boldsymbol{\muhat{Z,n}{}}$ into summands.}
Combining~\eqref{eq:second_recursion_muhat_D} with the above representation of $\muhat{Z,n}{t_{k+1}-}$ yields after a slight rearrangement of terms
\[ \muhat{J}{t_{k+1}}-\muhat{Z,n}{t_{k+1}-} = A^n+B^n+C^n+D^n+E^n+F^n, \]
where
\begingroup
\allowdisplaybreaks
	\begin{align*}
		A^n &= \rme^{-\alpha\Delta_n}(\muhat{J}{t_k}-\muhat{Z,n}{t_k-}), \\
		B^n &= \int_{t_k}^{t_{k+1}} \rme^{-\alpha(t_{k+1}-s)}(\gam{J}{s}-\gam{Z,n}{s})(\sigma_R\sigma_R^\transp)^{-1}\sigma_R\,\rmd V^Z_s, \\
		C^n &= \int_{t_k}^{t_{k+1}} \rme^{-\alpha(t_{k+1}-s)}\gam{J}{s}(\sigma_R\sigma_R^\transp)^{-1}(\muhat{Z,n}{s}-\muhat{J}{s})\,\rmd s, \\
		D^n &= \int_{t_k}^{t_{k+1}} \Bigl(\rme^{-\alpha(t_{k+1}-s)}\gam{J}{s}(\sigma_J\sigma_J^\transp)^{-1}-\rme^{-\alpha\Delta_n}L^{(n)}_k(\sigma_J\sigma_J^\transp)^{-1}\Bigr)\sigma_J\,\rmd W^J_s, \\
		E^n &= \int_{t_k}^{t_{k+1}} \rme^{-\alpha(t_{k+1}-s)}\gam{J}{s}(\sigma_J\sigma_J^\transp)^{-1}\mu_s\,\rmd s -\rme^{-\alpha\Delta_n}\Delta_nL^{(n)}_k(\sigma_J\sigma_J^\transp)^{-1}\mu_{t_k}, \\
		F^n &= \rme^{-\alpha\Delta_n}\Delta_nL^{(n)}_k(\sigma_J\sigma_J^\transp)^{-1}\muhat{Z,n}{t_k-}-\int_{t_k}^{t_{k+1}}\rme^{-\alpha(t_{k+1}-s)}\gam{J}{s}(\sigma_J\sigma_J^\transp)^{-1}\muhat{J}{s}\,\rmd s.
	\end{align*}
\endgroup

\paragraph{Application of the discrete Gronwall Lemma.}
The idea is now to apply the discrete Gronwall Lemma from Lemma~\ref{lem:discrete_gronwall} to the estimation
\begin{equation}\label{eq:estimation_for_gronwall}
	\begin{aligned}
		&\E\Bigl[\bigl\lVert \muhat{J}{t_{k+1}}-\muhat{Z,n}{t_{k+1}-}\bigr\rVert^2\Bigr] = \E\Bigl[\bigl\lVert A^n+B^n+C^n+D^n+E^n+F^n\bigr\rVert^2\Bigr] \\
		&\leq \E\Bigl[\bigl\lVert A^n\bigr\rVert^2\Bigr]+5\E\Bigl[\bigl\lVert B^n\bigr\rVert^2+\bigl\lVert C^n\bigr\rVert^2+\bigl\lVert D^n\bigr\rVert^2+\bigl\lVert E^n\bigr\rVert^2+\bigl\lVert F^n\bigr\rVert^2\Bigr]\\
		&\quad+2\E\Bigl[(A^n)^\transp(E^n+F^n)\Bigr].
	\end{aligned}
\end{equation}
In the inequality we have used that $(a_1+\dots+ a_p)^2\leq p(a_1^2+\dots+a_p^2)$, and the fact that $B^n+C^n+D^n$ can be written as a sum of stochastic integrals over Brownian motions between $t_k$ and $t_{k+1}$. Since $A^n=\rme^{-\alpha\Delta_n}(\muhat{J}{t_k}-\muhat{Z,n}{t_k-})$ is independent of these stochastic integrals, the term $\E[(A^n)^\transp(B^n+C^n+D^n)]$ vanishes.

\paragraph{Finding upper estimates for the single summands.}
We now show how to find upper estimates for the single summands in the decomposition above. First of all,
\[ \E\Bigl[\bigl\lVert A^n\bigr\rVert^2\Bigr]=\E\Bigl[\bigl\lVert \rme^{-\alpha\Delta_n}(\muhat{J}{t_k}-\muhat{Z,n}{t_k-})\bigr\rVert^2\Bigr]\leq \E\Bigl[\bigl\lVert \muhat{J}{t_k}-\muhat{Z,n}{t_k-}\bigr\rVert^2\Bigr] \]
by properties of the spectral norm and positive definiteness of $\alpha$.
By using the multidimensional It\^{o} isometry from Lemma~\ref{lem:ito_isometry_multidimensional} we deduce
\begin{equation*}
	\begin{aligned}
		\E\Bigl[\bigl\lVert B^n\bigr\rVert^2\Bigr] &\leq C_\textrm{norm} \E\biggl[\int_{t_k}^{t_{k+1}}\lVert \rme^{-\alpha(t_{k+1}-s)}(\gam{J}{s}-\gam{Z,n}{s})(\sigma_R\sigma_R^\transp)^{-1}\sigma_R\rVert^2\,\rmd s\biggr] \\
		&\leq C_\textrm{norm}\lVert(\sigma_R\sigma_R^\transp)^{-1}\sigma_R\rVert^2\int_{t_k}^{t_{k+1}}\lVert \gam{J}{s}-\gam{Z,n}{s}\rVert^2\,\rmd s \\
		&\leq C_\textrm{norm}\lVert(\sigma_R\sigma_R^\transp)^{-1}\sigma_R\rVert^2\int_{t_k}^{t_{k+1}}(K_Q\Delta_n)^2\,\rmd s =: C_B \Delta_n^3.
	\end{aligned}
\end{equation*}
Note that $\lVert \gam{J}{s}-\gam{Z,n}{s}\rVert\leq K_Q\Delta_n$ by Theorem~\ref{thm:q_C_n_goes_to_q_D}.

Now for the term $C^n$ we use the Cauchy--Schwarz inequality from Lemma~\ref{lem:cauchy_schwarz_multidimensional} to see that
\begin{equation*}
	\begin{aligned}
		\E\Bigl[\bigl\lVert C^n\bigr\rVert^2\Bigr] &= \E\biggl[ \Bigl\lVert \int_{t_k}^{t_{k+1}} \rme^{-\alpha(t_{k+1}-s)}\gam{J}{s}(\sigma_R\sigma_R^\transp)^{-1}(\muhat{Z,n}{s}-\muhat{J}{s})\,\rmd s \Bigr\rVert^2 \biggr] \\
		&\leq \Delta_n\int_{t_k}^{t_{k+1}} \E\Bigl[\bigl\lVert \rme^{-\alpha(t_{k+1}-s)}\gam{J}{s}(\sigma_R\sigma_R^\transp)^{-1}(\muhat{Z,n}{s}-\muhat{J}{s}) \bigr\rVert^2\Bigr]\,\rmd s \\
		&\leq \Delta_n C_{\gam{}{}}^2\lVert(\sigma_R\sigma_R^\transp)^{-1}\rVert^2 \int_{t_k}^{t_{k+1}} \E\Bigl[\bigl\lVert \muhat{Z,n}{s}-\muhat{J}{s} \bigr\rVert^2\Bigr]\,\rmd s.
	\end{aligned}
\end{equation*}
We then apply the mean value theorem for estimating the integral to see that
\begin{equation*}
	\begin{aligned}
		\int_{t_k}^{t_{k+1}} &\E\Bigl[\bigl\lVert \muhat{Z,n}{s}-\muhat{J}{s} \bigr\rVert^2\Bigr]\,\rmd s \leq \Delta_n\E\Bigl[\bigl\lVert \muhat{Z,n}{t_k}-\muhat{J}{t_k} \bigr\rVert^2\Bigr]+C_\textrm{mvt}\Delta_n^2 \\
		&\leq \Delta_n\Bigl(2\E\Bigl[\bigl\lVert \muhat{Z,n}{t_k-}-\muhat{J}{t_k} \bigr\rVert^2\Bigr] + 2\E\Bigl[\bigl\lVert \muhat{Z,n}{t_k}-\muhat{Z,n}{t_k-} \bigr\rVert^2\Bigr] \Bigr)+C_\textrm{mvt}\Delta_n^2.
	\end{aligned}
\end{equation*}
The jump size of $\muhat{Z,n}{}$ at an information date is bounded, hence all in all we obtain
\[ \E\Bigl[\bigl\lVert C^n\bigr\rVert^2\Bigr]\leq C_{C,1}\Delta_n^2\E\Bigl[\bigl\lVert \muhat{Z,n}{t_k-}-\muhat{J}{t_k} \bigr\rVert^2\Bigr] + C_{C,2}\Delta_n^2 \]
for constants $C_{C,1}$, $C_{C,2}>0$.

For the term $D^n$ we use again the multidimensional It\^{o} isometry from Lemma~\ref{lem:ito_isometry_multidimensional} and get
\begin{equation*}
	\begin{aligned}
		\E\Bigl[\bigl\lVert D^n\bigr\rVert^2\Bigr] &\leq C_\textrm{norm}\E\biggl[ \int_{t_k}^{t_{k+1}} \bigl\lVert\bigl(\rme^{-\alpha(t_{k+1}-s)}\gam{J}{s}-\rme^{-\alpha\Delta_n}L^{(n)}_k\bigr)(\sigma_J\sigma_J^\transp)^{-1}\sigma_J\bigr\rVert^2\,\rmd s\biggr] \\
		&\leq C_\textrm{norm}\lVert(\sigma_J\sigma_J^\transp)^{-1}\sigma_J\rVert^2\int_{t_k}^{t_{k+1}} \bigl\lVert\rme^{-\alpha(t_{k+1}-s)}\gam{J}{s}-\rme^{-\alpha\Delta_n}L^{(n)}_k\bigr\rVert^2\,\rmd s.
	\end{aligned}
\end{equation*}
For the integral above we first use a mean value theorem argument and then Lemma~\ref{lem:estimation_lemma} for the estimation of $\lVert\gam{J}{t_k}-L^{(n)}_k\rVert^2$ to obtain
\begin{equation*}
	\begin{aligned}
		\int_{t_k}^{t_{k+1}} &\bigl\lVert\rme^{-\alpha(t_{k+1}-s)}\gam{J}{s}-\rme^{-\alpha\Delta_n}L^{(n)}_k\bigr\rVert^2\,\rmd s
		\leq   \Delta_n\bigl\lVert\rme^{-\alpha\Delta_n}\gam{J}{t_k}-\rme^{-\alpha\Delta_n}L^{(n)}_k\bigr\rVert^2 +C_\textrm{mvt}\Delta_n^2 \\
		&\leq \Delta_n\lVert\gam{J}{t_k}-L^{(n)}_k\rVert^2 +C_\textrm{mvt}\Delta_n^2 \leq 2\Delta_n  \bigl(\lVert\gam{J}{t_k}-\gam{Z,n}{t_k-}\rVert^2+\bar{C}^2\Delta_n^2\bigr)+C_\textrm{mvt}\Delta_n^2.
	\end{aligned}
\end{equation*}
Putting these estimations together yields the existence of a constant $C_D>0$ such that
\[ \E\Bigl[\bigl\lVert D^n\bigr\rVert^2\Bigr] \leq C_D\Delta_n^2. \]
By writing the next summand $E^n$ as one integral, we can again apply the Cauchy--Schwarz inequality from Lemma~\ref{lem:cauchy_schwarz_multidimensional} and get
\begin{equation*}
	\begin{aligned}
		\E\Bigl[\bigl\lVert E^n\bigr\rVert^2\Bigr] &= \E\biggl[\biggl\lVert \int_{t_k}^{t_{k+1}} \Bigl(\rme^{-\alpha(t_{k+1}-s)}\gam{J}{s}(\sigma_J\sigma_J^\transp)^{-1}\mu_s -\rme^{-\alpha\Delta_n}L^{(n)}_k(\sigma_J\sigma_J^\transp)^{-1}\mu_{t_k}\Bigr)\rmd s \biggr\rVert^2\biggr] \\
		&\leq \Delta_n \int_{t_k}^{t_{k+1}} \E\Bigl[\Bigl\lVert \rme^{-\alpha(t_{k+1}-s)}\gam{J}{s}(\sigma_J\sigma_J^\transp)^{-1}\mu_s - \rme^{-\alpha \Delta_n}L^{(n)}_k(\sigma_J\sigma_J^\transp)^{-1}\mu_{t_k} \Bigr\rVert^2\Bigr]\rmd s.
	\end{aligned}
\end{equation*}
When using again the mean value theorem and the same argumentation as before we see that the integral is bounded by
\begin{equation*}
	\begin{aligned}
		&\Delta_n \E\Bigl[\Bigl\lVert \rme^{-\alpha\Delta_n}\bigl(\gam{J}{t_k}-L^{(n)}_k\bigr)(\sigma_J\sigma_J^\transp)^{-1}\mu_{t_k} \Bigr\rVert^2\Bigr] +C_\textrm{mvt}\Delta_n^2 \\
		&\leq \Delta_n\bigl\lVert\gam{J}{t_k}-L^{(n)}_k\bigr\rVert^2 \lVert(\sigma_J\sigma_J^\transp)^{-1}\rVert^2 \E[\lVert\mu_{t_k}\rVert^2] + C_\textrm{mvt}\Delta_n^2 \\
		&\leq \Delta_n C_{\mu}\lVert(\sigma_J\sigma_J^\transp)^{-1}\rVert^2\Bigl(2\lVert\gam{J}{t_k}-\gam{Z,n}{t_k-}\rVert^2+2\bar{C}^2\Delta_n^2\Bigr) + C_\textrm{mvt}\Delta_n^2.
	\end{aligned}
\end{equation*}
In conclusion, we have a constant $C_E>0$ with
\[ \E\Bigl[\bigl\lVert E^n\bigr\rVert^2\Bigr] \leq C_E\Delta_n^3. \]
In a similar way, $F^n$ can be treated. By first writing $F^n$ as a single integral and applying the Cauchy--Schwarz inequality from Lemma~\ref{lem:cauchy_schwarz_multidimensional} as well as the mean value theorem we get
\begin{equation*}
	\begin{aligned}
		\E\Bigl[\bigl\lVert F^n\bigr\rVert^2\Bigr] &= \E\biggl[\biggl\lVert \int_{t_k}^{t_{k+1}} \Bigl(\rme^{-\alpha\Delta_n}L^{(n)}_k(\sigma_J\sigma_J^\transp)^{-1}\muhat{Z,n}{t_k-} - \rme^{-\alpha(t_{k+1}-s)}\gam{J}{s}(\sigma_J\sigma_J^\transp)^{-1}\muhat{J}{s}\Bigr) \rmd s \biggr\rVert^2\biggr] \\
		&\leq \Delta_n \int_{t_k}^{t_{k+1}} \E\Bigl[\Bigl\lVert \rme^{-\alpha\Delta_n}L^{(n)}_k(\sigma_J\sigma_J^\transp)^{-1}\muhat{Z,n}{t_k-} - \rme^{-\alpha(t_{k+1}-s)}\gam{J}{s}(\sigma_J\sigma_J^\transp)^{-1}\muhat{J}{s} \Bigr\rVert^2\Bigr]\rmd s \\
		&\leq \Delta_n^2 \E\Bigl[\Bigl\lVert \rme^{-\alpha\Delta_n} \Bigl( L^{(n)}_k(\sigma_J\sigma_J^\transp)^{-1}\muhat{Z,n}{t_k-} - \gam{J}{t_k}(\sigma_J\sigma_J^\transp)^{-1}\muhat{J}{t_k} \Bigr) \Bigr\rVert^2\Bigr] +C_\textrm{mvt}\Delta_n^3.
	\end{aligned}
\end{equation*}
The expectation above is bounded by
\begin{equation*}
	\begin{aligned}
		&2\E\bigl[\bigl\lVert \bigl(L^{(n)}_k-\gam{J}{t_k}\bigr)(\sigma_J\sigma_J^\transp)^{-1}\muhat{Z,n}{t_k-}\bigr\rVert^2\bigr] +2\E\bigl[\bigl\lVert\gam{J}{t_k}(\sigma_J\sigma_J^\transp)^{-1}\bigl(\muhat{Z,n}{t_k-} -\muhat{J}{t_k}\bigr) \bigr\rVert^2\bigr] \\
		&\leq 2 \lVert(\sigma_J\sigma_J^\transp)^{-1}\rVert^2\E\bigl[\lVert\muhat{Z,n}{t_k-}\rVert^2\bigr] \bigl\lVert L^{(n)}_k-\gam{J}{t_k}\bigr\rVert^2+2C_{\gam{}{}}^2\lVert(\sigma_J\sigma_J^\transp)^{-1}\rVert^2 \E\bigl[\lVert\muhat{Z,n}{t_k-}-\muhat{J}{t_k}\rVert^2\bigr].
	\end{aligned}
\end{equation*}
By the same reasons as in the calculations above we obtain all in all that there exist constants $C_{F,1}$ and $C_{F,2}>0$ such that
\[ \E\Bigl[\bigl\lVert F^n\bigr\rVert^2\Bigr] \leq C_{F,1}\Delta_n^2\E\bigl[\lVert\muhat{Z,n}{t_k-}-\muhat{J}{t_k}\rVert^2\bigr]+C_{F,2}\Delta_n^3. \]
We have now found upper bounds for all quadratic terms in~\eqref{eq:estimation_for_gronwall}. Only the mixed terms $(A^n)^\transp E^n$ and $(A^n)^\transp F^n$ remain to be considered.
Firstly, we again rewrite $E^n$ as one integral
\[ E^n=\int_{t_k}^{t_{k+1}} \Bigl(\rme^{-\alpha(t_{k+1}-s)}\gam{J}{s}(\sigma_J\sigma_J^\transp)^{-1}\mu_s -\rme^{-\alpha\Delta_n}L^{(n)}_k(\sigma_J\sigma_J^\transp)^{-1}\mu_{t_k}\Bigr)\rmd s. \]
We see that
\begin{equation*}
	\begin{aligned}
		& \E\bigl[(A^n)^\transp E^n\bigr] \\
		&= \int_{t_k}^{t_{k+1}} \E\bigl[(\muhat{J}{t_k}-\muhat{Z,n}{t_k-})^\transp\rme^{-\alpha\Delta_n}\bigl(\rme^{-\alpha(t_{k+1}-s)}\gam{J}{s}(\sigma_J\sigma_J^\transp)^{-1}\mu_s -\rme^{-\alpha\Delta_n}L^{(n)}_k(\sigma_J\sigma_J^\transp)^{-1}\mu_{t_k}\bigr)\bigr]\,\rmd s \\
		&= \int_{t_k}^{t_{k+1}} \E\bigl[(\muhat{J}{t_k}-\muhat{Z,n}{t_k-})^\transp\rme^{-\alpha\Delta_n}\rme^{-\alpha(t_{k+1}-s)}\gam{J}{s}(\sigma_J\sigma_J^\transp)^{-1}\mu_s\bigr]\,\rmd s \\
		&\qquad\qquad- \E\bigl[(\muhat{J}{t_k}-\muhat{Z,n}{t_k-})^\transp\rme^{-2\alpha\Delta_n} \Delta_nL^{(n)}_k(\sigma_J\sigma_J^\transp)^{-1}\mu_{t_k}\bigr].
	\end{aligned}
\end{equation*}
By using the mean value theorem and sublinearity of the spectral norm we obtain
\begin{equation*}
	\begin{aligned}
		\bigl\lvert\E\bigl[(A^n)^\transp E^n\bigr]\bigr\rvert&\leq \Bigl\lvert \Delta_n\E\bigl[(\muhat{J}{t_k}-\muhat{Z,n}{t_k-})^\transp\rme^{-2\alpha\Delta_n}\gam{J}{t_k}(\sigma_J\sigma_J^\transp)^{-1}\mu_{t_k}\bigr] \\
		&\quad- \E\bigl[(\muhat{J}{t_k}-\muhat{Z,n}{t_k-})^\transp\rme^{-2\alpha\Delta_n} \Delta_nL^{(n)}_k(\sigma_J\sigma_J^\transp)^{-1}\mu_{t_k}\bigr] \Bigr\rvert + C_\textrm{mvt}\Delta_n^2 \\
		&= \Delta_n \Bigl\lvert \E\Bigl[(\muhat{J}{t_k}-\muhat{Z,n}{t_k-})^\transp\rme^{-2\alpha\Delta_n} \bigl( \gam{J}{t_k}-L^{(n)}_k \bigr)(\sigma_J\sigma_J^\transp)^{-1}\mu_{t_k}\Bigr] \Bigr\rvert + C_\textrm{mvt}\Delta_n^2 \\
		&\leq \Delta_n\bigl\lVert(\sigma_J\sigma_J^\transp)^{-1}\bigr\rVert \bigl\lVert \gam{J}{t_k}-L^{(n)}_k \bigr\rVert \E\Bigl[\bigl\lVert\muhat{J}{t_k}-\muhat{Z,n}{t_k-}\bigr\rVert\bigl\lVert\mu_{t_k}\bigr\rVert\Bigr] + C_\textrm{mvt}\Delta_n^2 \\
		&\leq C_{A,E}\Delta_n^2.
	\end{aligned}
\end{equation*}
The last inequality is due to boundedness of $\E[\lVert\muhat{J}{t_k}-\muhat{Z,n}{t_k-}\rVert\lVert\mu_{t_k}\rVert]$ together with the fact that $\lVert\gam{J}{t_k}-L^{(n)}_k\rVert$ is bounded by a constant times $\Delta_n$, see Lemma~\ref{lem:estimation_lemma}.

The mixed term $(A^n)^\transp F^n$ can be handled in a similar way. It holds that
\begin{equation*}
	\begin{aligned}
		(A^n)^\transp F^n &= (\muhat{J}{t_k}-\muhat{Z,n}{t_k-})^\transp\rme^{-2\alpha\Delta_n}\Delta_nL^{(n)}_k(\sigma_J\sigma_J^\transp)^{-1}\muhat{Z,n}{t_k-} \\
		&\quad-\int_{t_k}^{t_{k+1}}(\muhat{J}{t_k}-\muhat{Z,n}{t_k-})^\transp\rme^{-\alpha\Delta_n}\rme^{-\alpha(t_{k+1}-s)}\gam{J}{s}(\sigma_J\sigma_J^\transp)^{-1}\muhat{J}{s}\,\rmd s
	\end{aligned}
\end{equation*}
and hence by another application of the mean value theorem
\begin{equation*}
	\begin{aligned}
		&\bigl\lvert \E\bigl[(A^n)^\transp F^n\bigr] \bigr\rvert \\
		&\leq \biggl\lvert \E\Bigl[(\muhat{J}{t_k}-\muhat{Z,n}{t_k-})^\transp\rme^{-2\alpha\Delta_n} \Delta_nL^{(n)}_k(\sigma_J\sigma_J^\transp)^{-1}\muhat{Z,n}{t_k-}\Bigr] \\
		&\quad- \Delta_n\E\bigl[(\muhat{J}{t_k}-\muhat{Z,n}{t_k-})^\transp\rme^{-2\alpha\Delta_n}\gam{J}{t_k}(\sigma_J\sigma_J^\transp)^{-1}\muhat{J}{t_k}\bigr] \biggr\rvert + C_\textrm{mvt}\Delta_n^2 \\
		&= \Delta_n\biggl\lvert \E\Bigl[(\muhat{J}{t_k}-\muhat{Z,n}{t_k-})^\transp\rme^{-2\alpha\Delta_n} \Bigl(L^{(n)}_k(\sigma_J\sigma_J^\transp)^{-1}\muhat{Z,n}{t_k-}-\gam{J}{t_k}(\sigma_J\sigma_J^\transp)^{-1}\muhat{J}{t_k}\Bigr) \Bigr] \biggr\rvert + C_\textrm{mvt}\Delta_n^2.
	\end{aligned}
\end{equation*}
The absolute value of the expectation is split into two summands as
\begin{equation*}
	\begin{aligned}
		&\Bigl\lvert \E\Bigl[(\muhat{J}{t_k}-\muhat{Z,n}{t_k-})^\transp\rme^{-2\alpha\Delta_n} \Bigl(L^{(n)}_k(\sigma_J\sigma_J^\transp)^{-1}\muhat{Z,n}{t_k-}-\gam{J}{t_k}(\sigma_J\sigma_J^\transp)^{-1}\muhat{J}{t_k}\Bigr) \Bigr] \Bigr\rvert \\
		&\leq  \Bigl\lvert\E\Bigl[(\muhat{J}{t_k}-\muhat{Z,n}{t_k-})^\transp\rme^{-2\alpha\Delta_n} \bigl(L^{(n)}_k-\gam{J}{t_k}\bigr)(\sigma_J\sigma_J^\transp)^{-1}\muhat{Z,n}{t_k-} \Bigr]\Bigr\rvert \\
		&\quad+ \Bigl\lvert\E\Bigl[(\muhat{J}{t_k}-\muhat{Z,n}{t_k-})^\transp\rme^{-2\alpha\Delta_n} \gam{J}{t_k}(\sigma_J\sigma_J^\transp)^{-1}\bigl(\muhat{Z,n}{t_k-}-\muhat{J}{t_k}\bigr) \Bigr]\Bigr\rvert \\
		&\leq \bigl\lVert(\sigma_J\sigma_J^\transp)^{-1}\bigr\rVert\biggl(\E\Bigl[\bigl\lVert\muhat{J}{t_k}-\muhat{Z,n}{t_k-}\bigr\rVert\bigl\lVert\muhat{Z,n}{t_k-}\bigr\rVert\Bigr]\bigl\lVert L^{(n)}_k-\gam{J}{t_k} \bigr\rVert + C_{\gam{}{}} \E\Bigl[\bigl\lVert\muhat{J}{t_k}-\muhat{Z,n}{t_k-}\bigr\rVert^2\Bigr]\biggr).
	\end{aligned}
\end{equation*}
From the same argumentations as above we deduce that there exist constants $C_{A,F,1}$ and $C_{A,F,2}>0$ with
\[ \bigl\lvert \E\bigl[(A^n)^\transp F^n\bigr] \bigr\rvert \leq C_{A,F,1}\Delta_n\E\bigl[\bigl\lVert\muhat{J}{t_k}-\muhat{Z,n}{t_k-}\bigr\rVert^2\bigr] + C_{A,F,2}\Delta_n^2. \]

\paragraph{Conclusion with discrete Gronwall Lemma.}
Now we plug all these upper bounds into~\eqref{eq:estimation_for_gronwall} and obtain that there exist constants $L_1,L_2>0$ such that
\[ \E\Bigl[\bigl\lVert \muhat{J}{t_{k+1}}-\muhat{Z,n}{t_{k+1}-}\bigr\rVert^2\Bigr] \leq \bigl(1+L_1\Delta_n\bigr) \E\Bigl[\bigl\lVert \muhat{J}{t_k}-\muhat{Z,n}{t_k-}\bigr\rVert^2\Bigr]+L_2\Delta_n^2. \]
Setting $a_k:=\E\Bigl[\bigl\lVert \muhat{J}{t_k}-\muhat{Z,n}{t_k-}\bigr\rVert^2\Bigr]$ in the discrete version of Gronwall's Lemma, see Lemma~\ref{lem:discrete_gronwall}, we can conclude that
\[ \E\Bigl[\bigl\lVert \muhat{J}{t_k}-\muhat{Z,n}{t_k-}\bigr\rVert^2\Bigr]\leq \frac{\rme^{L_1T}-1}{L_1}L_2\Delta_n=:\tilde{C}\Delta_n \]
which proves the claim for $t=t_k$.
To find an upper bound that is valid for arbitrary time $t\in[0,T]$ with $t\in[t_k,t_{k+1})$, we observe that
\begin{equation*}
	\begin{aligned}
		\E\Bigl[\bigl\lVert \muhat{J}{t}&-\muhat{Z,n}{t}\bigr\rVert^2\Bigr] = \E\Bigl[\bigl\lVert \muhat{J}{t}-\muhat{J}{t_k}+\muhat{J}{t_k}-\muhat{Z,n}{t_k-}+\muhat{Z,n}{t_k-}-\muhat{Z,n}{t}\bigr\rVert^2\Bigr] \\
		&\leq 3\Bigl( \E\Bigl[\bigl\lVert \muhat{J}{t}-\muhat{J}{t_k}\bigr\rVert^2\Bigr]+\E\Bigl[\bigl\lVert\muhat{J}{t_k}-\muhat{Z,n}{t_k-}\bigr\rVert^2\Bigr]+\E\Bigl[\bigl\lVert\muhat{Z,n}{t_k-}-\muhat{Z,n}{t}\bigr\rVert^2\Bigr] \Bigr).
	\end{aligned}
\end{equation*}
The first summand is bounded by a constant times $\Delta_n$ which can be seen from the representation in Lemma~\ref{lem:filter_D_dynamics}. From~\eqref{eq:update_step_muhat_C_N} we can deduce the same for the third summand. Hence, all in all there exists a constant $K_{m,2}>0$ such that
\[ \E\Bigl[\bigl\lVert \muhat{Z,n}{t}-\muhat{J}{t}\bigr\rVert^2\Bigr]\leq K_{m,2}\Delta_n, \]
which proves the claim of the theorem for $p=2$.

For proving the claim in the case $p\neq 2$ note that the joint distribution of the conditional means is Gaussian. A classical result, see for example Rosi\'{n}ski and Suchanecki~\cite[Lem.~2.1]{rosinski_suchanecki_1980}, hence yields that there is a constant $C_{p,2}>0$ with
\[ \E\Bigl[\bigl\lVert \muhat{Z,n}{t}-\muhat{J}{t}\bigr\rVert^p\Bigr] \leq C_{p,2}\E\Bigl[\bigl\lVert \muhat{Z,n}{t}-\muhat{J}{t}\bigr\rVert^2\Bigr]^{\frac{p}{2}} \leq C_{p,2}(K_{m,2}\Delta_n)^{\frac{p}{2}}=K_{m,p}\Delta_n^{p/2} \]
for all $t\in[0,T]$. This concludes the proof in the case $p\neq 2$.\qed

\section{Proofs for Random Information Dates}\label{app:long_proofs_random_times}

\subsection{Proof of Theorem~\ref{thm:q_C_lambda_goes_to_q_D}: Convergence of Covariance Matrices}

We first consider $p=2$.
Using the representations from Proposition~\ref{prop:integral_equations_for_q_D_and_q_C_lambda} we see
\begin{equation*}
	\begin{aligned}
		&\gam{Z,\lambda}{t}-\gam{J}{t} = \int_0^t\int_{\R^d}-\gam{Z,\lambda}{s-}(\gam{Z,\lambda}{s-}+\lambda\sigma_J\sigma_J^\transp)^{-1}\gam{Z,\lambda}{s-}\,\tilde{N}(\rmd s,\rmd u)\\
		&\quad+ \int_0^t \Bigl(L(\gam{Z,\lambda}{s})-L(\gam{J}{s}) -\lambda\gam{Z,\lambda}{s-}(\gam{Z,\lambda}{s-}+\lambda\sigma_J\sigma_J^\transp)^{-1}\gam{Z,\lambda}{s-}+\gam{J}{s}(\sigma_J\sigma_J^\transp)^{-1}\gam{J}{s}\Bigr)\rmd s.
	\end{aligned}
\end{equation*}
Denote the integral with respect to the compensated measure $\tilde{N}$ by $X^{\lambda}_t$ and the second one by $A^{\lambda}_t$.
Now for $r\in[0,T]$ it holds
\begin{equation}\label{eq:two_summands_for_u_lambda}
	\begin{aligned}
		u^{\lambda}_r:=\E\biggl[\sup_{t\leq r}\, \lVert\gam{Z,\lambda}{t}-\gam{J}{t}\rVert^2\biggr]
		\leq 2\E\biggl[\sup_{t\leq r}\,\lVert X^{\lambda}_t\rVert^2\biggr]+2\E\biggl[\sup_{t\leq r}\,\lVert A^{\lambda}_t\rVert^2\biggr].
	\end{aligned}
\end{equation}

\paragraph{Estimate for the martingale term $\boldsymbol{X^\lambda}$.}
Every component of the matrix-valued process $(X^{\lambda}_t)_{t\geq 0}$ is a martingale since we integrate a bounded integrand with respect to the compensated measure $\tilde{N}$. In the following, for finding an upper bound for the term involving $X^\lambda_t$ in~\eqref{eq:two_summands_for_u_lambda} we first use Doob's inequality for martingales to get rid of the supremum. In a second step we can calculate the second moment of the integral because we know the corresponding intensity measure of the Poisson random measure. In detail, we proceed as follows. By equivalence of norms there is a constant $C_\textrm{norm}>0$ such that
\begin{equation}\label{eq:square_of_martingale_part}
	\begin{aligned}
		\E\biggl[\sup_{t\leq r} \,\lVert X^{\lambda}_t\rVert^2\biggr]
		&\leq C_\textrm{norm}\E\biggl[\sup_{t\leq r}\,\lVert X^{\lambda}_t\rVert_F^2\biggr]
		= C_\textrm{norm}\E\biggl[\sup_{t\leq r}\,\sum_{i,j=1}^d (X^{\lambda}_t(i,j))^2\biggr] \\
		&\leq C_\textrm{norm}\sum_{i,j=1}^d\E\biggl[\sup_{t\leq r}\, (X^{\lambda}_t(i,j))^2\biggr]
		\leq C_\textrm{norm}\sum_{i,j=1}^d 4\E\Bigl[(X^{\lambda}_r(i,j))^2\Bigr].
	\end{aligned}
\end{equation}
The last inequality follows from Doob's inequality for martingales. Next, we can apply Lemma~\ref{lem:variance_of_integral_with_respect_to_compensated_measure} to the definition of $X^\lambda$ and get
\begin{equation*}
	\begin{aligned}
		\E\Bigl[(X^{\lambda}_r(i,j))^2\Bigr] &= \E\biggl[\int_0^r\int_{\R^d} \Bigl(\bigl(-\gam{Z,\lambda}{s-}(\gam{Z,\lambda}{s-}+\lambda\sigma_J\sigma_J^\transp)^{-1}\gam{Z,\lambda}{s-}\bigr)(i,j)\Bigr)^2 \lambda\varphi(u)\,\rmd u\,\rmd s\biggr] \\
		&= \lambda\E\biggl[\int_0^r \Bigl(\bigl(-\gam{Z,\lambda}{s-}(\gam{Z,\lambda}{s-}+\lambda\sigma_J\sigma_J^\transp)^{-1}\gam{Z,\lambda}{s-}\bigr)(i,j)\Bigr)^2 \rmd s\biggr],
	\end{aligned}
\end{equation*}
using that the integrand does not depend on $u$ and $\varphi$ is a density.
Plugging back into~\eqref{eq:square_of_martingale_part}, we get, again by equivalence of norms,
\begin{equation}\label{eq:simplified_integral}
	\begin{aligned}
		\E\biggl[\sup_{t\leq r}\,\lVert X^{\lambda}_t\rVert^2\biggr]
		&\leq 4C_\textrm{norm}^2\lambda\int_0^r \E\Bigl[\lVert-\gam{Z,\lambda}{s-}(\gam{Z,\lambda}{s-}+\lambda\sigma_J\sigma_J^\transp)^{-1}\gam{Z,\lambda}{s-}\rVert^2\Bigr] \rmd s.
	\end{aligned}
\end{equation}
Since the norm of the matrices $\gam{Z,\lambda}{}$ is bounded by $C_{\gam{}{}}$, see Lemma~\ref{lem:boundedness_of_covariances}, we obtain
\begin{equation}\label{eq:norm_of_update}
	\begin{aligned}
		&\E\Bigl[\lVert-\gam{Z,\lambda}{s-}(\gam{Z,\lambda}{s-}+\lambda\sigma_J\sigma_J^\transp)^{-1}\gam{Z,\lambda}{s-}\rVert^2\Bigr] \leq C_{\gam{}{}}^4\E\Bigl[\lVert(\gam{Z,\lambda}{s-}+\lambda\sigma_J\sigma_J^\transp)^{-1}\rVert^2\Bigr] \\
		&=C_{\gam{}{}}^4 \E\Bigl[\bigl(\lmin(\gam{Z,\lambda}{s-}+\lambda\sigma_J\sigma_J^\transp)\bigr)^{-2}\Bigr] \leq C_{\gam{}{}}^4\E\Bigl[\bigl(\lmin(\lambda\sigma_J\sigma_J^\transp)\bigr)^{-2}\Bigr]
		= \frac{C_{\gam{}{}}^4}{\lambda^2}\lVert(\sigma_J\sigma_J^\transp)^{-1}\rVert^2.
	\end{aligned}
\end{equation}
When reinserting this upper bound into~\eqref{eq:simplified_integral}, we can conclude that
\begin{equation}\label{eq:upper_bound_for_M}
	\begin{aligned}
		\E\biggl[\sup_{t\leq r}\,\lVert X^{\lambda}_t\rVert^2\biggr]&\leq \frac{1}{\lambda}4 C_\textrm{norm}^2C_{\gam{}{}}^4\lVert(\sigma_J\sigma_J^\transp)^{-1}\rVert^2r \leq \frac{1}{\lambda}4 C_\textrm{norm}^2C_{\gam{}{}}^4\lVert(\sigma_J\sigma_J^\transp)^{-1}\rVert^2T.
	\end{aligned}
\end{equation}

\paragraph{Estimate for the finite variation term $\boldsymbol{A^\lambda}$.}
Using the short-hand notation $g$ for the integrand of $A^{\lambda}_t$ we get
\begin{equation}\label{eq:cauchy_schwarz_for_A}
	\sup_{t\leq r}\,\lVert A^{\lambda}_t\rVert^2 = \sup_{t\leq r}\,\biggl\lVert\int_0^t g(s)\,\rmd s\biggr\rVert^2 \leq \sup_{t\leq r}\, t \int_0^t \lVert g(s)\rVert^2\,\rmd s \leq r\int_0^r \lVert g(s)\rVert^2\,\rmd s
\end{equation}
by the Cauchy--Schwarz inequality in Lemma~\ref{lem:cauchy_schwarz_multidimensional}. We now address the integrand of $A^{\lambda}$. Since
\begin{equation*}
	\begin{aligned}
		\lVert g(s)\rVert &\leq 2\lVert\gam{Z,\lambda}{s}-\gam{J}{s}\rVert\bigl(\lVert\alpha\rVert+C_{\gam{}{}}\lVert(\sigma_R\sigma_R^\transp)^{-1}\rVert\bigr) \\
		&\quad +\lVert\lambda\gam{Z,\lambda}{s-}(\gam{Z,\lambda}{s-}+\lambda\sigma_J\sigma_J^\transp)^{-1}\gam{Z,\lambda}{s-}-\gam{J}{s}(\sigma_J\sigma_J^\transp)^{-1}\gam{J}{s}\rVert
	\end{aligned}
\end{equation*}
we obtain
\begin{equation}\label{eq:two_summands_for_A}
	\begin{aligned}
		\E\biggl[\sup_{t\leq r}\,\lVert A^{\lambda}_t\rVert^2\biggr]&\leq r\int_0^r 8\bigl(\lVert\alpha\rVert+C_{\gam{}{}}\lVert(\sigma_R\sigma_R^\transp)^{-1}\rVert\bigr)^2\E\bigl[\lVert\gam{Z,\lambda}{s}-\gam{J}{s}\rVert^2\bigr]\,\rmd s \\
		&\quad + r\int_0^r 2\E\bigl[\lVert\lambda\gam{Z,\lambda}{s-}(\gam{Z,\lambda}{s-}+\lambda\sigma_J\sigma_J^\transp)^{-1}\gam{Z,\lambda}{s-}-\gam{J}{s}(\sigma_J\sigma_J^\transp)^{-1}\gam{J}{s}\rVert^2\bigr]\,\rmd s \\
		&\leq 8T\bigl(\lVert\alpha\rVert+C_{\gam{}{}}\lVert(\sigma_R\sigma_R^\transp)^{-1}\rVert\bigr)^2\int_0^r u^{\lambda}_s\,\rmd s \\
		&\quad + 2T\int_0^r \E\bigl[\lVert\lambda\gam{Z,\lambda}{s-}(\gam{Z,\lambda}{s-}+\lambda\sigma_J\sigma_J^\transp)^{-1}\gam{Z,\lambda}{s-}-\gam{J}{s}(\sigma_J\sigma_J^\transp)^{-1}\gam{J}{s}\rVert^2\bigr]\,\rmd s.
	\end{aligned}
\end{equation}
We analyze the second summand in more detail. For that purpose, we decompose
\begin{equation}\label{eq:integrand_decomposition}
	\begin{aligned}
		&\lambda\gam{Z,\lambda}{s-}(\gam{Z,\lambda}{s-}+\lambda\sigma_J\sigma_J^\transp)^{-1}\gam{Z,\lambda}{s-}-\gam{J}{s}(\sigma_J\sigma_J^\transp)^{-1}\gam{J}{s} \\
		&= \bigl(\lambda\gam{Z,\lambda}{s-}(\gam{Z,\lambda}{s-}+\lambda\sigma_J\sigma_J^\transp)^{-1}\gam{Z,\lambda}{s-}-\gam{Z,\lambda}{s-}(\sigma_J\sigma_J^\transp)^{-1}\gam{Z,\lambda}{s-}\bigr)\\
		&\quad+ \bigl(\gam{Z,\lambda}{s-}(\sigma_J\sigma_J^\transp)^{-1}\gam{Z,\lambda}{s-}-\gam{Z,\lambda}{s}(\sigma_J\sigma_J^\transp)^{-1}\gam{Z,\lambda}{s}\bigr)\\
		&\quad+ \bigl(\gam{Z,\lambda}{s}(\sigma_J\sigma_J^\transp)^{-1}\gam{Z,\lambda}{s}-\gam{J}{s}(\sigma_J\sigma_J^\transp)^{-1}\gam{J}{s}\bigr)
	\end{aligned}
\end{equation}
and find upper bounds for the three summands. For the first summand we find
\begin{equation}\label{eq:split_up_first_summand}
	\begin{aligned}
		&\E\Bigl[\lVert\lambda\gam{Z,\lambda}{s-}(\gam{Z,\lambda}{s-}+\lambda\sigma_J\sigma_J^\transp)^{-1}\gam{Z,\lambda}{s-}-\gam{Z,\lambda}{s-}(\sigma_J\sigma_J^\transp)^{-1}\gam{Z,\lambda}{s-}\rVert^2\Bigr] \\
		&=\E\Bigl[\lVert \gam{Z,\lambda}{s-}\Bigl((\gam{Z,\lambda}{s-}+\lambda\sigma_J\sigma_J^\transp)^{-1}(\lambda\sigma_J\sigma_J^\transp-\gam{Z,\lambda}{s-}-\lambda\sigma_J\sigma_J^\transp)\Bigr)(\sigma_J\sigma_J^\transp)^{-1}\gam{Z,\lambda}{s-} \rVert^2\Bigr] \\
		&=\E\Bigl[\lVert -\gam{Z,\lambda}{s-}(\gam{Z,\lambda}{s-}+\lambda\sigma_J\sigma_J^\transp)^{-1}\gam{Z,\lambda}{s-}(\sigma_J\sigma_J^\transp)^{-1}\gam{Z,\lambda}{s-} \rVert^2\Bigr] \\
		&\leq C_{\gam{}{}}^2\lVert(\sigma_J\sigma_J^\transp)^{-1}\rVert^2\frac{1}{\lambda^2}C_{\gam{}{}}^4\lVert(\sigma_J\sigma_J^\transp)^{-1}\rVert^2 = \frac{1}{\lambda^2}C_{\gam{}{}}^6\lVert(\sigma_J\sigma_J^\transp)^{-1}\rVert^4.
	\end{aligned}
\end{equation}
For the second summand note that $\E\bigl[\lVert\gam{Z,\lambda}{s-}(\sigma_J\sigma_J^\transp)^{-1}\gam{Z,\lambda}{s-}-\gam{Z,\lambda}{s}(\sigma_J\sigma_J^\transp)^{-1}\gam{Z,\lambda}{s}\rVert^2\bigr]$ is equal to zero since a jump at time $s$ occurs with probability zero.
For the third summand we observe
\begin{equation}\label{eq:split_up_third_summand}
	\begin{aligned}
		&\E\bigl[\lVert\gam{Z,\lambda}{s}(\sigma_J\sigma_J^\transp)^{-1}\gam{Z,\lambda}{s}-\gam{J}{s}(\sigma_J\sigma_J^\transp)^{-1}\gam{J}{s}\rVert^2\bigr] \\
		&= \E\bigl[\lVert\gam{Z,\lambda}{s}(\sigma_J\sigma_J^\transp)^{-1}(\gam{Z,\lambda}{s}-\gam{J}{s})+(\gam{Z,\lambda}{s}-\gam{J}{s})(\sigma_J\sigma_J^\transp)^{-1}\gam{J}{s}\rVert^2\bigr] \\
		&\leq \bigl(2C_{\gam{}{}}\lVert(\sigma_J\sigma_J^\transp)^{-1}\rVert\bigr)^2\E\bigl[\lVert\gam{Z,\lambda}{s}-\gam{J}{s}\rVert^2\bigr]
		\leq 4C_{\gam{}{}}^2\lVert(\sigma_J\sigma_J^\transp)^{-1}\rVert^2 u^{\lambda}_s.
	\end{aligned}
\end{equation}
We now use these upper bounds in~\eqref{eq:integrand_decomposition} and obtain
\begin{equation}
	\begin{aligned}
		&\E\bigl[\lVert\lambda\gam{Z,\lambda}{s-}(\gam{Z,\lambda}{s-}+\lambda\sigma_J\sigma_J^\transp)^{-1}\gam{Z,\lambda}{s-}-\gam{J}{s}(\sigma_J\sigma_J^\transp)^{-1}\gam{J}{s}\rVert^2\bigr] \\
		&\leq \frac{3}{\lambda^2}C_{\gam{}{}}^6\lVert(\sigma_J\sigma_J^\transp)^{-1}\rVert^4+12C_{\gam{}{}}^2\lVert(\sigma_J\sigma_J^\transp)^{-1}\rVert^2 u^{\lambda}_s.
	\end{aligned}
\end{equation}
Hence we can write
\begin{equation}\label{eq:upper_bound_for_A}
	\begin{aligned}
		\E\biggl[\sup_{t\leq r}\,\lVert A^{\lambda}_t\rVert^2\biggr]
		&\leq 8T\Bigl(\bigl(\lVert\alpha\rVert+C_{\gam{}{}}\lVert(\sigma_R\sigma_R^\transp)^{-1}\rVert\bigr)^2+3C_{\gam{}{}}^2\lVert(\sigma_J\sigma_J^\transp)^{-1}\rVert^2\Bigr)\int_0^r u^{\lambda}_s\,\rmd s \\
		&\quad+ 6T^2C_{\gam{}{}}^6\lVert(\sigma_J\sigma_J^\transp)^{-1}\rVert^4\frac{1}{\lambda^2}.
	\end{aligned}
\end{equation}

\paragraph{Conclusion with Gronwall's Lemma.}
We have found upper bounds for both summands from~\eqref{eq:two_summands_for_u_lambda}. Plugging in yields constants $C_1$, $C_2>0$ such that
\begin{equation*}
	\begin{aligned}
		u^{\lambda}_r &\leq \frac{C_1}{\lambda}+C_2\int_0^r u^{\lambda}_s\,\rmd s
	\end{aligned}
\end{equation*}
for all $\lambda\geq 1$. By Gronwall's Lemma in integral form, see Lemma~\ref{lem:gronwall}, it follows
\begin{equation}
	\E\biggl[\sup_{t\leq T}\, \lVert\gam{Z,\lambda}{t}-\gam{J}{t}\rVert^2\biggr] =u^{\lambda}_T\leq\frac{C_1}{\lambda}\rme^{C_2 T}=\frac{\widetilde{K}_{Q,2}}{\lambda}
\end{equation}
for $\widetilde{K}_{Q,2}=C_1\rme^{C_2T}$, which proves the claim for $p=2$. For $p<2$ we use Lyapunov's inequality to get
\[ \E\Bigl[\bigl\lVert\gam{Z,\lambda}{t}-\gam{J}{t}\bigr\rVert^p\Bigr] \leq \E\Bigl[\bigl\lVert\gam{Z,\lambda}{t}-\gam{J}{t}\bigr\rVert^2\Bigr]^\frac{p}{2} \leq \Bigl(\frac{\widetilde{K}_{Q,2}}{\lambda}\Bigr)^\frac{p}{2}=\frac{\widetilde{K}_{Q,p}}{\lambda^{p/2}}. \]
For $p>2$ it holds
\begin{equation*}
	\begin{aligned}
		\E\Bigl[\bigl\lVert\gam{Z,\lambda}{t}-\gam{J}{t}\bigr\rVert^p\Bigr]
		\leq (2C_Q)^{p-2}\E\Bigl[\bigl\lVert\gam{Z,\lambda}{t}-\gam{J}{t}\bigr\rVert^2\Bigr] \leq (2C_Q)^{p-2}\frac{\widetilde{K}_{Q,2}}{\lambda}=\frac{\widetilde{K}_{Q,p}}{\lambda}
	\end{aligned}
\end{equation*}
due to boundedness of the conditional covariance matrices, see Lemma~\ref{lem:boundedness_of_covariances}.\qed

\subsection{Proof of Theorem~\ref{thm:muhat_C_lambda_goes_to_muhat_D}: Convergence of Conditional Means}

Throughout the proof, we omit the superscript $\lambda$ at time points $T^{(\lambda)}_k$ and at the Poisson process $(N^{(\lambda)}_t)_{t\geq 0}$ for better readability.

We first prove the claim for $p\geq 2$. The proof uses again Gronwall's Lemma, see Lemma~\ref{lem:gronwall}. Define $v^\lambda_t := \E[\lVert \muhat{Z,\lambda}{t}-\muhat{J}{t}\rVert^p]$ for $t\in[0,T]$. 
The filtering equations from Lemma~\ref{lem:filter_C_dynamics} yield
\begin{equation}\label{eq:representation_of_muhat_C_lambda}
	\muhat{Z,\lambda}{t} = \int_0^t \alpha(\delta-\muhat{Z,\lambda}{s})\,\rmd s + \int_0^t\gam{Z,\lambda}{s}(\sigma_R\sigma_R^\transp)^{-1}\sigma_R\,\rmd V^Z_s + \sum_{k=1}^{N_t} \frac{1}{\lambda}P^\lambda_k\bigl(Z^{(\lambda)}_k-\muhat{Z,\lambda}{T_k-}\bigr),
\end{equation}
where $\rmd R_s-\muhat{Z,\lambda}{s}\,\rmd s=\sigma_R\,\rmd V^Z_s$ defines the innovations process $V^Z$, an $m$-dimensional $\calF^{Z,\lambda}$-Brownian motion, and where
\[ P^\lambda_k=\lambda\bigl(I_d-\rho^{(\lambda)}(\gam{Z,\lambda}{T_k-})\bigr) = \lambda\gam{Z,\lambda}{T_k-}(\gam{Z,\lambda}{T_k-}+\lambda\sigma_J\sigma_J^\transp)^{-1}. \]
Note that the matrices $P^\lambda_k$ are bounded since
\begin{equation*}
	\begin{aligned}
		\lVert P^\lambda_k\rVert = \bigl\lVert\lambda\gam{Z,\lambda}{T_k-}(\gam{Z,\lambda}{T_k-}+\lambda\sigma_J\sigma_J^\transp)^{-1}\bigr\rVert
		&= \bigl\lVert\gam{Z,\lambda}{T_k-}(\gam{Z,\lambda}{T_k-}+\lambda\sigma_J\sigma_J^\transp)^{-1}\lambda\sigma_J\sigma_J^\transp(\sigma_J\sigma_J^\transp)^{-1}\bigr\rVert \\
		&\leq C_Q\lVert(\sigma_J\sigma_J^\transp)^{-1}\rVert=:C_P.
	\end{aligned}
\end{equation*}
The conditional mean $\muhat{J}{}$ can be written as
\begin{equation}\label{eq:representation_of_muhat_D}
	\begin{aligned}
		\muhat{J}{t} &= \int_0^t \alpha(\delta-\muhat{J}{s})\,\rmd s + \int_0^t\gam{J}{s}(\sigma_R\sigma_R^\transp)^{-1}(\rmd R_s-\muhat{J}{s}\,\rmd s) \\
		&\quad+\int_0^t \gam{J}{s}(\sigma_J\sigma_J^\transp)^{-1}(\rmd J_s-\muhat{J}{s}\,\rmd s).
	\end{aligned}
\end{equation}
Note that
\begin{align*}
	\rmd R_s-\muhat{J}{s}\,\rmd s &= \sigma_R\,\rmd V^Z_s+(\muhat{Z,\lambda}{s}-\muhat{J}{s})\,\rmd s,\\
	\rmd J_s-\muhat{J}{s}\,\rmd s &= \sigma_J\,\rmd W^J_s+(\mu_s-\muhat{J}{s})\,\rmd s.
\end{align*}
This yields the representation $\muhat{Z,\lambda}{t}-\muhat{J}{t} = A^\lambda_t+B^\lambda_t+C^\lambda_t+D^\lambda_t+E^\lambda_t$, where
\begingroup
\allowdisplaybreaks
	\begin{align}
		A^\lambda_t &= -\alpha\int_0^t (\muhat{Z,\lambda}{s}-\muhat{J}{s})\,\rmd s, \\
		B^\lambda_t &= \int_0^t (\gam{Z,\lambda}{s}-\gam{J}{s})(\sigma_R\sigma_R^\transp)^{-1}\sigma_R\,\rmd V^Z_s, \\
		C^\lambda_t &= \int_0^t \gam{J}{s}(\sigma_R\sigma_R^\transp)^{-1}(\muhat{J}{s}-\muhat{Z,\lambda}{s})\,\rmd s, \\
		D^\lambda_t &= \sum_{k=1}^{N_t} P^\lambda_k\sigma_J\int_{\frac{k-1}{\lambda}}^{\frac{k}{\lambda}}\rmd W^J_s - \int_0^t \gam{J}{s}(\sigma_J\sigma_J^\transp)^{-1}\sigma_J\,\rmd W^J_s, \\
		E^\lambda_t &= \sum_{k=1}^{N_t} \frac{1}{\lambda}P^\lambda_k\bigl(\mu_{T_k}-\muhat{Z,\lambda}{T_k-}\bigr) - \int_0^t \gam{J}{s}(\sigma_J\sigma_J^\transp)^{-1}\bigl(\mu_s-\muhat{J}{s}\bigr)\,\rmd s.
	\end{align}
\endgroup
Hence we have
\[ v^\lambda_t \leq 5^{p-1}\E\Bigl[\lVert A^\lambda_t\rVert^p+\lVert B^\lambda_t\rVert^p+\lVert C^\lambda_t\rVert^p+\lVert D^\lambda_t\rVert^p+\lVert E^\lambda_t\rVert^p\Bigr], \]
and it suffices to find upper bounds for the single summands on the right-hand side.

\paragraph{Estimation of stochastic integrals.}
As a preliminary step, we deduce upper bounds for the $p$-th moments of certain stochastic integrals w.r.t.\ $W\in\{V^Z,W^J\}$. Let
\( G_t=\int_0^t f^N_s\,\rmd W_s, \)
where $f^N$ is a matrix-valued integrand measurable with respect to $\calF^N_t:=\sigma(N_u, u\leq t)$. Then, $G_t$ conditional on $\calF^N_t$ is Gaussian with $\E[G_t\,|\,\calF^N_t]=0$. By Rosi\'{n}ski and Suchanecki~\cite[Lem.~2.1]{rosinski_suchanecki_1980} there is a constant $C_p>0$ such that
\[ \E\bigl[\lVert G_t\rVert^p\,|\,\calF^N_t\bigr] \leq C_p \E\bigl[\lVert G_t\rVert^2\,|\,\calF^N_t\bigr]^\frac{p}{2}. \]
The multivariate version of It\^{o}'s isometry from Lemma~\ref{lem:ito_isometry_multidimensional} yields
\[ \E\bigl[\lVert G_t\rVert^2\,|\,\calF^N_t\bigr] \leq C_\textrm{norm} \E\biggl[\int_0^t \lVert f^N_s\rVert^2\,\rmd s \,\bigg|\, \calF^N_t\biggr] = C_\textrm{norm}\int_0^t \lVert f^N_s\rVert^2\,\rmd s. \]
By putting these inequalities together we get
\begin{equation}\label{eq:preliminary_stochastic_integrals}
	\begin{aligned}
		\E\bigl[\lVert G_t\rVert^p\bigr]
		&= \E\Bigl[\E\bigl[\lVert G_t\rVert^p\,|\,\calF^N_t\bigr]\Bigr]
		\leq C_pC_\textrm{norm}^{p/2}\E\biggl[\biggl(\int_0^t \lVert f^N_s\rVert^2\,\rmd s\biggr)^{p/2}\biggr] \\
		&\leq C_pC_\textrm{norm}^{p/2}\E\biggl[t^{\frac{p-2}{2}}\int_0^t \lVert f^N_s\rVert^p\,\rmd s\biggr]
		=: \overline{C}_p\E\biggl[t^{\frac{p-2}{2}}\int_0^t \lVert f^N_s\rVert^p\,\rmd s\biggr].
	\end{aligned}
\end{equation}

\paragraph{Estimate for $\boldsymbol{A^\lambda}$.}
By using H\"{o}lder's inequality we have
\begin{equation}\label{eq:estimation_A}
	\begin{aligned}
		\E\Bigl[\bigl\lVert A^\lambda_t\bigr\rVert^p\Bigr]
		&\leq \lVert\alpha\rVert^p t^{p-1} \int_0^t \E\bigl[\lVert\muhat{Z,\lambda}{s}-\muhat{J}{s}\rVert^p\bigr]\,\rmd s \\
		&\leq \lVert\alpha\rVert^p T^{p-1} \int_0^t v^\lambda_s\,\rmd s =: C_A\int_0^t v^\lambda_s\,\rmd s.
	\end{aligned}
\end{equation}

\paragraph{Estimate for $\boldsymbol{B^\lambda}$.}
For the summand $B^\lambda_t$ we use~\eqref{eq:preliminary_stochastic_integrals} as well as Theorem~\ref{thm:q_C_lambda_goes_to_q_D} to get
\begin{equation}\label{eq:estimation_B}
	\begin{aligned}
		\E\Bigl[\bigl\lVert B^\lambda_t\bigr\rVert^p\Bigr]
		&\leq \overline{C}_p \E\biggl[t^{\frac{p-2}{2}}\int_0^t \lVert (\gam{Z,\lambda}{s}-\gam{J}{s})(\sigma_R\sigma_R^\transp)^{-1}\sigma_R\rVert^p\,\rmd s\biggr] \\
		&\leq \overline{C}_p T^{\frac{p-2}{2}}\lVert(\sigma_R\sigma_R^\transp)^{-1}\sigma_R\rVert^p\int_0^t \E\bigl[\lVert \gam{Z,\lambda}{s}-\gam{J}{s} \rVert^p\bigr]\,\rmd s \\
		&\leq \overline{C}_p T^{\frac{p}{2}}\lVert(\sigma_R\sigma_R^\transp)^{-1}\sigma_R\rVert^p \frac{\widetilde{K}_{Q,p}}{\lambda} =: \frac{C_B}{\lambda}.
	\end{aligned}
\end{equation}

\paragraph{Estimate for $\boldsymbol{C^\lambda}$.}
For the summand $C^\lambda_t$ we can argue similarly as for $A^\lambda_t$ and get
\begin{equation}\label{eq:estimation_C}
	\begin{aligned}
		\E\Bigl[\bigl\lVert C^\lambda_t\bigr\rVert^p\Bigr]
		&\leq t^{p-1} \int_0^t \lVert\gam{J}{s}(\sigma_R\sigma_R^\transp)^{-1}\rVert^p \E\bigl[\lVert\muhat{Z,\lambda}{s}-\muhat{J}{s}\rVert^p\bigr]\,\rmd s \\
		&\leq C_{\gam{}{}}^p\lVert(\sigma_R\sigma_R^\transp)^{-1}\rVert^p T^{p-1} \int_0^t v^\lambda_s\,\rmd s =: C_C\int_0^t v^\lambda_s\,\rmd s.
	\end{aligned}
\end{equation}

\paragraph{Estimate for $\boldsymbol{D^\lambda}$.}
The estimation of $D^\lambda_t$ is more involved. We can write
\begin{equation*}
	\begin{aligned}
		D^\lambda_t
		&= \int_0^{\frac{N_t}{\lambda}} H^\lambda_s\sigma_J\,\rmd W^J_s - \int_0^t \gam{J}{s}(\sigma_J\sigma_J^\transp)^{-1}\sigma_J\,\rmd W^J_s,
	\end{aligned}
\end{equation*}
where $H^\lambda_s=P^\lambda_k$ for $s\in[\frac{k-1}{\lambda},\frac{k}{\lambda})$.
Note that the two stochastic integrals do not align. We distinguish different cases by means of the random variable $n_t:=\min\{N_t,\lambda t\}$. This leads to the representation of $D^\lambda_t$ as $D^{1,\lambda}_t+D^{2,\lambda}_t+D^{3,\lambda}_t$, where
\begingroup
\allowdisplaybreaks
	\begin{align}
		D^{1,\lambda}_t &= \int_0^{\frac{n_t}{\lambda}} \bigl(H^\lambda_s-\gam{J}{s}(\sigma_J\sigma_J^\transp)^{-1}\bigr)\sigma_J\,\rmd W^J_s,\\
		D^{2,\lambda}_t &= \mathbbm{1}_{\{ N_t>\lambda t \}}\int_{t}^{\frac{N_t}{\lambda}}H^\lambda_s\sigma_J\,\rmd W^J_s,\\
		D^{3,\lambda}_t &= -\mathbbm{1}_{\{ N_t<\lambda t \}}\int_{\frac{N_t}{\lambda}}^t \gam{J}{s}(\sigma_J\sigma_J^\transp)^{-1}\sigma_J\,\rmd W^J_s.
	\end{align}
\endgroup
For the first term due to~\eqref{eq:preliminary_stochastic_integrals} it holds
\begin{equation}\label{eq:D_1_lambda_in_integral_form}
	\begin{aligned}
		\E\Bigl[\bigl\lVert D^{1,\lambda}_t\bigr\rVert^p\Bigr]
		&\leq \overline{C}_p \E\biggl[ t^{\frac{p-2}{2}}\int_0^t \bigl\lVert\mathbbm{1}_{\{s\leq\frac{n_t}{\lambda}\}}\bigl(H^\lambda_s-\gam{J}{s}(\sigma_J\sigma_J^\transp)^{-1}\bigr)\sigma_J\bigr\rVert^p\,\rmd s\biggr] \\
		&\leq \overline{C}_p T^{\frac{p-2}{2}}\bigl\lVert\sigma_J\bigr\rVert^p \E\biggl[ \int_0^{\frac{n_t}{\lambda}} \bigl\lVert H^\lambda_s-\gam{J}{s}(\sigma_J\sigma_J^\transp)^{-1}\bigr\rVert^p\,\rmd s\biggr].
	\end{aligned}
\end{equation}
Let $k\leq n_t$ and $s\in[\frac{k-1}{\lambda},\frac{k}{\lambda})$. Then
\begin{equation*}
	\begin{aligned}
		H^\lambda_s-\gam{J}{s}(\sigma_J\sigma_J^\transp)^{-1} &= \bigl( \gam{Z,\lambda}{T_k-}(\gam{Z,\lambda}{T_k-}+\lambda\sigma_J\sigma_J^\transp)^{-1}\lambda\sigma_J\sigma_J^\transp-\gam{J}{s} \bigr)(\sigma_J\sigma_J^\transp)^{-1}.
	\end{aligned}
\end{equation*}
Hence, we can deduce that there exists a constant $\overline{C}>0$ with
\begin{equation*}
	\begin{aligned}
		\bigl\lVert H^\lambda_s-\gam{J}{s}&(\sigma_J\sigma_J^\transp)^{-1}\bigr\rVert^p \leq \bigl\lVert(\sigma_J\sigma_J^\transp)^{-1}\bigr\rVert^p \bigl\lVert \gam{Z,\lambda}{T_k-}(\gam{Z,\lambda}{T_k-}+\lambda\sigma_J\sigma_J^\transp)^{-1}\lambda\sigma_J\sigma_J^\transp-\gam{J}{s} \bigr\rVert^p \\
		&\leq 3^{p-1}\bigl\lVert(\sigma_J\sigma_J^\transp)^{-1}\bigr\rVert^p \Bigl(\lVert\gam{J}{s}-\gam{J}{T_k}\rVert^p+\lVert\gam{J}{T_k}-\gam{Z,\lambda}{T_k-}\rVert^p+\frac{\overline{C}^p}{\lambda^p}\Bigr)
	\end{aligned}
\end{equation*}
by means of Lemma~\ref{lem:estimation_lemma}. Since $\gam{J}{s}$ is differentiable in $s$ with bounded derivative we deduce that $\lVert\gam{J}{s}-\gam{J}{T_k}\rVert^p\leq \widetilde{C}_{\gam{}{}}^p|T_k-s|^p$.
Using the moment generating function of $T_k\sim\mathrm{Erl}(k,\lambda)$ we can show $\E[|T_k-s|^p]\leq C_\textrm{Erl}\lambda^{-\frac{1}{2}}$ for a constant $C_\textrm{Erl}>0$ and all $\lambda\geq 1$. Using also Theorem~\ref{thm:q_C_lambda_goes_to_q_D} and plugging back into~\eqref{eq:D_1_lambda_in_integral_form} this implies
\begin{equation}\label{eq:estimation_D_1}
	\E\Bigl[\bigl\lVert D^{1,\lambda}_t\bigr\rVert^p\Bigr] \leq 3^{p-1}\overline{C}_p T^{\frac{p}{2}}\bigl\lVert\sigma_J\bigr\rVert^p\bigl\lVert(\sigma_J\sigma_J^\transp)^{-1}\bigr\rVert^p  \Bigl( \frac{\widetilde{C}_{\gam{}{}}^pC_\textrm{Erl}}{\sqrt{\lambda}} + \frac{\widetilde{K}_{Q,p}}{\lambda}+\frac{\overline{C}^p}{\lambda^p}\Bigr) \leq \frac{C_{D,1}}{\sqrt{\lambda}}
\end{equation}
for all $\lambda\geq 1$, where $C_{D,1}>0$ is a constant.
Next, we consider $D^{2,\lambda}_t$, where~\eqref{eq:preliminary_stochastic_integrals} yields
\begin{equation}\label{eq:estimation_D_2}
	\begin{aligned}
		\E\Bigl[\bigl\lVert D^{2,\lambda}_t\bigr\rVert^p\Bigr]
		&\leq \overline{C}_p\E\biggl[\Bigl(\frac{N_t}{\lambda}-t\Bigr)^{\frac{p-2}{2}} \int_{t}^{\frac{N_t}{\lambda}}\bigl\lVert\mathbbm{1}_{\{ N_t>\lambda t \}}H^\lambda_s\sigma_J\bigr\rVert^p\,\rmd s\biggr] \\
		&\leq \overline{C}_pC_P^p\lVert\sigma_J\rVert^p \E\biggl[\mathbbm{1}_{\{ N_t>\lambda t \}} \Bigl(\frac{N_t}{\lambda}-t\Bigr)^{\frac{p}{2}}\biggr] \\
		&\leq \overline{C}_pC_P^p\lVert\sigma_J\rVert^p\lambda^{-\frac{p}{2}} \E\Bigl[|N_t-\lambda t|^{\frac{p}{2}}\Bigr]
		\leq \frac{C_{D,2}}{\sqrt{\lambda}}.
	\end{aligned}
\end{equation}
For the last inequality note that using the moment generating function of $N_t\sim\mathrm{Poi}(\lambda t)$ it can be shown for any $r\geq 1$ that $\E[|N_t-\lambda t|^r]\leq C_\textrm{Poi}(\lambda t)^{r-\frac{1}{2}}$ for all $\lambda\geq 1$ and a constant $C_\textrm{Poi}>0$.
For $D^{3,\lambda}_t$ the estimation works similarly. By using~\eqref{eq:preliminary_stochastic_integrals} we obtain
\begin{equation}\label{eq:estimation_D_3}
	\begin{aligned}
		\E\Bigl[\bigl\lVert D^{3,\lambda}_t\bigr\rVert^p\Bigr]
		&\leq \overline{C}_p \E\biggl[\Bigl(t-\frac{N_t}{\lambda}\Bigr)^{\frac{p-2}{2}} \int_{\frac{N_t}{\lambda}}^t \bigl\lVert\mathbbm{1}_{\{ N_t<\lambda t \}}\gam{J}{s}(\sigma_J\sigma_J^\transp)^{-1}\sigma_J\bigr\rVert^p\,\rmd s \biggr] \\
		&\leq \overline{C}_pC_{\gam{}{}}^p\bigl\lVert(\sigma_J\sigma_J^\transp)^{-1}\sigma_J\bigr\rVert^p\E\biggl[\mathbbm{1}_{\{ N_t<\lambda t \}}\Bigl(t-\frac{N_t}{\lambda}\Bigr)^\frac{p}{2}\biggr] \\
		&\leq \overline{C}_pC_{\gam{}{}}^p\bigl\lVert(\sigma_J\sigma_J^\transp)^{-1}\sigma_J\bigr\rVert^p\lambda^{-\frac{p}{2}}\E\Bigl[|\lambda t-N_t|^\frac{p}{2}\Bigr] \leq \frac{C_{D,3}}{\sqrt{\lambda}}.
	\end{aligned}
\end{equation}
Combining~\eqref{eq:estimation_D_1}, \eqref{eq:estimation_D_2} and~\eqref{eq:estimation_D_3}, for $C_D=3^{p-1}(C_{D,1}+C_{D,2}+C_{D,3})$ and all $\lambda\geq 1$ it holds
\begin{equation}\label{eq:estimation_D}
	\begin{aligned}
		\E\Bigl[\bigl\lVert D^\lambda_t\bigr\rVert^p\Bigr]
		&\leq 3^{p-1}\Bigl(\E\Bigl[\bigl\lVert D^{1,\lambda}_t\bigr\rVert^p\Bigr]+\E\Bigl[\bigl\lVert D^{2,\lambda}_t\bigr\rVert^p\Bigr]+\E\Bigl[\bigl\lVert D^{3,\lambda}_t\bigr\rVert^p\Bigr]\Bigr)
		\leq \frac{C_D}{\sqrt{\lambda}}.
	\end{aligned}
\end{equation}

\paragraph{Estimate for $\boldsymbol{E^\lambda}$.}
By the same approach as for $D^\lambda_t$ we find $C_{E,1},C_{E,2}>0$ such that for all $\lambda\geq 1$ it holds
\begin{equation}\label{eq:estimation_E}
	\E\Bigl[\bigl\lVert E^\lambda_t\bigr\rVert^p\Bigr] \leq C_{E,1}\int_0^t v^\lambda_s\,\rmd s + \frac{C_{E,2}}{\sqrt{\lambda}}.
\end{equation}

\paragraph{Conclusion with Gronwall's Lemma.}
The upper bounds in~\eqref{eq:estimation_A}, \eqref{eq:estimation_B}, \eqref{eq:estimation_C}, \eqref{eq:estimation_D}, \eqref{eq:estimation_E} now imply that for all $\lambda\geq 1$ it holds
\begin{equation*}
	\begin{aligned}
		v^\lambda_t &\leq 5^{p-1}(C_A+C_C+C_{E,1})\int_0^t v^\lambda_s\,\rmd s+5^{p-1}(C_B+C_D+C_{E,2})\frac{1}{\sqrt{\lambda}}.
	\end{aligned}
\end{equation*}
Now Gronwall's Lemma, see Lemma~\ref{lem:gronwall}, implies
\begin{equation*}
	\begin{aligned}
		v^\lambda_t &\leq 5^{p-1}(C_B+C_D+C_{E,2})\rme^{5^{p-1}(C_A+C_C+C_{E,1})T}\frac{1}{\sqrt{\lambda}} =: \frac{\widetilde{K}_{m,p}}{\sqrt{\lambda}}.
	\end{aligned}
\end{equation*}
This proves the claim for $p\geq 2$. For $p<2$ we obtain
\[ \E\Bigl[\bigl\lVert \muhat{Z,\lambda}{t}-\muhat{J}{t}\bigr\rVert^p\Bigr] \leq \E\Bigl[\bigl\lVert\muhat{Z,\lambda}{t}-\muhat{J}{t}\bigr\rVert^2\Bigr]^\frac{p}{2} \leq \Bigl(\frac{\widetilde{K}_{m,2}}{\sqrt{\lambda}}\Bigr)^\frac{p}{2}=\frac{\widetilde{K}_{m,p}}{\lambda^{p/4}} \]
from Lyapunov's inequality.\qed

\bibliographystyle{dissertation_style}

\end{document}